\PassOptionsToPackage{unicode}{hyperref}
\PassOptionsToPackage{hyphens}{url}
\PassOptionsToPackage{dvipsnames,svgnames,x11names}{xcolor}
\documentclass[12pt]{article}

\usepackage{setspace}
\usepackage{iftex}
\ifPDFTeX
  \usepackage[T1]{fontenc}
  \usepackage[utf8]{inputenc}
  \usepackage{textcomp} 
\else 
  \usepackage{unicode-math}
  \defaultfontfeatures{Scale=MatchLowercase}
  \defaultfontfeatures[\rmfamily]{Ligatures=TeX,Scale=1}
\fi
\usepackage{lmodern}
\ifPDFTeX\else  
\fi
\IfFileExists{upquote.sty}{\usepackage{upquote}}{}
\IfFileExists{microtype.sty}{
  \usepackage[]{microtype}
  \UseMicrotypeSet[protrusion]{basicmath} 
}{}
\makeatletter
\@ifundefined{KOMAClassName}{
  \IfFileExists{parskip.sty}{%
    \usepackage{parskip}
  }{
    \setlength{\parindent}{0pt}
    \setlength{\parskip}{6pt plus 2pt minus 1pt}}
}{
  \KOMAoptions{parskip=half}}
\makeatother
\usepackage{xcolor}
\makeatletter
\ifx\paragraph\undefined\else
  \let\oldparagraph\paragraph
  \renewcommand{\paragraph}{
    \@ifstar
      \xxxParagraphStar
      \xxxParagraphNoStar
  }
  \newcommand{\xxxParagraphStar}[1]{\oldparagraph*{#1}\mbox{}}
  \newcommand{\xxxParagraphNoStar}[1]{\oldparagraph{#1}\mbox{}}
\fi
\ifx\subparagraph\undefined\else
  \let\oldsubparagraph\subparagraph
  \renewcommand{\subparagraph}{
    \@ifstar
      \xxxSubParagraphStar
      \xxxSubParagraphNoStar
  }
  \newcommand{\xxxSubParagraphStar}[1]{\oldsubparagraph*{#1}\mbox{}}
  \newcommand{\xxxSubParagraphNoStar}[1]{\oldsubparagraph{#1}\mbox{}}
\fi
\makeatother

\usepackage{longtable,booktabs,array}
\usepackage{calc} 
\usepackage{etoolbox}
\makeatletter
\patchcmd\longtable{\par}{\if@noskipsec\mbox{}\fi\par}{}{}
\makeatother
\IfFileExists{footnotehyper.sty}{\usepackage{footnotehyper}}{\usepackage{footnote}}
\makesavenoteenv{longtable}
\usepackage{graphicx}
\makeatletter
\def\maxwidth{\ifdim\Gin@nat@width>\linewidth\linewidth\else\Gin@nat@width\fi}
\def\maxheight{\ifdim\Gin@nat@height>\textheight\textheight\else\Gin@nat@height\fi}
\makeatother
\setkeys{Gin}{width=\maxwidth,height=\maxheight,keepaspectratio}
\makeatletter
\def\fps@figure{htbp}
\makeatother

\makeatletter
\@ifpackageloaded{caption}{}{\usepackage{caption}}
\AtBeginDocument{%
\ifdefined\contentsname
  \renewcommand*\contentsname{Table of contents}
\else
  \newcommand\contentsname{Table of contents}
\fi
\ifdefined\listfigurename
  \renewcommand*\listfigurename{List of Figures}
\else
  \newcommand\listfigurename{List of Figures}
\fi
\ifdefined\listtablename
  \renewcommand*\listtablename{List of Tables}
\else
  \newcommand\listtablename{List of Tables}
\fi
\ifdefined\figurename
  \renewcommand*\figurename{Figure}
\else
  \newcommand\figurename{Figure}
\fi
\ifdefined\tablename
  \renewcommand*\tablename{Table}
\else
  \newcommand\tablename{Table}
\fi
}
\@ifpackageloaded{float}{}{\usepackage{float}}
\floatstyle{ruled}
\@ifundefined{c@chapter}{\newfloat{codelisting}{h}{lop}}{\newfloat{codelisting}{h}{lop}[chapter]}
\floatname{codelisting}{Listing}

\makeatother
\makeatletter
\@ifpackageloaded{caption}{}{\usepackage{caption}}
\@ifpackageloaded{subcaption}{}{\usepackage{subcaption}}
\makeatother

\ifLuaTeX
  \usepackage{selnolig}  
\fi
\usepackage[]{natbib}
\usepackage{bookmark}

\IfFileExists{xurl.sty}{\usepackage{xurl}}{} 
\hypersetup{
  pdftitle={Quantifying uncertainty and stability among highly correlated predictors: a subspace perspective},
  pdfauthor={Xiaozhu Zhang; Jacob Bien; Armeen Taeb },
  pdfkeywords={multicollinearity, multiple testing, stability selection, variable selection},
  colorlinks=true,
  linkcolor={blue},
  filecolor={Maroon},
  citecolor={Blue},
  urlcolor={Blue},
  pdfcreator={LaTeX via pandoc}}

\newcommand{\anon}{1}

\newtheorem{theorem}{Theorem}
\newtheorem{lemma}[theorem]{Lemma}
\newtheorem{proposition}[theorem]{Proposition}
\newtheorem{corollary}[theorem]{Corollary}
\newtheorem{definition}{Definition}

\newtheorem{assumption}{Assumption}
\newtheorem{remark}{Remark}

\newenvironment{proof}{\par\noindent\textit{Proof.}}

\usepackage{sectsty}
\sectionfont{\fontsize{14}{14}\selectfont}
\subsectionfont{\fontsize{13}{13}\selectfont}

\newcommand{\tr}{{\rm trace}}
\newcommand{\spa}{{\rm span}}
\newcommand{\st}{{\rm s.t.}\quad}
\newcommand{\supp}{{\rm supp}}
\newcommand{\rank}{{\rm rank}}
\renewcommand{\P}{\mathcal{P}}

\usepackage{multirow}
\usepackage{mathrsfs} 
\usepackage{algorithm}
\usepackage{algorithmic}
\usepackage{amsfonts, amssymb}
\usepackage{svg}
\usepackage{enumitem}
\usepackage{amsmath}
\usepackage{mathabx}
\usepackage{dsfont}
\usepackage{placeins}

\DeclareMathOperator*{\argmax}{argmax}

\usepackage{xr}

\date{}

\begin{document}

\def\spacingset#1{\renewcommand{\baselinestretch}%
{#1}\small\normalsize} \spacingset{1}


\if1\anon
{
  \title{\bf Quantifying uncertainty and stability among highly correlated predictors: a subspace perspective}
  \author{Xiaozhu Zhang\thanks{Corresponding author. Email: \texttt{xzzhang@uw.edu}.}\hspace{.2cm}\\
    Department of Statistics, University of Washington \vspace{.2cm}\\
    Jacob Bien \\
    Department of Data Sciences and Operations, University of Southern California \vspace{.2cm} \\
    Armeen Taeb \\
    Department of Statistics, University of Washington
    }
  \maketitle
} \fi

\if0\anon
{
  \bigskip
  \bigskip
  \bigskip
  \begin{center}
    {\LARGE\bf Title}
\end{center}
  \medskip
} \fi

\bigskip
\begin{abstract}
We study the problem of linear feature selection when features are highly correlated. Such settings pose two fundamental challenges. First, how should model similarity be defined? Simply counting features in common can be misleading: two models may share no features, yet highly correlated features can make the two models very similar in terms of predictive ability. Second, how can feature stability be assessed across runs of a variable selection method? High correlation can yield very different feature sets, so counting how often a feature is selected may label most features as unstable, and selecting stable features would result in models that are too small with poor predictive performance. In essence, these issues arise because existing notions of similarity and stability are ``discrete'' in nature. To overcome these challenges, we propose a novel framework based on feature subspaces---the subspaces spanned by selected columns of the feature matrix. This new perspective leads to ``continuous'' measures of similarity and stability, as well as false positive error, all of which are defined in terms of ``closeness'' of feature subspaces. Our measures naturally account for feature correlation and reduce to existing discrete notions when features are uncorrelated. To obtain stable models, we propose and theoretically analyze a subspace-based generalization of stability selection \citep{meinshausen2010stability, taeb2020false}, which combines a discrete model search with a continuous subspace-based assessment of stability. On synthetic and real gene expression data, our method improves on existing stability-based approaches by (i) producing multiple stable models that capture feature interchangeability, and (ii) generating larger models with better predictive performance. Our method is implemented in the R package \texttt{substab}. 

\end{abstract}

\noindent%
\vfill

\newpage
\spacingset{1.8} 

\section{Introduction}
\label{sec:introduction}
Variable selection is a perennial problem in data science. Given a large set of features, the objective is to identify a small subset that is predictive of a response variable. Despite the massive amount of work in this area, variable selection in highly correlated settings remains a significant challenge.

One challenge is assessing the similarity between two feature sets (which we call models) $S_1$ and $S_2$, and in particular, between an estimated model $\widehat{S}$ and the true model $S^\star$. The most commonly used notion of similarity is quantified by the number of shared features $|S_1 \cap S_2|$, which amounts to the number of true positives when comparing $\widehat{S}$ and $S^\star$. While intuitive, this metric can be misleading when features are highly correlated. For example, consider a gene expression dataset (analyzed in Section~\ref{sec:real_data}) containing many genes, from which we highlight four features $(X_1 = \textit{H3F3A}, X_2 = \textit{IL1RN}, X_3 = \textit{EIF5B}, X_4 =\textit{IAPP})$ that are highly associated with breast cancer prognosis~\citep{sotiriou2006gene}. Now consider two models, $S_1 = \{1,2\}$ and $S_2 = \{3,4\}$, which share no common features, so that $|S_1 \cap S_2| = 0$. However, features $X_1$ and $X_3$ are highly correlated (correlation coefficient $0.974$), as are $X_2$ and $X_4$ (correlation coefficient $0.994$). In fact, each pair of features is nearly identical up to a small perturbation. Therefore, despite having zero set overlap, the two models $S_1$ and $S_2$ exhibit very similar predictive behavior. This suggests the need for a more refined measure of similarity that accounts for feature correlation. Ideally, under such a metric, two nearly identical models in terms of predictive ability would attain a high similarity score.

A second challenge is measuring the stability of feature selection procedures. The importance of stability is highlighted in \citet{doi:10.1073/pnas.1901326117} and \citet{yu2013stability}, which uses the term to refer to when ``statistical conclusions are robust or stable
to appropriate perturbations to data.'' A widely used stability-based procedure in regression is stability selection \citep{meinshausen2010stability}. The basic idea is that instead of applying one’s favorite variable selection algorithm such as the Lasso to the whole data set, one applies it several times to random subsamples of the data. The proportion of subsamples where each feature is selected then represents the stability of that feature; the features can then be ranked according to their stability. However, returning to the previous gene expression example, the high correlation between $X_1$ and $X_3$ creates a fundamental difficulty for this procedure. Across subsamples, roughly half may select feature $X_1$ while the other half select $X_3$ \citep{faletto2022cluster}. This phenomenon, known as ``vote splitting'' \citep{shah2013variable}, can result in neither $X_1$ nor $X_3$ receiving a high stability score, despite both being strongly associated with the response. Ideally, a reasonable stability metric should assign large values to both $X_1$ and $X_3$.

Building on stability, a third challenge is aggregating stable features into an interpretable model output. 
Assume that the second challenge has been addressed and all four features in the gene expression example are deemed stable. Even in this case, these features should not be combined into a single model, as including both $X_1$ and $X_3$, or both $X_2$ and $X_4$, would introduce redundancy. Instead, it is more natural to consider multiple stable models, such as $(X_1,X_2)$, $(X_1,X_4)$, $(X_2,X_3)$, and $(X_3,X_4)$.
Outputting all these models provides a more faithful representation of uncertainty over the variable selection procedure. This challenge motivates the following questions: (i) how can we find stable models (rather than merely stable features)? and (ii) how can we efficiently find multiple stable models when the model space is large?

Returning to the gene expression dataset for breast cancer prognosis~\citep{sotiriou2006gene}, we find that stability selection \citep{meinshausen2010stability} (i) produces a single stable model even when high correlation implies significant feature interchangeability and (ii) selects few stable features, resulting in poor predictive performance. Cluster-based stability selection~\citep{faletto2022cluster} again returns a single stable model with somewhat improved power and prediction performance. Crucially, the performance and interpretation of cluster stability selection crucially depend on a pre-specified clustering of features. In contrast, without any prior knowledge of the correlation structure among the features, we aim to identify multiple stable models with good power, while also providing an understanding of the relationships among the resulting models.

In essence, the aforementioned challenges arise because existing notions of similarity and stability are discrete in nature. Continuous measures of similarity and stability may result in better assessment of model selection when features are highly correlated; this is the focus of the present paper.

Throughout, unless otherwise specified, we use the terminology ``highly correlated'' to mean that the predictors $X$ are (nearly) linearly dependent. The dependency structure we consider can be much more general and complex than simply a high pairwise correlation among some predictors.

\subsection{Our Contributions}
As our first contribution, we propose in Section \ref{sec:subspace} a subspace perspective for assessing model selection in highly correlated settings. Suppose $X$ is a data matrix of predictors where columns index features and rows index samples. Instead of quantifying similarity between models $S_1$ and $S_2$ in terms of commonality between the two sets of features, we quantify it in terms of the degree of alignment between \emph{feature subspaces} $\mathrm{col}(X_{S_j})$, $j=1,2$---the subspaces spanned by the columns of the matrix $X$ corresponding to these features. 
Building on this notion of similarity, we measure the amount of true positives and false positives based on the closeness of certain subspaces. Further, given models $\widehat{S}^{(\ell)} (\ell = 1,2,\dots,B)$ obtained from applying one's favorite variable selection procedure on $B$ subsamples of the data, we measure stability of a feature $X_j$ as the degree of alignment between the subspace $\mathrm{col}(X_{j})$ and the subspaces $\mathrm{col}(X_{\widehat{S}^{(\ell)}}), \ell=1,2,\dots,B$; this idea can also be generalized to find the stability of a set of features. These subspace-based metrics reduce to the classical notions of similarity, the numbers of true and false positives, and the standard measure of stability in \citep{meinshausen2010stability} when features are orthogonal, but they provide more meaningful assessments in highly correlated settings.

As our second contribution, we propose in Section \ref{sec:algorithm} the algorithm \emph{feature subspace stability selection} (FSSS), a procedure for enumerating stable models.  Similar to stability selection, as input, FSSS takes selected features $\widehat{S}^{(\ell)} (\ell = 1,2,\dots,B)$ after deploying one's favorite variable selection procedure on subsamples of the data. After mapping each subset to a feature subspace $\mathrm{col}(X_{\widehat{S}^{(\ell)}})$, FSSS uses a sequential procedure to obtain a collection of stable models, where each model maps to a feature subspace that is close to most of the estimated subspaces; the specific measure of distance here is based on the stability formalism described earlier.

In Section \ref{sec:theory}, we analyze the theoretical performance of FSSS. We focus on the setting where the features can be grouped into highly correlated clusters,
and describe generalizations to more complex dependency structures in the Appendix. We provide false positive error guarantees for models selected by FSSS, where the bound on the error depends on a certain notion of the quality of the base procedure. Under some assumptions on the base procedure, we show that FSSS has feature selection consistency: it returns all ``equivalent'' selection sets consisting of one feature from each signal cluster (a cluster containing a signal feature) and no feature from a non-signal cluster is selected.

Finally, in Section~\ref{sec:simulation}, we compare FSSS with several existing stability selection variants using both synthetic data and the gene expression dataset for breast cancer prognosis \citep{sotiriou2006gene}. Across a range of synthetic experiments, FSSS returns larger models with greater predictive power while maintaining a small false positive rate. On the gene expression data, FSSS performs comparably or better in terms of predictive accuracy and tends to select larger models. Furthermore, FSSS identifies multiple stable models among the gene expression features. We analyze the similarity between these models and show how such comparisons can yield additional insights into the interpretability of breast cancer prognosis.

As an illustration, using the synthetic data described in Appendix~\ref{sec:sup-figure1}, Figure~\ref{fig:stability} compares stability values from stability selection \citep{meinshausen2010stability}, its variant cluster stability selection \citep{faletto2022cluster,alexander2011stability}, and our proposed subspace framework. We observe a “vote-splitting” phenomenon in stability selection, in which correlated signals receive low stability values. Although cluster stability selection partially improves upon stability selection, it still fails to distinguish some signal variables from noise variables. In contrast, our method, FSSS, cleanly separates noise from non-noise features and consequently produces larger models with better predictive performance. When the stability threshold is set to $0.8$ for all methods, stability selection yields an empty model with a test mean squared error of 17.38; cluster stability selection yields a size-8 model with a test mean squared error of 9.20; and FSSS yields a size-12 model with a test mean squared error of 0.32. Moreover, whereas previous methods output a single model, our approach returns a collection of stable models that capture feature interchangeability. See Appendix~\ref{sec:sup-figure1} for additional analysis.



We implement our methods in the R package \texttt{substab}. Our package provides two primary functionalities: $(i)$ tools for computing similarity and stability metrics, and $(ii)$ subsampling and FSSS procedures, which take as input a feature matrix and a response variable and output collections of stable models.

\begin{figure}[!ht]
    \centering
    \includegraphics[width=0.8\textwidth]{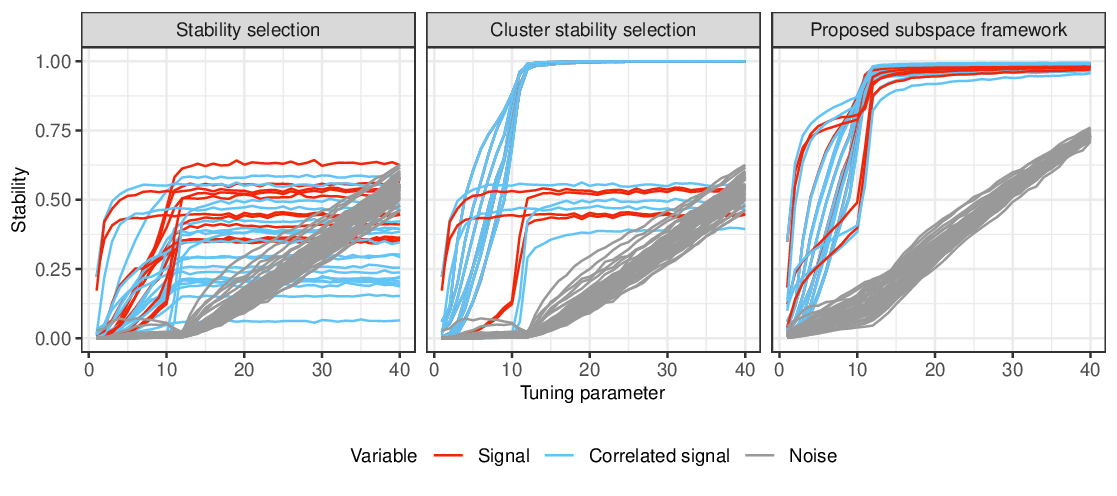}
    \caption{\small  Stability paths of the base procedure $\ell_0$-penalized regression on a synthetic dataset (see Appendix~\ref{sec:sup-figure1} for details). From left to right: stability values using stability selection, cluster stability selection, and our proposed subspace framework for different choice of tuning parameter in the base procedure. The ``signal'' refers to features directly generating the response. The ``correlated signal'' refers to features that are highly correlated with the signals. The ``noise'' refers to features independent of signals.}
    \label{fig:stability}
\end{figure}

\subsection{Related Work}\label{sec:relate-work}
In the last decade there has been some work devoted to tackling different aspects of the variable selection problem under high feature correlation. Regularization approaches such as the elastic net \citep{zou2005regularization} combine $\ell_1$ and $\ell_2$ penalties to handle correlated predictors. The resulting “grouping effect” encourages correlated variables to be selected together, each with shrunk coefficients; however, this often yields dense models that retain multiple redundant variables within a group, failing to induce within-group sparsity and limiting the interpretability of the resulting model.

Another line of work considers explicitly modeling correlation structure. For example, \citet{buhlmann2013correlated} propose a feature selection method designed for the setting in which there are clusters of highly correlated features; however, this work does not address questions of stability or the inadequacy of traditional measures of false positive errors when features are highly correlated.

\citet{g2013false} tackle this latter challenge by quantifying false positives as features in the selected model that can be explained by other features in the selected model. While in some highly correlated settings, this new metric provides more reasonable values than the standard notion, it typically gives smaller false positive errors than one may expect; see Appendix~\ref{sec:sup-g2013false}.

Cluster stability selection (CSS; \citealt{faletto2022cluster}, and \citealt{alexander2011stability}) is a recent method designed to improve stability selection. CSS uses a pre-defined cluster assignment of features and calculates the selection probability for each cluster. A cluster is considered selected by a subsample if any of its features are selected. CSS gathers votes from the features in each cluster, then ranks them based on the cluster's selection probability. When features are well-clustered (either orthogonal or highly correlated), CSS helps avoid the ``vote splitting'' issue in stability selection. However, CSS may perform poorly in more complex settings, as demonstrated in Figure~\ref{fig:stability}.


Approaches based on principal components analysis \citep{Bair01032006,Yu2006SupervisedPP} are widely used in highly correlated settings. These methods identify a few principal components of the feature matrix and use them to predict the response. However, since principal components are typically linear combinations of many features, they can be less interpretable than directly selecting a small number of features from the original set, which is the task considered in this paper. 

High feature correlation undermines the value of finding a unique ``best'' model and leads, as we do here, to thinking about a set of good models.  This connects our work to Breiman's notion of the {\em Rashomon effect} \citep{breiman2001statistical}, which acknowledges that there may be a ``multitude of different descriptions'' of models that all achieve close to the best performance.  A growing literature considers the Rashomon effect \citep{fisher2019all,dong2020exploring,zhong2023exploring}.  In fact, \citet{breiman2001statistical} discusses stability (or rather instability) in this context. Some recent work has explicitly considered stability in the context of this multitude of models \citep{kissel2024forward,adrian2024stabilizing}, although neither focuses on subspaces in the way that our work does.

Our work is largely inspired by \cite{taeb2020false}, who study model selection over subspaces. In particular, they propose a measure for false positive error in subspace selection and a stability-based procedure that controls this error. A key novelty of our paper is to adapt these ideas to formulate a new perspective on variable selection with highly correlated features, blending discrete model selection with continuous, subspace-based assessments of model quality, such as stability. This blending leads to several key differences from \cite{taeb2020false}. First, our approach targets discrete model selection over the feature-set space, whereas the previous work focuses on continuous selection over the space of subspaces. In our framework, the subspaces of interest are explicitly anchored to features in the data matrix, and FSSS searches over discrete models while leveraging these feature-induced subspaces. In contrast, the subspace stability algorithm in \cite{taeb2020false} operates directly over subspaces without reference to specific features. Second, our theoretical analysis substantially extends the results of \cite{taeb2020false}. By restricting subspaces to those arising from the data matrix, we can exploit its structure, leading to stronger theoretical guarantees and more interpretable results.

\subsection{Notation}
\label{sec:notations}
For a subspace $T \subseteq \mathbb{R}^n$, we denote the projection onto $T$ by $\mathcal{P}_T$. The subspace consisting of all vectors that are orthogonal to those in $T$ are denoted by $T^\perp$. For a matrix $M \in \mathbb{R}^{m \times m}$, we denote the span of the columns of $M$ (i.e. its column space) by $\mathrm{col}(M)$. We denote the singular values of $M$, in descending order by $\sigma_1,\sigma_2,\dots,\sigma_m$. For a set $S \subseteq \{1,2,\dots,p\}$, we denote $S^c$ as its complement. We denote $\mathrm{Corr}(a,b) = |a^\top{b}|/(\|a\|_{\ell_2}\|b\|_{\ell_2})$ for $a,b\in \mathbb{R}^n$.

\section{Why are subspaces useful?}\label{sec:subspace}
\emph{Setup}: Let $X \in \mathbb{R}^{n \times p}$ represent the data matrix of predictors where columns index the $p$ features and rows index the $n$ samples. Let $y \in \mathbb{R}^{n}$ encode observations of the response variable. Throughout, we consider linear regression models in which linear combinations of features are used to predict the response; we discuss extensions to generalized linear models in Appendix~\ref{sec:sup-glm}. We assume that the features are centered, so that the sum of entries in every column in $X$ is zero.

This section is organized as follows: Section \ref{sec:representing} gives a brief description of what it means to represent selection sets as subspaces; Section~\ref{sec:similarity} describes how a subspace perspective enables a natural measure of similarity between two models in highly correlated settings; Section~\ref{sec:tp-fp} extends the proposed notion of similarity to measuring the amount of true positives and false positives; and Section~\ref{sec:stability} describes a measure of stability using subspaces.

\subsection{Representing selection sets as subspaces}
\label{sec:representing}

Recall the gene expression example for breast cancer prognosis introduced in Section \ref{sec:introduction}, where we consider two gene expression models $(X_1 =\textit{H3F3A}, X_2 = \textit{IL1RN})$ and $(X_3 = \textit{EIF5B}, X_4 =\textit{IAPP})$. A common perspective represents the first model by the set $S_1 = \{1,2\}$ and the second by $S_2 = \{3,4\}$.  Since $S_1 \cap S_2 =\varnothing$, these two model structures would traditionally be viewed as having nothing in common. However, this conclusion is misleading for two reasons. First, due to high correlations, $X_1$ and $X_3$ are nearly identical up to small perturbation, and likewise for $X_2$ and $X_4$. Consequently, the two model structures share substantial similarity: linear combinations of $X_1$ and $X_2$ are very close to linear combinations of $X_3$ and $X_4$. Second, when similarity is measured by $|S_1 \cap S_2|$, it remains identically zero regardless of the correlation level. Intuitively, however, stronger correlations should imply greater similarity between the two model structures.  

The previous arguments motivate the following question: If not the selection set, what structure accounts for correlation and enables an appropriate comparison between models? Consider a model that uses a set of predictors $S \subseteq \{1,2,\dots,p\}$. Any linear prediction obtained using these predictors lies in the \emph{features subspace} $\mathrm{col}(X_S)$---the subspace spanned by the columns of $X$ corresponding to the features in $S$. Thus, instead of using $S$ to encode the structure of the model, it is arguably more natural to use the feature subspace $\mathrm{col}(X_S)$. Furthermore, to compare models $S_1$ and $S_2$, we can compare their associated subspaces $\mathrm{col}(X_{S_1})$ and $\mathrm{col}(X_{S_2})$.

Subspaces naturally account for feature correlations when comparing models. Consider again the gene expression example. When the correlations between $X_1$ and $X_3$, and between $X_2$ and $X_4$ are strong, the feature subspaces $\mathrm{col}(X_{S_1})$ and $\mathrm{col}(X_{S_2})$ become closely aligned, and the corresponding model structures can therefore be regarded as similar. As the correlation level increases, the alignment of the subspaces---and thus the similarity of the model structures---increases smoothly. In other words, the smooth representation via subspaces enables a smooth measure of proximity between model structures. 

\subsection{Assessing similarity via subspaces}
\label{sec:similarity}
Let $S_1$ and $S_2$ be two sets of features, both subsets of $\{1, 2, \dots,p\}$. We use feature subspaces $\mathrm{col}(X_{S_1})$ and $\mathrm{col}(X_{S_2})$ as a natural encoding of model structures. How do we assess the extent to which the two feature subspaces $\mathrm{col}(X_{S_1})$ and $\mathrm{col}(X_{S_2})$ are similar?

One approach is to measure the amount of similarity as the dimension of the intersection between $\mathrm{col}(X_{S_1})$ and $\mathrm{col}(X_{S_2})$. However, when both $X_{S_1}$ and $X_{S_2}$ are linearly independent, $\mathrm{dim}(\mathrm{col}(X_{S_1}) \cap \mathrm{col}(X_{S_2}))$ reduces to $|S_1 \cap S_2|$, which is, as described earlier, ill-suited for highly correlated settings. For example, two distinct but highly-correlated features are very similar while $|S_1 \cap S_2 |$ encodes zero similarity. 
In some sense, the preceding attempt fails because it is based on a sharp binary choice that declares each dimension as either captured or not. Consequently, we need a measure of similarity that varies smoothly with the underlying subspaces. \citet{taeb2020false} introduce a measure of similarity between two arbitrary subspaces. Restricting these subspaces to feature column spaces yields the following measure of similarity between two sets of features.
\begin{definition}[Similarity]\label{def:similar}
    Let $S_1, S_2 \subseteq\{1,2, \dots, p\}$ be two sets of features. Then their similarity is defined as
    $$
    \tau(S_1, S_2) := \tr(\P_{\mathrm{col}(X_{S_1})} \P_{\mathrm{col}(X_{S_2}) } ).
    $$
\end{definition}

The metric $\tau(S_1, S_2)$ provides a meaningful measure of similarity between the subspaces $\mathrm{col}(X_{S_1})$ and $\mathrm{col}(X_{S_2})$ for several reasons. First, the metric is equal to 
\vspace{-0.1in}
\begin{align}\label{eq:p-angles}
    \tr(\P_{\mathrm{col}(X_{S_1})} \P_{\mathrm{col}(X_{S_2}) } ) = \sum_{j = 1}^{\min\{|S_1|, |S_2|\} } \cos^2 \theta_j,
\vspace{-0.1in}
\end{align}
where $\theta_j$'s are the principal angles between $\mathrm{col}(X_{S_1})$ and $\mathrm{col}(X_{S_2})$ \citep{Bjoerck1971NumericalMF}. As a result, $\tau(S_1, S_2)$ measures the degree of alignment between then subspaces $\mathrm{col}(X_{S_1})$ and $\mathrm{col}(X_{S_2})$, in a manner that varies smoothly via change of principal angles. Second, the correlation structure among $X_{S_1}$ and $X_{S_2}$ is taken into account naturally. To see this, suppose $\{u_j\in\mathbb{R}^n \}_{j=1}^{\min \{|S_1|, |S_2|\} }$ and $\{v_j\in\mathbb{R}^n\}_{j=1}^{\min \{|S_1|, |S_2|\} }$ are canonical correlation analysis (CCA) vectors of $X_{S_1}$ and $X_{S_2}$, then a simple calculation shows that
$$
\cos^2\theta_j = (u_j^\top v_j)^2 = \sigma^2_j(\P_{\mathrm{col}(X_{S_1})} \P_{\mathrm{col}(X_{S_2})}), \quad \forall j \in\{1, 2,\dots, \min\{|S_1|, |S_2|\} \},
$$
where $\sigma_j(\cdot)$ computes the $j$-th singular value. Third, when the features are orthogonal, or when $S_1 \subseteq S_2$, the metric $\tau(S_1, S_2)$ reduces to the traditional measure of similarity $|S_1 \cap S_2|$. Finally, we have that $0 \leq \tau(S_1, S_2) \leq \min\{|S_1|, |S_2|\}$, which is in agreement with the intuition that the amount of similarity should not exceed the size of either model.

As shown in \eqref{eq:p-angles}, the metric $\tau(S_1, S_2)$ aggregates the full spectrum of principal angles. We present two additional variants that also operate on subspaces via principal angles:
(i) Normalized similarity 
\vspace{-0.1in}
\begin{equation}
\bar{\tau}(S_1, S_2):=\tau(S_1, S_2) /  \min\{|S_1|, |S_2|\} \in [0,1],
\label{eqn:normalized_similarity}
\end{equation} 
whose scale is independent of model size, making it more convenient than $\tau(S_1, S_2)$ for pairwise comparisons across multiple models. We adopt the convention $0/0 = 1$.
(ii) Cosine squared of the largest principal angle
\vspace{-0.1in}
\begin{equation}
\tilde{\tau}(S_1, S_2) :=\sigma_{ \max\{|S_1|, |S_2|\} }^2(\P_{\mathrm{col}(X_{S_1})} \P_{\mathrm{col}(X_{S_2})}) 
= \cos^2 \theta_{\max\{ |S_1|, |S_2| \} }
\in [0,1],\vspace{-0.1in}
\label{eqn:conservative_similarity}
\end{equation}
where $\sigma_{ \max\{|S_1|, |S_2|\}}(\cdot)$ computes the $\max\{|S_1|, |S_2|\}$-th singular value. This measure is more conservative than $\bar{\tau}(S_1, S_2)$ in terms of similarity, since
$
\tilde{\tau}(S_1, S_2) \leq \bar{\tau}(S_1, S_2)
$.
It degenerates to zero when $|S_1| \neq |S_2|$. Moreover, when features are perfectly orthogonal, $\tilde{\tau}(S_1, S_2)$ reduces to $\mathds{1}(S_1 = S_2)$. Although its conservative nature dismisses models of unequal sizes, this conservatism has strong implications for predictive similarity between $S_1$ and $S_2$. To see this, we observe that
\vspace{-0.1in}
\begin{align}\label{eq:worst-prediction}
     1 - \tilde{\tau}(S_1, S_2) = \sin^2 \theta_{\max\{|S_1|, |S_2|\}} = \sup_{\|y\|_2 = 1} \left\| \P_{\mathrm{col}(X_{S_1})} y - \P_{\mathrm{col}(X_{S_2})} y \right\|_2^2 =  \sup_{\|y\|_2 = 1} \left\| X_{S_1} \widehat{\beta}^y_{S_1} - X_{S_2} \widehat{\beta}^y_{S_2} \right\|_2^2 ,
\vspace{-0.1in}
\end{align}
where $\widehat{\beta}^y_{S_j}$, $j=1,2$ denotes the ordinary least squares coefficient vector obtained by regressing $y$ on the features in $S_j$; consequently, $\P_{\mathrm{col}(X_{S_j})} y = X_{S_j} \widehat{\beta}^y_{S_j}$ represents the best linear predictor of $y$ using the features in $S_{j}$. Moreover, the right-hand side of \eqref{eq:worst-prediction} quantifies the discrepancy between the linear prediction maps induced by $S_1$ and $S_2$, uniformly over $y$. As a result, a larger value of $\tilde{\tau}$ indicates that $S_1$ and $S_2$ yield similar predictive mappings, even under worst-case perturbation of the response variable. In this sense, $\tilde{\tau}$ captures the robustness of their predictive behavior across substantially different response-generating mechanisms.


\begin{remark}
The response variable $y$ does not enter the similarity metrics defined above for several reasons. First, the metric $\tau(S_1, S_2)$ serves as a subspace-based extension of the commonly used notion of similarity, and it reduces to $|S_1 \cap S_2|$ when the features are perfectly orthogonal. 
    Second, the metric $\tau(S_1, S_2)$ is invariant to perturbations and distributional shifts in the response variable. Its variant $\tilde{\tau}(S_1, S_2)$ further characterizes worst-case predictive similarity, as formalized in \eqref{eq:worst-prediction}.
    A large similarity implies that the models $S_1$ and $S_2$ can be used interchangeably for prediction on future perturbed data.
\end{remark}

\begin{remark}
\label{rem:y-similarity}
    In Appendix~\ref{sec:sup-y-similarity}, we present the similarity measure $\tau^y(S_1, S_2) := \|\P_{\mathrm{col}(X_{S_1})} y -\P_{\mathrm{col}(X_{S_2})}y \|_2^2 / \|y\|_2^2$ that explicitely incorporates $y$. This approach is likewise subspace-based and can be expressed as a linear combination of the squared cosines of the principal angles between $\mathrm{col}(X_{S_1})$ and $\mathrm{col}(X_{S_2})$. Unlike the previous measures, this stability measure does not account for possible future perturbations on the response variable.  
\end{remark}

\subsection{Assessing true and false positives via subspaces}
\label{sec:tp-fp}

Let $S^\star \subseteq \{1,2,\dots,p\}$ be the set of signal features that generate the response variable in the true data-generating mechanism, and $\widehat{S} \subseteq \{1,2,\dots,p\}$ be an estimated set of features. We assume that $X_{S^\star}$ consist of linearly independent columns; otherwise, the true model could have been reduced to a smaller set of features. The feature subspace $\mathrm{col}(X_{S^\star})$ represents the population model structure and the feature subspace $\mathrm{col}(X_{\widehat{S}})$ represents our estimated model structure or our ``discovery''.

\begin{definition}[True and false positives] \label{def:fd-power}
Let $S^\star \subseteq \{1,2,\dots,p\}$ be the true set of features and $\widehat{S} \subseteq \{1,2,\dots,p\}$ be an estimated set. Then, the amount of true and false positives are respectively: 
\vspace{-0.2in}
\begin{eqnarray}
\begin{aligned}
\mathrm{TP}(\widehat{S},S^\star) &:= \tau(\widehat{S}, S^\star)~~~~\& ~~~~~\mathrm{FP}(\widehat{S},S^\star) &:= |\widehat{S}|-\mathrm{TP}(\widehat{S},S^\star).
\label{eq:fd_td_def}
\end{aligned}
\end{eqnarray}
\vspace{-0.5in}
\end{definition}
The conventional ``discrete'' notion of true positives, given by the count $|\widehat{S}\cap S^\star|$, is replaced here by the ``continuous'' subspace-based measure $\tau(\widehat{S}, S^\star)$. This new measure computes the degree of alignment between the true subspace $\mathrm{col}(X_{S^\star})$ and the estimated subspace $\mathrm{col}(X_{\widehat{S}})$ with several desirable properties: $0 \leq \mathrm{TP}(\widehat{S},S^\star) \leq \min\{|\widehat{S}|,|S^\star|\}$, $\mathrm{TP}(\widehat{S},S^\star)= |\widehat{S}|$ if and only if $\widehat{S} \subseteq S^\star$,  $\mathrm{TP}(\widehat{S},S^\star)= |S^\star|$ if and only if $S^\star \subseteq \widehat{S}$, and $\mathrm{TP}(\widehat{S},S^\star) = |\widehat{S}\cap S^\star|$ when features are orthogonal. Similarly, the conventional ``discrete'' definition of false positives $|\widehat{S}\cap {S^\star}^c|$ is replaced by the ``continuous'' measure $\mathrm{FP}(\widehat{S},S^\star)$ which can be equivalently expressed as $\tr(\P_{\mathrm{col}( X_{\widehat{S}})}  \P_{ \mathrm{col}( X_{S^\star})^\perp})$. This measures the degree of alignment between the orthogonal subspace $\mathrm{col}(X_{S^\star})^\perp$ and the estimated subspace $\mathrm{col}(X_{\widehat{S}})$. 





\subsection{Assessing stability via subspaces}
\label{sec:stability}
The stability principle in model selection focuses on finding model features that remain consistent when data is randomly changed. Features that change a lot with small data changes are likely due to noise or artifacts and are considered unreliable. Given this, how can we assess stability when selecting features in regression?

A common method to assess stability is called stability selection \citep{meinshausen2010stability,shah2013variable}. It works by applying a variable selection procedure (e.g., Lasso) to subsamples of the data to obtain estimated sets $\{\widehat{S}^{(\ell)}\}_{\ell=1}^{B} \subseteq \{1,\dots,p\}$, where $B$ represents the number of datasets produced from subsampling. Then, the stability of a feature $X_j$ is measured as $\frac{1}{B}\sum_{\ell=1}^{B} \mathds{1}(j \in \widehat{S}^{(\ell)})$, the proportion of subsamples where the feature $j$ is selected. Features are ranked by stability, and the most stable ones are chosen for the final model. The model's stability is determined by the least stable feature.

The previous method for measuring stability does not work well in highly correlated settings. In the introduction, we illustrated this point with two examples, where we highlighted that: $i)$ High correlation can cause ``vote splitting'', where subsampling may choose a correlated non-signal feature instead of the true signal, reducing the signal feature's stability; $ii)$ Due to the vote-splitting phenomenon, a correlated non-signal feature may also have low stability, even though it can serve as a good substitute for the signal; and $iii)$ Methods like cluster stability selection can help reduce vote-splitting by grouping variables, but they may still give inaccurate stability values if the pre-assigned cluster assignment is wrong or if the relationships between features are more complex than simple correlations.


The failure of previous methods is again mainly due to the sharp binary choice that determines whether a feature is included in the selection set $\widehat{S}^{(\ell)}$, or how variables are grouped together. A subspace approach provides a smoother measure of stability that works better in settings with highly correlated variables and unknown correlation
structures. To this end, we build on the measure of \citet{taeb2020false}, which evaluates the stability of a subspace relative to estimated subspaces obtained through subsampling. By anchoring these subspaces to the feature data matrix, we translate this idea into a similarity measure defined directly for subsets of features.

\begin{definition}[Stability of a set $S$ via subspaces] 
\label{def:stability}
Given estimates $\{\widehat{S}^{(\ell)}\}_{\ell=1}^{B}$ from applying a variable selection procedure to $B$ subsamples of the data, the stability of a set $S \subseteq \{1,2,\dots,p\}$ is given by
\vspace{-0.1in}
\begin{equation*}\pi(S):= \sigma_{|S|}\left( \mathcal{P}_{\mathrm{col}(X_S)}\mathcal{P}_\mathrm{avg}\mathcal{P}_{\mathrm{col}(X_S)} \right), \quad\text{where }
\mathcal{P}_\mathrm{avg} := \frac{1}{B}\sum_{\ell=1}^B\mathcal{P}_{\mathrm{col}({X_{\widehat{S}^{(\ell)}}})} .
\label{eqn:stability}
\end{equation*}
Here, $\sigma_{|S|}(\cdot)$ computes the $|S|$-th singular value. A larger $\pi(S)$ indicates a more stable $S$. 
\label{defn:stability}
\end{definition}
Here, $\mathcal{P}_\mathrm{avg}$ is the average of projection matrices obtained from subsampling, although it is itself not a projection matrix.
The quantity $\pi(S)$ captures how aligned the feature subspace $\mathrm{col}(X_S)$ is with this average projection matrix. A simple calculation shows that $\pi(S) \in [0,1]$.


To better understand the metric $\pi(S)$, first consider size-one sets $S = \{j\}$. In this case, the stability measure can be written as $\pi(\{j\}) =\frac{1}{B} \sum_{\ell=1}^B \tau(\{j\}, \widehat{S}^{(\ell)})$, calculating how aligned the subspace spanned by $X_j$ is with the subspaces obtained from subsampling. It naturally adapts to the correlation structure of features: even if  $X_j$ is not often in the selection sets, $\pi(\{j\})$ will be close to one if $X_j$ is nearly in the span of features in $\widehat{S}^{(\ell)}$ for most $\ell$. This is in contrast to stability selection that would assign a low stability value to a potential signal $X_j$ that is highly correlated with non-signal features, offering more meaningful stability measures. 
For arbitrary-sized sets $S$, if $X_S$ consists of linearly dependent columns, the rank of the matrix $\mathcal{P}_{\mathrm{col}(X_S)}\mathcal{P}_\mathrm{avg}\mathcal{P}_{\mathrm{col}(X_S)}$ is strictly smaller than $|S|$, and thus $\pi(S) = 0$. In other words, if there are redundant features in the set, the stability would be zero. Otherwise, we can show that $\pi(S)$ can be equivalently expressed as 
\vspace{-0.1in}
\begin{align}\label{eq:stability-explain}
    \pi(S) = \min_{z \in \mathrm{col}(X_S)} \frac{1}{B} \sum_{\ell=1}^B \allowbreak\mathrm{trace}\allowbreak(\mathcal{P}_{\mathrm{col}(z)}\mathcal{P}_{\mathrm{col}(X_{\widehat{S}^{(\ell)}})})
\vspace{-0.1in}
\end{align}
(see Appendix~\ref{sec:equivalence} for a proof). This means that $\pi(S)$ measures the lowest alignment of any direction in the subspace $\mathrm{col}(X_S)$ with the subspaces obtained from subsampling. 

The reader may wonder why we did not use $\min_{j \in S}\pi(\{j\})$ to generalize our stability measure from size-one set $S$ to multi-element set $S$. For illustration, suppose $S = \{j,k\}$ where $X_j$ and $X_k$ are highly correlated. Although each variable individually lies nearly in the span of the subsampled models, they are rarely captured simultaneously. In this setting, selecting both $X_j$ and $X_k$ is redundant and undesirable. Although both $\pi(\{j\})$ and $\pi(\{k\})$ can be close to one and thus yield misleading assessments, our measure $\pi(\{j,k\})$ is necessarily no better than $\tr(\P_{\mathrm{col}(r)} \P_{\rm avg} )$, which quantifies the alignment between $\P_{\rm avg}$ and the direction $r = X_j - \P_{\mathrm{col}(X_k)} X_j \in \spa(\{X_k, X_j\})$. This alignment is expected to be small, reflecting the fact that $X_j$ and $X_k$ are seldom captured jointly. Although generally  $\pi(S)$ may not equal $\min_{j \in S}\allowbreak\pi(\{j\})$ for multi-element sets $S$, $\pi(S)$ simplifies to $\min_{j \in S}\pi(\{j\})$ when $X_j$ and $X_k$ are orthogonal.

We remark on some additional points. First, finding stable sets $S$ (i.e., sets with $\pi(S) \approx 1$) is not just about combining stable features with $\pi(\{j\}) \approx 1$. When features are highly correlated, combining individual stable features might be redundant and result in a lower $\pi(S)$. Second, if the features are orthogonal, $\pi(S)$ reduces to the usual measure of model stability used in stability selection \citep{meinshausen2010stability}. That is, for orthogonal features, $\pi(S) = \min_{j \in S}\frac{1}{B}\sum_{\ell=1}^B \mathbb{I}[j \in \widehat{S}^{(\ell)}]$, which computes the minimum proportion of subsample estimates that contain a given feature, taken over all features in $S$. Third, any sub-model $\widetilde{S}$ of $S$ (with $\widetilde{S} \subseteq S$) satisfies $\pi(\widetilde{S}) \geq \pi(S)$. 
This property leads to a definition of maximal stable sets that is more meaningful than small stable sub-models, such as individual features.
\begin{definition}[maximal $\alpha$-stable set]
\label{def:maximal-stable}
    A feature set $S$ is $\alpha$-stable if $\pi(S) \geq \alpha$. It is further a maximal $\alpha$-stable set if
    \vspace{-0.1in}
    $$
    \pi(S \cup \{j\}) < \alpha,\ \forall j\in\{1, \dots, p\} \setminus S.
    \vspace{-0.1in}
    $$
    Such a set is akin to the largest possible discovery subject to stability control.
\end{definition}
For a fixed $\alpha$, maximal $\alpha$-stable sets are not unique in general. In Section~\ref{sec:algorithm}, we describe algorithms to identify all, or a subset of these maximal sets.

\section{FSSS: An algorithm for selecting multiple stable models}
\label{sec:algorithm}
In this section, we propose a model selection algorithm, \textit{Feature Subspace Stability Selection} (FSSS), which efficiently identifies maximal $\alpha$-stable selection sets for a user-specified threshold $\alpha$ (according to the measure in Definition \ref{def:maximal-stable}). 
For notational simplicity, we use $\P_{v}:= \P_{\mathrm{col}(v)}$ for any vector $v\in\mathbb{R}^n$, $\P_{S}:= \P_{\mathrm{col}(X_S)}$ and $\P_{S^\perp}:= \P_{\mathrm{col}(X_S)^\perp}$ for any subset $S\subseteq\{1, \dots, p\}$.

The subsampling mechanism is crucial to the stability measure $\pi(S)$ in Definition \ref{defn:stability}. Our algorithm uses a subsampling mechanism~\citep{shah2013variable} that creates $B/2$ complementary subsamples of size $\lfloor n/2 \rfloor$, where $B$ is an even integer. For each subsample $\ell \in \{1, \dots, B\}$, a base procedure like Lasso or $\ell_0$-regression selects a feature set $\widehat{S}^{(\ell)} \subseteq \{1, \dots, p\}$. The matrix $\mathcal{P}_\mathrm{avg} = \frac{1}{B} \sum_{\ell=1}^B \P_{\widehat{S}^{(\ell)}}$ is then computed from these sets and used to measure the stability $\pi({S})$ of a set ${S}$.

\subsection{Reducing search space}
\label{sec:algorithm_rules}
Let $\alpha \in (1/2,1)$ be a user-specified threshold for the desired stability. We search for maximal $\alpha$-stable models through an iterative model expansion process, in which starting from the null set, features are added one at a time along multiple paths. Whenever $\pi(S)\geq \alpha$, each candidate $S \cup \{j\}$, $j\in\{1, \dots, p\} \setminus S$ creates a new path in the search process. However, searching over the model space $\mathcal{M}:=\{S: S \subseteq \{1, \dots, p\}\}$ can be computationally demanding and become intractable when $p$ is large. Leveraging the properties of the stability metric $\pi(S)$, we develop a set of tricks that substantially reduce the search space.

First, two search paths may lead to the same model $\bar{S}$ when $S_1 \cup \{j_1\} = S_2 \cup \{j_2\} = \bar{S}$. If $\bar{S}$ has been fully examined along the path originating from $S_1$, then expanding $S_2$ into $\bar{S}$ is unnecessary. Second, the monotonicity of stability prevents any unstable set $S$ from further expanding; that is, if $\pi(S) < \alpha$, then $\pi(S \cup \{j\}) \leq\pi(S)<
\alpha$ for all $j \in \{1, \dots, p\} \setminus S$. Third, monotonicity further reduces unnecessary evaluations of the stability metric, considering computing singular values can be computationally expensive. Specifically, if $S$ has a super-set $\widetilde{S}$ that has already been verified as stable, then $\pi(S) \geq \pi(\widetilde{S}) \geq \alpha$ follows immediately, and no additional evaluation of $\pi(S)$ is required.

Finally, a pre-evaluation can be performed before evaluating the stability of expansion $S \cup \{j\}$. Adding a new feature $X_j$ to $S$ amounts to including a new orthogonal direction $v_j = X_j - \P_{X_S} X_j$ into the subspace $\mathrm{col}(X_S)$, resulting in $\mathrm{col}(X_{S\cup \{j\} }) = \mathrm{col}(X_S) \oplus \mathrm{col}(v_j)$. Consequently, appealing to \eqref{eq:stability-explain}, if $\tr(\P_{v_j} \P_{\rm avg}) < \alpha$, then $\pi(S \cup\{j \}) \leq \tr(\P_{v_j} \P_{\rm avg}) < \alpha$. In other words, the quantity $\tr(\P_{v_j} \P_{\rm avg})$ serves as a early rejection criterion for stability; failing to satisfy $\tr(\P_{v_j} \P_{\rm avg}) \geq \alpha$ would immediately rule out the candidate $S\cup \{j\}$. This pre-evaluation substantially reduces the computational burden of evaluating the stability metric across all candidate models and, as discussed in Section~\ref{sec:algorithm_FSSS}, provides a surrogate measure of stability for weighting candidate models.

\subsection{Algorithm description}
\label{sec:algorithm_FSSS}

A search procedure that identifies all maximal $\alpha$-stable models must, in principle, examine all model combinations, even though many candidates can be efficiently eliminated using the proposed search-space reduction rules. Organizing the model space as a tree and performing a depth-first search guarantees completeness. Nevertheless, this approach is computationally prohibitive when $p$ is large.

In contrast to identifying all maximal $\alpha$-stable models, an alternative approach is to sample a sub-collection of maximal $\alpha$-stable models, where we randomly choose a search path to proceed. Suppose we are at an $\alpha$-stable model $S$, then the next model is chosen among $S \cup \{j\}$, $j\in\{1, \dots, p\}\setminus S$ with probability $w_j$ as defined in \eqref{eq:path-weight}. Notably, although a more proper choice for $w_j$ is the normalized $\pi(S \cup \{j\}) \cdot\mathds{1}(\pi(S \cup \{j\}) \geq \alpha)$, we replace $\pi(S \cup \{j\})$ with $\tr(\P_{v_j} \P_{\rm avg})$ for the sake of computational efficiency. This procedure can be viewed as a random walk in the model space $\mathcal{M}$ until a maximal $\alpha$-stable model is found. 

FSSS is summarized in Algorithm~\ref{algo:all-path}, which incorporates the search-space reduction rules introduced in Section~\ref{sec:algorithm_rules}. The sets {MAX\_STABLE} and {VISITED} both store models that have been fully explored, meaning that all of their sub-models have been examined; in particular, {MAX\_STABLE} contains only maximal $\alpha$-stable models. 
Suppose the algorithm is currently at an $\alpha$-stable model $S$. It first constructs a candidate set $\mathcal{F}$ consisting of features that can potentially be added to $S$, excluding those that would lead to models that have already been fully explored. A feature $j \in \mathcal{F}$ is then selected at random with probability $w_j$. If the resulting model $S \cup \{j\}$ is $\alpha$-stable, the algorithm expands $S$ to this new model; otherwise, $j$ is removed from $\mathcal{F}$ and another feature is selected from $\mathcal{F}$.
If the candidate set $\mathcal{F}$ is empty, the model $S$ is declared fully explored. When this occurs, if no super-set of $S$ appears in {MAX\_STABLE}, then all one-step extensions of $S$ must be unstable, implying that $S$ itself is a maximal $\alpha$-stable model. The algorithm then restarts the expansion process from the null set and continues until $K$ maximal $\alpha$-stable models are identified. 
A special case is when the total number of maximal $\alpha$-stable models (denoted as $\bar{K}$) is smaller than the user-specified $K$. In this case, after all maximal $\alpha$-stable models are identified, the candidate set $\mathcal{F}$ for the null model would eventually be empty, and the algorithm terminates automatically.

\begin{algorithm}[!ht] 
    \caption{Feature subspace stability selection (FSSS)}
    \label{algo:all-path}
    {\small 
    \begin{algorithmic}
    \STATE \textbf{Input:} {Number of models $K\in\mathbb{N}$, features $X\in\mathbb{R}^{n\times p}$, response variable $y\in\mathbb{R}^n$, stability threshold $\alpha\in(1/2, 1)$, even integer $B$, and a base variable selection procedure. } 
    \vspace{0.05in}
    \STATE\textbf{Subsampling:} Obtain estimates $\{\widehat{S}^{(2\ell-1)}, \widehat{S}^{(2\ell)}\}_{l=1}^{B/2}$ from applying base procedure to complementary subsamples, and compute $\P_{\rm avg} = \frac{1}{B} \sum_{l=1}^B \P_{\widehat S^{(\ell)}}$. 
    \vspace{0.05in}

    \STATE\textbf{Initialization:} Set $\mathrm{MAX\_STABLE} \gets \{\}$, $\mathrm{VISITED}\gets \{\}$.
    \vspace{0.05in}

    \STATE\textbf{Sequentially add variables:}

    \begin{enumerate}[label=(\arabic*)]

        \item Set $S \gets \varnothing$.
        \item Let 
        $$
        \mathcal{F} \gets \left\{ j\in\{1, \dots, p\}\setminus S : S \cup \{j\} \notin \text{MAX\_STABLE $\cup$ VISITED},\  \tr( \P_{v_j} \P_{\rm avg}) \geq \alpha \text{ for }v_j = \P_{S^\perp}X_j \right\}  .
        $$
        \item If $\mathcal{F} = \varnothing$, go to step (4). Otherwise, for each $j\in \mathcal{F}$,
        \begin{align}\label{eq:path-weight}
            w_j \gets \frac{ \tr(\P_{v_j} \P_{\rm avg}) \cdot \mathds{1}(\tr(\P_{v_j} \P_{\rm avg}) \geq \alpha) }{ \sum_{i\in \mathcal{F} } \tr(\P_{v_i} \P_{\rm avg}) \cdot \mathds{1}(\tr(\P_{v_i} \P_{\rm avg}) \geq \alpha)  },
        \end{align}
        and select $j \in \mathcal{F}$ with probability $w_j$.
        \begin{enumerate}[label = (\alph*)]
            \item If $S \cup \{j\}$ has no super-set in MAX\_STABLE, then evaluate $\pi(S \cup \{j\})$. Otherwise it must be $\pi(S\cup \{j\}) \geq \alpha .$
            
            \item If $\pi(S\cup \{j\}) \geq \alpha $, let $S \gets S \cup \{j\}$, and go to step (2). Otherwise, add $S\cup \{j\}$ into VISITED, let $\mathcal{F} \gets \mathcal{F} \setminus \{j\}$, and repeat step (3).            
        \end{enumerate}

        \item If $S=\varnothing$, return MAX\_STABLE. Otherwise, 
        \begin{enumerate}
            \item If $S$ has no super-set in MAX\_STABLE, add $S$ into MAX\_STABLE.
            \item Otherwise, add $S$ into VISITED.

            \item If $|\mathrm{MAX\_STABLE}| < K$, go to step (1). Otherwise, return MAX\_STABLE.
        \end{enumerate}

    \end{enumerate}
    \vspace{0.04in}
    \STATE{\textbf{Output:}} {MAX\_STABLE containing $\min\{K, \bar{K}\}$ models.}
    \end{algorithmic}
    }
\end{algorithm}

\begin{remark}
    FSSS can be modified into a greedy algorithm, by fixing $K = 1$ and setting $w_j  = \mathds{1}( j = \argmax_{i\in \mathcal{F}} \tr(\P_{v_i} \P_{\rm avg}) )$. That is, at each model $S$, we always select the candidate $S\cup \{i\}$ that maximizes $\tr(\P_{v_i} \P_{\rm avg})$. For a given $\P_{\rm avg}$, the greedy FSSS always returns the same one selection set. Further, it can be seen as a special case of forward stepwise selection, where a new variable is added at each step if the resulting model is $\alpha$-stable.
    \label{remark:greedy}
\end{remark}

\begin{remark}
    When the features are orthogonal, FSSS reduces to the stability selection algorithm. Specifically, $\tr(\P_{v_j}\P_{\mathrm{avg}})$ simplifies to the selection proportion of feature $j$, and $\pi({S}\cup\{j\})$ simplifies to the smallest selection proportion of features in ${S}\cup \{j\}$.  Thus, when the features are orthogonal, FSSS will return a set consisting of those features with a selection proportion greater than $\alpha$.
\end{remark}

\section{Theoretical properties of FSSS}
\label{sec:theory}
We provide statistical guarantees for the models returned by FSSS. Section~\ref{sec:clustering_setup} introduces a data-generating mechanism in which features are highly correlated within clusters but independent across clusters, with each cluster containing at most one signal feature. Although stylized, this construction captures a common and challenging scenario in feature selection. Analyzing the behavior of maximal $\alpha$-stable sets in this framework provides insight into their behavior in related tasks and more general settings.
Section~\ref{sec:theo-results} presents two theoretical results under this clustering design. The first establishes control of the expected false positive error, while the second provides a more concrete guarantee for feature selection. The latter result is stronger but requires a more stringent condition on $\alpha$. Finally, Section~\ref{sec:complex_dependency} introduces an alternative data-generating mechanism to broaden the scope of the analysis.

The theoretical analysis in this section differs substantially from those of \cite{meinshausen2010stability} and \cite{taeb2020false}. Specifically, the analysis in \cite{meinshausen2010stability} is entirely discrete and is based on feature selection proportions across subsamples. On the other hand, \cite{taeb2020false} study general subspace selection without explicitly linking it to individual feature selection. Our analysis integrates both the discrete and continuous perspectives, and particularly associates the subspace $\mathrm{col}(X_{S^\star})^\perp$ with noise and perturbation features. Such a connection would not be possible without fully exploring the data-generation mechanism.

We use the big-O notation $\mathcal{O}(f(n,p))$ to suppress constants and higher-order terms as $f(n,p) \to 0$ when $n,p\rightarrow \infty$. Specifically, $g(n,p) = \mathcal{O}(f(n,p))$ means there exist constants $C > 0$, such that $|g(n,p)| \leq C |f(n,p)|$ when $f(n,p) \to 0$. For constants $c^\star > 0$ and $c_\star < 0$, we write $\mathcal{O}(f(n,p)) \leq c^\star$ if $C|f(n,p)| \leq c^\star$, and $\mathcal{O}(f(n,p)) \geq c_\star$ if $C|f(n,p)| \leq -c_\star$, when $n$ and $p$ are sufficiently large.

\subsection{Setup}
\label{sec:clustering_setup}
We assume that the data consists of multiple clusters, where features within each cluster are small perturbations of a \emph{representative feature} (and thus highly correlated), while features between clusters are nearly orthogonal. We consider extensions to more complex structures in Section~\ref{sec:complex_dependency}.

Let $\mathcal{K}$ be the index set of {representative features}. For each $k \in \mathcal{K}$, let $\mathcal{V}_k$ be the index set of \textit{proxy features} associated with $X_k$, and let $\mathcal{V}:= \bigcup_k \mathcal{V}_k$ be the index set of all proxy features. The cluster corresponding to $X_k$ is $\mathcal{C}_k := \{k\} \cup \mathcal{V}_k$, which includes both the representative feature and its proxies. We let $D := \max_{k \in \mathcal{K}}|\mathcal{C}_k|$ be the maximum number of features in any cluster. For any $k \in \mathcal{K}$, we assume for simplicity that $X_k$ is normalized so that $\|X_k\|_{\ell_2} = 1$. Further, its proxy features are generated as $X_j = X_k + \delta_j$ for $j \in \mathcal{V}$, where $\delta_j$ is a perturbation with $\|\delta_j\|_{\ell_2} = \eta_1 \ll 1$. We also assume that all representative features and perturbations are nearly orthogonal, meaning that $\mathrm{Corr}(X_{k},X_{k'})$ (for any $k,k' \in \mathcal{K}$), $\mathrm{Corr}(X_{k},\delta_{j})$ (for any $k \in \mathcal{K}$ and $j \in \mathcal{V}$) and $\mathrm{Corr}(\delta_{j},\delta_{j'})$ (for any $j,j' \in \mathcal{V}$) are bounded by $\mathcal{O}(\sqrt{\log p/n})$.  We suppose the response $y \in \mathbb{R}^n$ is generated according to the linear model $
y = X\beta^\star + \epsilon,
$ where $S^\star \subseteq \mathcal{K}$ is the support of $\beta^\star \in \mathbb{R}^p$ with $s^\star := |S^\star|$ and $N^\star := (S^\star)^c$. Here, $\epsilon \in \mathbb{R}^n$ is a random vector with independent and identically distributed entries. Note that $X \in \mathbb{R}^{n \times p}$ is viewed as a fixed (non-random) feature matrix. 

Since $\eta_1 \ll 1$, unless the sample size is very large or the signal-to-noise ratio is very high, one cannot distinguish a representative signal feature from the highly correlated features in its cluster. As a result, targeting the signal set $S^\star$ is not meaningful. Instead, we define as our target a class of ``equally good models'':
\vspace{-0.1in}
\begin{align}
\begin{split}
    \mathcal{S}:= \left\{S \subseteq \{1,\dots,p\}:
\left|S \cap \mathcal{C}_k \right| = 1, \forall k\in S^\star; \left| S \cap \mathcal{C}_k \right| = 0, \forall k\in N^\star\cap \mathcal{K}
\right\} .
\end{split}
\label{eqn:set_S}
\vspace{-0.1in}
\end{align}
Each model $S \in \mathcal{S}$ contains exactly one feature from each signal cluster (a cluster with a representative signal feature) and no features from non-signal clusters. In Remark \ref{rem:l0-base}, we will exhibit results for a common variable selection method showing that targeting $\mathcal S$ can be achieved with much smaller sample size than targeting $S^\star$. In Appendix~\ref{sec:sup-prediction-similarity}, we will show that all models in $\mathcal{S}$ have similar predictive power.


\subsection{False positive error control and consistency using FSSS}
\label{sec:theo-results}

Given the average projection matrix $\P_{\rm avg}$, let $\mathfrak{S}$ be the collection of all maximal $\alpha$-stable sets for some $\alpha \in (1/2,1)$. FSSS returns a sub-collection of $\mathfrak{S}$. In this subsection, we show that, under the clustering setup in Section~\ref{sec:clustering_setup}, the expected false positive error is bounded for any member in $\mathfrak{S}$. Additionally, we show that FSSS can consistently estimate ``equally good models'' in $\mathcal{S}$ as defined in \eqref{eqn:set_S}, meaning that with high probability (the ``better'' the base procedure, the higher the probability), any $S \in \mathfrak{S}$ also belongs to $\mathcal{S}$. 

Let $\mathscr{T}$ be the collection of subsample indices used in FSSS's subsampling process. For simplicity, we assume this collection is fixed (e.g., if FSSS uses all complementary $n$ choose $\lfloor n/2\rfloor$ subsamples). Let $\widehat{S}:\mathscr{T} \to \{1,\dots,p\}$ be a base procedure that takes an index set $T \in \mathscr{T}$ and returns a selection set based on the subsamples corresponding to $T$. We assume that for all $T \in \mathscr{T}$, the base procedure produces a selection set with at most $s_0$ features, where $s_0 \ll p$.
Since FSSS is a wrapper around a base procedure, we expect its performance to rely on how ``good'' the base procedure is. Our theoretical analysis thus involves the following quantities that captures the ``quality of the base procedure''. 

\begin{definition}[Quality of the base procedure]\label{def:quality_base_procedure}
We quantify the quality of the base procedure by the quantities 
$$
\gamma_\mathcal{W} := \max_{u \in \mathcal{W}} \max_{T\in \mathscr{T}} \mathbb{E}\left[\tr(\P_{\mathrm{col}(u)}
\P_{\mathrm{col}(X_{\widehat{S}(T)})}) \right], \quad
\gamma_\mathcal{U} := \max_{u \in \mathcal{U}}  \max_{T\in \mathscr{T}}  \mathbb{E}\Big[\tr (\P_{\mathrm{col}(u)} \P_{\mathrm{col}(X_{\widehat{S}(T)})})\Big],
$$ 
where the sets $\mathcal{W} := \{ \delta_j : j\in\mathcal{V} \} \cup \bigcup_{j\in \mathcal{K} \cap N^\star} \{X_j\}$ and $\mathcal{U}:= \bigcup_{k\in S^\star} \{\delta_j - \delta_l: j\neq l \in \mathcal{V}_k \} \cup \bigcup_{k\in \mathcal{K} \cap S^\star} \{\delta_j: j\in\mathcal{V}_k\} \cup \bigcup_{k\in \mathcal{K} \cap N^\star} \{ X_j: j\in\mathcal{C}_k \}$ consist of ``noise directions''.
\end{definition}

Here, the sets $\mathcal{W}$ and $\mathcal{U}$ consist of different combinations of perturbation directions as well as features in noise clusters. The quantity $\gamma_\mathcal{W}$ (respectively, $\gamma_\mathcal{U}$) represents the degree of alignment (averaged over the randomness in $\epsilon$) between any noise direction in $\mathcal{W}$ (respectively, $\mathcal{U}$) and a subspace $\mathrm{col}(X_{\widehat{S}(T)})$ obtained by the base procedure. A small $\gamma_\mathcal{W}$ or $\gamma_\mathcal{U}$ means that the base procedure does not significantly favor any noise direction in the corresponding set. In Appendix~\ref{sec:sup-bpquality}, under assumptions similar in spirit to those in other stability-based methods \citep{meinshausen2010stability,shah2013variable,faletto2022cluster}, we show that $\gamma_\mathcal{W}\approx s_0/p$. Additional insights on the behavior of $\gamma_\mathcal{U}$ are also discussed therein.

Given the subsamples, the following theorem bounds the expected false positive error of any maximal $\alpha$-stable set in terms of the base procedure quality $\gamma_\mathcal{W}$. 
\begin{theorem}\label{prop:stab-ubd}
    Suppose that $(i)$ $\mathrm{Corr}(X_j,X_k)^2 \leq {1}/{2}$ for any $S \in \mathfrak{S}$ and $j,k \in S$; and $(ii)$ $1 / (1 + \eta_1^2) + \mathcal{O}( \sqrt{\log p / n} ) > \sqrt{2} / 2$. 
    Then, any set $S \in \mathfrak{S}$ has the false positive error bound
    \begin{equation}
   \mathbb{E} \left[\mathrm{FP}(S,S^\star)\right]
    \leq \frac{p(\gamma_\mathcal{W}+  b)^2}{2\alpha-1} + a,
\label{eqn:false_discovery_bound}
\end{equation}    
where 
    $
    a = s_0  [\eta_1^2 / (1 + \eta_1^2) + \mathcal{O} ( \sqrt{ \log p / n } ) ]
    $, and
    $
     b = 2\sqrt{\eta_1^2 / (1 + \eta_1^2)} + \mathcal{O} (s_0 \sqrt{\log p / n} ).
    $
    \label{thm:1}
\end{theorem}
We prove Theorem~\ref{thm:1} in Appendix~\ref{sec:sup-proofthm1}. The conditions of this theorem are mild: $(i)$ can be easily enforced in Algorithm~\ref{algo:all-path} with little effect on its outputs, and $(ii)$ holds as long as the perturbation $\eta_1$ is not too large. 
From \eqref{eqn:false_discovery_bound}, as expected, a larger stability threshold $\alpha$ and a better base procedure (i.e., smaller $\gamma_\mathcal{W}$) lead to a smaller false positive error. The quantities $a$ and $b$ account for the correlation between feature pairs in different clusters (assumed to be on the order of $\sqrt{\log{p}/n}$) and the perturbation level $\eta_1$ within each cluster. If the features are orthogonal (i.e., features form size-one clusters with zero correlations among them and $\eta_1 = 0$), the quantities $a$ and $b$ vanish, and the bound in \eqref{eqn:false_discovery_bound} simplifies to the false discovery bounds of stability selection \citep{meinshausen2010stability,shah2013variable}. This reduction follows from the fact that our subspace notion of false positive error reduces to the standard notion in highly correlated settings. 


Theorem~\ref{thm:1} is rather general and makes few assumptions about the base procedure. Next, we show that under certain assumptions on the base procedure, FSSS produces models in the set $\mathcal{S}$ with high probability.
\begin{theorem}\label{prop:clust-consist-subsample}
    Suppose that: $(i)$ the base procedure, when applied to a subsample corresponding to some index set $T \in \mathscr{T}$, selects at least one feature from every signal cluster; $(ii)$ $\eta_1^2 / (1 + \eta_1^2) + \mathcal{O}(s^\star \sqrt{\log p / n}) < 1/2 $; and $(iii)$ $s^\star \leq (n / \log p)^{1/8}$.  Let $
    f_1(s^\star, \eta_1) := 1+20 s^\star ( 2\eta_1^2 (\eta_1^2 + 2) )^{1/2} / (\eta_1^2 + 1) + 8(s^\star)^2 \eta_1^2 (\eta_1^2 + 2) / (1 + \eta_1^2)^2$, and
    $f_2(s^\star, \eta_1, s_0) := 1 - ( s^\star s_0 + 2(s^\star)^2 s_0 + 3(s^\star)^3  )  ( \eta_1^2 / (1 + \eta_1^2) )^{1/2} - ( (s^\star)^2 s_0 + 3 (s^\star)^3  ) \eta_1^2 / (1 + \eta_1^2) - (s^\star)^3  ( \eta_1^2 (1 + \eta_1^2) )^{3/2}
    $. Then, for any $\alpha_0 \in (0,1/2)$, if the stability threshold $\alpha$ is chosen so that:
    \vspace{-0.1in}
    \begin{eqnarray*}
    f_1(s^\star, \eta_1)-\alpha_0 + \mathcal{O}\left( (s^\star)^2 \sqrt{{\log p}/{n}} \right) \leq  \alpha \leq  f_2(s^\star, \eta_1, s_0) + \mathcal{O}\left( (s^\star)^2 s_0 D \sqrt{{\log p}/{n}} \right),
    \vspace{-0.1in}
    \end{eqnarray*}
    we have with probability $1 - ( s^\star \binom{D-1}{2} + p - s^\star ) \cdot \gamma_\mathcal{U}^2/(1-2\alpha_0)$ that $\mathfrak{S} \subseteq \mathcal{S}$. 

\label{thm:2}
\end{theorem}
We prove Theorem~\ref{thm:2} in Appendix~\ref{sec:sup-proofthm2}. This theorem states that if the base procedure picks at least one feature from each signal cluster, the perturbation $\eta_1$ is small, the true model is sparse, and the stability threshold $\alpha$ is chosen well, then FSSS will likely return models in the set of ``equally good models'' $\mathcal{S}$. Importantly, the result holds even if the base procedure chooses many noise or redundant features, highlighting the benefits of selecting stable models. Notice that $\alpha$ cannot be too small or too large: if it is too small, FSSS may pick up noise features (i.e., features that are not in any set $S \in \mathcal{S}$); if it is too large, FSSS may miss ``signal'' features (i.e., features that either are signal or are highly correlated with signals). The interval of good choices for $\alpha$ expands as $\eta_1$ decreases, with the interval approaching $(1 - \alpha_0 + \mathcal{O}((s^\star)^2\sqrt{\log{p}/n}), 1)$ as $\eta_1 \to 0$. Further, as with Theorem~\ref{thm:1}, the quality of the base procedure $\gamma_\mathcal{U}$ plays a key role in Theorem~\ref{thm:2}: smaller values of $\gamma_\mathcal{U}$ lead to a higher probability of FSSS returning models in $\mathcal{S}$. 
\begin{remark}\label{rem:l0-base}
We provide a base procedure that, with high probability, satisfies condition $(i)$ in Theorem~\ref{thm:2}. For any index set $T \in \mathscr{T}$, let $y_T \in \mathbb{R}^{\lfloor n/2 \rfloor}$ and $X_{T,:} \in \mathbb{R}^{\lfloor n/2 \rfloor \times p}$ represent the subsample corresponding to $T$. We apply the following $\ell_0$-regression method to this subsample,
\vspace{-0.1in}
\begin{eqnarray*}
\widehat{\beta}(T) \in \arg\min_{\beta \in \mathbb{R}^p} \|y_T - X_{T,:}\beta\|_{\ell_2}^2, \quad\text{subject-to}\quad \|\beta\|_{\ell_0} \leq s_0,
\label{eqn:ell0_estimator}
\vspace{-0.1in}
\end{eqnarray*}
with $\widehat{S}(T)=\{j: \widehat{\beta}_j(T) \neq 0\}$ representing its selection set. 
Here, $\|\beta\|_{\ell_0}$ computes the number of non-zero entries in $\beta$, and $s_0$ is a user-specified threshold for the maximal number of features allowed in the regression model. In Appendix~\ref{sec:sup-cluster}, for Gaussian $\epsilon$, we show that if $s_0 \geq s^\star$, the nonzero coefficients $\beta^\star$ are sufficiently away from zero, and the sample size is sufficiently large, then condition $(i)$ is satisfied with high probability. If we further set $s_0 = s^\star$, then $\ell_0$-regression selects a model $S \in \mathcal{S}$ with high probability. In contrast, FSSS achieves a similar result requiring only $s_0 \geq s^\star$ rather than $s_0 = s^\star$.
As a side note, recovering the exact model $S^\star$ may require a substantially larger sample size, matching the intuition that $\mathcal{S}$ is a more reasonable and attainable target.
\end{remark}

\subsection{Theoretical analysis beyond the clustering setup}
\label{sec:complex_dependency}
In Appendix~\ref{sec:sup-complex}, we extend our theoretical results to more complex dependency structures than the clustering setup in Section~\ref{sec:clustering_setup}. Specifically, we consider a parent-child dependency, where a child feature is a perturbed linear combination of its parent features. In this setup, we define a set of ``equally good models'' similar to $\mathcal{S}$ in \eqref{eqn:set_S}. We then show that, under certain assumptions on the base procedure, FSSS returns models from the set of ``equally good models'' with high probability, and this probability increases as the base procedure improves.

\section{Simulations and real data application}
\label{sec:simulation}

In this section, we present experiments on synthetic and real data to illustrate the utility of FSSS (Algorithm~\ref{algo:all-path}).

\subsection{Synthetic study}
In Figure~\ref{fig:stability}, we presented a synthetic experiment in which our subspace-based notion of stability outperforms previous stability-based metrics in distinguishing between noise and non-noise features. We show in the Appendix~\ref{sec:sup-figure1} that in this setting, compared to stability selection and cluster stability selection, our method FSSS yields a better predictive model with a substantially higher amount of true discoveries, while maintaining a similar false positive error. 

Next, we analyze a more challenging synthetic setup with a low signal-to-noise ratio. In this setup, we generate: three signal clusters, ${\mathcal{C}_k}$ $(k \in \{1, 4, 7\})$, each of size three; three complex dependency blocks, $\mathcal{B}_1$, $\mathcal{B}_2$, and $\mathcal{B}_3$, with a parent–child relationship;  and 179 individual features, amounting to a total of $p = 200$ features. We inherit the notations in Section \ref{sec:theory}.

For each cluster $\mathcal{C}_k$ ($k = 1, 4, 7$), the representative feature is generated as
$X_k \overset{iid}{\sim} \mathcal{N}(0, I_n)$, and its proxies are defined as $X_j = X_k + \delta_j$ for all $j \in \mathcal{V}_k = \{k+1, k+2\}$, where the perturbation directions follow
$\delta_j \overset{iid}{\sim} \mathcal{N}(0, 0.5^2 I_n)$. 
We set $\beta^\star_k = 1$ for $k \in \{1, 4, 7\}$, while $\beta^\star_j = 0$ for any $j\in \{\mathcal{V}_k\}_{k\in\{1,4,7\}}$.
The first dependency block $\mathcal{B}_1$ consists of two parents $\mathcal{K}_1 = \{10, 11\}$ and a child $\{c_1\} = \{12\}$. The second dependency block $\mathcal{B}_2$ consists of three parents $\mathcal{K}_2 = \{13, 14, 15\}$ and a child $\{c_2\} = \{16\}$. The last dependency block $\mathcal{B}_3$ consists of four parents $\mathcal{K}_3 = \{17, 18, 19, 20\}$ and a child $\{c_3\} = \{21\}$. For each block $\mathcal{B}_l$ ($l = 1,2,3$), the parent features are generated by
$X_j \overset{iid}{\sim} \mathcal{N}(0, I_n)$, $j\in \mathcal{K}_l$, and the child is defined as $X_{c_l} = \sum_{j\in \mathcal{K}_l} X_j + \delta_{c_l}$, where the perturbation directions follow $\delta_{c_l} \overset{iid}{\sim} \mathcal{N}(0, \eta_l^2 I_n)$ with $\eta_1 = 0.01$, $\eta_2 = 0.1$ and $\eta_3 = 0.1$.
We set all parents as signals: $(\beta^\star_{10}, \beta^\star_{11}) = (1, -1)$, $(\beta^\star_{13}, \beta^\star_{14}, \beta^\star_{15}) = (1, -1, 1)$, and $(\beta^\star_{17}, \beta^\star_{18}, \beta^\star_{19}, \beta^\star_{20}) = (1, -1, 1, -1)$. The children features are set to be noise, i.e.,  $\beta^\star_j = 0$ for all $j\in \{12, 16, 21\}$.
Moreover, the individual features $\{X_j\}_{j \in \mathcal{I}}$, $\mathcal{I} = \{22, \dots, 200\}$ are given by
$
X_j \overset{iid}{\sim} \mathcal{N}(0, I_n)$, for all $j \in \mathcal{I}.
$
We pick 5 weak-signal individual features with $\beta^\star_j = 0.2$ for $j \in \{22, \dots, 26 \}$. The remaining features are noise, i.e.,  $\beta^\star_j = 0$ for $j \in \{ 27, \dots, 200 \}$.
Finally, the response variable is generated from the linear model
$y = X \beta^\star + \epsilon$, where $\epsilon \sim \mathcal{N}(0, \sigma^2 I_n)$, drawn independently from the other components. We set $\sigma = 1.5$, the training sample size to be $n_\text{training} = 600$, and the testing sample size to be $n_\text{test} = 500$. 

We compare the performance of greedy variant of FSSS (see Remark~\ref{remark:greedy}) with the stability-based methods stability selection (SS) and cluster stability selection (CSS). (For CSS, we use the version \emph{SPS} that outputs features instead of assigned clusters). For each procedure, model selection is carried out using a three-fold partition of the training data: subsampling (using $B = 200$) and the selection procedure are performed on the first fold (size $n_{\text{fitting}}{ = 200}$), the model coefficients are estimated on the second fold (size $n_{\text{model}}{ = 200}$), and validation is performed on the third fold (size $n_{\text{val}} { = 200}$). 
For all stability-based methods, the stability threshold $\alpha$ is selected from the set $\{0.8,0.85,0.9,0.95\}$ based on validation mean squared error (MSE). In the case of CSS, the hierarchical clustering cutoff height $h$ is also chosen from the set $\{0.1,0.3,0.5\}$ using validation MSE. 

The following performance metrics are evaluated on the test data: test MSE, the amount of true positives (TP) and false positives (FP) as measured in \eqref{eq:fd_td_def}. In addition, we assess the consistency of the output models across different runs. To this end, we generate $M=50$ additional training trials and evaluate the stability of the resulting models using the \textit{output stability (OS)} metric, computed over the $M$ models obtained from each trial:
\vspace{-0.1in}
$$
\text{OS} :=
\frac{2}{M(M-1)}  \sum_{i=1}^M \sum_{j=i+1}^M \bar{\tau}(\widehat{S}^{[i]}, \widehat{S}^{[j]}) \in [0, 1].
\vspace{-0.1in}
$$
Here, $\widehat{S}^{[i]}$ denotes the selection set obtained from the $i$-th trial, and $\bar{\tau}$ is the normalized similarity in \eqref{eqn:normalized_similarity}. As discussed in Section~\ref{sec:similarity}, the metric $\text{OS}$ quantifies average normalized similarity of the selected models across the $M$ trials, with larger values indicating a more robust variable selection procedure.

We consider two base procedures. The first is $\ell_0$-regression (described in Remark \ref{rem:l0-base}), which is approximately solved using the package \texttt{L0Learn} \citep{hazimeh2023l0learn}. The second is Lasso~\citep{tibshirani1996regression}, solved using the package \texttt{glmnet}~\citep{friedman2010regularization}. Let $s_0$ denote the number of features selected by a base procedure. For each $s_0 \in \{2, 3, \dots, 41\}$, we apply the experiment described above, obtain the performance metrics, and average the results over 500 independent repetitions. Performance metric values for all $s_0$ are reported in Appendix~\ref{sec:sup-synthetic}. While one method may outperform another at specific $s_0$ values, a fair comparison requires evaluating each method at its own optimal $s_0$ that yields the best performance. 

Table \ref{tab:performance} summarizes the results, with $s_0$ chosen to minimize the test MSE for each method. First, we observe that Lasso can be substantially stabilized by SS, CSS, and FSSS, leading to significant reductions in false positive error and notable gains in output stability. In contrast, the gains from applying FSSS to $\ell_0$-regression are less pronounced, demonstrating the inherent strength of $\ell_0$-penalized regression in handling highly correlated features. Notably, $\ell_0$-regression outperforms SS in terms of MSE, true positives, and output stability, as SS is affected by vote splitting, whereas $\ell_0$ tends to select features that span similar subspaces. While CSS mitigates vote splitting within clusters, its performance may still be limited by imperfect cluster assignments that fail to reliably capture the underlying dependency structure. Overall, these results suggest that greedy FSSS, using $\ell_0$-regression as its base procedure, offers the most balanced performance by effectively controlling MSE and false positive error while achieving the highest output stability; it also yields the largest amount of true positives among the stable methods.

\begin{table}[!ht]
\centering
{\small
\begin{tabular}{|l|cccc||l|cccc|}
\hline
\multicolumn{1}{|c|}{}                                       & MSE                                                   & FP                                                   & TP                                                     & OS                                                    &                                                              & MSE                                                   & FP                                                   & TP                                                     & OS                                                    \\ \hline
$\ell_0$                                                     & \begin{tabular}[c]{@{}c@{}}0.19\\ (0.01)\end{tabular} & \begin{tabular}[c]{@{}c@{}}0.77\\ (0.21)\end{tabular} & \begin{tabular}[c]{@{}c@{}}12.10\\ (0.23)\end{tabular} & \begin{tabular}[c]{@{}c@{}}0.88\\ (0.02)\end{tabular} & \begin{tabular}[c]{@{}l@{}}Lasso \\ $\ell_1$\end{tabular}    & \begin{tabular}[c]{@{}c@{}}0.19\\ (0.01)\end{tabular} & \begin{tabular}[c]{@{}c@{}}4.52\\ (0.31)\end{tabular} & \begin{tabular}[c]{@{}c@{}}12.80\\ (0.30)\end{tabular} & \begin{tabular}[c]{@{}c@{}}0.73\\ (0.01)\end{tabular} \\ \hline
\begin{tabular}[c]{@{}l@{}}FSSS\\ $[\ell_0 \text{ BP}]$ \end{tabular} & \begin{tabular}[c]{@{}c@{}}0.19\\ (0.01)\end{tabular} & \begin{tabular}[c]{@{}c@{}}0.08\\ (0.04)\end{tabular} & \begin{tabular}[c]{@{}c@{}}11.51\\ (0.30)\end{tabular} & \begin{tabular}[c]{@{}c@{}}0.95\\ (0.01)\end{tabular} & \begin{tabular}[c]{@{}l@{}}FSSS\\ $[\ell_1 \text{ BP}]$ \end{tabular} & \begin{tabular}[c]{@{}c@{}}0.19\\ (0.01)\end{tabular} & \begin{tabular}[c]{@{}c@{}}0.90\\ (0.60)\end{tabular} & \begin{tabular}[c]{@{}c@{}}12.30\\ (0.30)\end{tabular} & \begin{tabular}[c]{@{}c@{}}0.89\\ (0.06)\end{tabular} \\ \hline
\begin{tabular}[c]{@{}l@{}}CSS \\ $[\ell_0 \text{ BP}]$ \end{tabular} & \begin{tabular}[c]{@{}c@{}}0.26\\ (0.04)\end{tabular} & \begin{tabular}[c]{@{}c@{}}0.14\\ (0.10)\end{tabular} & \begin{tabular}[c]{@{}c@{}}11.01\\ (0.54)\end{tabular} & \begin{tabular}[c]{@{}c@{}}0.89\\ (0.04)\end{tabular} & \begin{tabular}[c]{@{}l@{}}CSS\\ $[\ell_1 \text{ BP}]$ \end{tabular}  & \begin{tabular}[c]{@{}c@{}}0.19\\ (0.02)\end{tabular} & \begin{tabular}[c]{@{}c@{}}0.91\\ (0.41)\end{tabular} & \begin{tabular}[c]{@{}c@{}}12.22\\ (0.26)\end{tabular} & \begin{tabular}[c]{@{}c@{}}0.88\\ (0.04)\end{tabular} \\ \hline
\begin{tabular}[c]{@{}l@{}}SS\\ $[\ell_0 \text{ BP}]$\end{tabular}   & \begin{tabular}[c]{@{}c@{}}0.34\\ (0.04)\end{tabular} & \begin{tabular}[c]{@{}c@{}}0.02\\ (0.02)\end{tabular} & \begin{tabular}[c]{@{}c@{}}9.90\\ (0.50)\end{tabular}  & \begin{tabular}[c]{@{}c@{}}0.85\\ (0.03)\end{tabular} & \begin{tabular}[c]{@{}l@{}}SS\\ $[\ell_1 \text{ BP}]$\end{tabular}   & \begin{tabular}[c]{@{}c@{}}0.19\\ (0.02)\end{tabular} & \begin{tabular}[c]{@{}c@{}}1.07\\ (0.36)\end{tabular} & \begin{tabular}[c]{@{}c@{}}12.24\\ (0.23)\end{tabular} & \begin{tabular}[c]{@{}c@{}}0.86\\ (0.03)\end{tabular} \\ \hline
\end{tabular}
}
\caption{\small  Performance of SS, sparsity CSS (CSS SPS), and greedy FSSS using the base procedures $\ell_0$-regression and Lasso on synthetic datasets. For comparison, we also include results from the base procedures without any stability method. Each value is averaged over 500 independent experiments. Each entry gives the result at the $s_0$ value that gives the lowest test MSE for each method. Square brackets after each method (e.g., $\ell_0$ BP) indicate the corresponding base procedure. Standard deviations are presented inside parentheses.} 
\label{tab:performance}
\end{table}

\begin{figure}[!ht]
    \centering
    \includegraphics[width=0.7\textwidth]{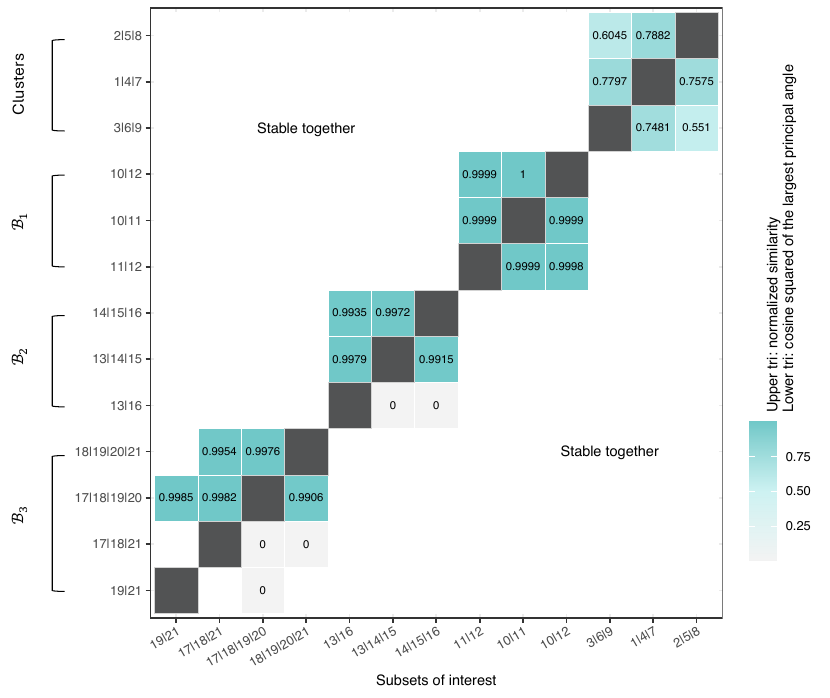}
    \caption{\small A tile plot for several subsets of interest in the synthetic data. 
    For any two subsets $S_1$ and $S_2$, the upper (and lower) triangular matrix in the tile plot displays the normalized similarity $\bar{\tau}(S_1, S_2)$ (and the cosine squared of the largest principal angle $\tilde{\tau}(S_1, S_2)$) for subsets with $\pi(S_1 \cup S_2) < 0.7$. Blank entries indicate subset pairs that is stable together. }
    \label{fig:syn_tile}
\end{figure}

In addition to better performance on the metrics described earlier, in highly correlated settings, FSSS can produce multiple stable models, whereas CSS and SS produce only a single model. 
We apply Algorithm~\ref{algo:all-path} on one training dataset with the parameters $K = 100$, $s_0 = 35$ and $\alpha = 0.7$. We denote the collection of selection sets by $\mathfrak{S} = \{S_i\}_{i=1}^{100}$. Notably, each selection set includes exactly one feature per signal cluster, $|\mathcal{K}_l|$ features from blocks $\mathcal{B}_l$ ($l=1,2,3$), but no weak signals $X_{22}, \dots, X_{26}$ or noise features. In other words, these sets demonstrate a form of highly correlated selection consistency within $\{\mathcal{C}_k\}_{k\in\mathcal{K}}$ and $\{\mathcal{B}_l\}_{l=1}^3$, and span similar subspaces through diverse feature combinations, making them equally valid for both interpretation and prediction. 

Figure~\ref{fig:syn_tile} presents pairwise similarities for several subsets of interest, where highly similar sub-models can be regarded as substitutes for one another. The similarity metrics $\bar{\tau}(S_1, S_2)$ and $\tilde{\tau}(S_1, S_2)$ are defined in Section~\ref{sec:similarity}. We evaluate the similarity between two sub-models $S_1$ and $S_2$ only when they are not jointly stable, that is, when $\pi(S_1 \cup S_2) < \alpha$; otherwise, the two sub-models can be combined into a single selection set and are therefore not substitutable. We observe that sub-models within clusters or the same block $\mathcal{B}_l$ exhibit high normalized similarity values, as shown in the upper triangular portion of the matrix; moreover, $\bar{\tau}(S_1, S_2)$ increases as perturbation direction length decreases. The lower triangular portion, which reports $\tilde{\tau}(S_1, S_2)$, is more conservative and reduces to zero whenever $|S_1| \neq |S_2|$.

\subsection{Real data application: gene expression in breast cancer}
\label{sec:real_data}
We consider a publicly available gene expression dataset from the Gene Expression Omnibus with accession code GSE2990~\citep{sotiriou2006gene}. The dataset contains microarray measurements from $n = 189$ breast cancer tumor samples across 22,283 probes. We filtered out probes with an IQR less than 1.5 and those lacking annotation, resulting in $p = 1111$ probes. These remaining expression levels serve as the explanatory variables after centralization. The aim of this analysis is to explore genes potentially associated with breast cancer prognosis. Following the results of \cite{wu2020estrogen}, we use the centralized ESR1 gene expression as the response variable $y$, given its established relevance to ER+ breast cancer outcomes.

We use $\ell_0$-regression as the base procedure for each stability-based method (greedy FSSS, CSS, and SS). For model selection and assessment, we split the dataset into three folds: subsampling (using $B = 200$), feature selection, and coefficient estimation are performed on the first fold (size $n_{\text{fitting}} = 143$), validation is performed on the second fold (size $n_\text{valid} = 16$), and test MSE is assessed using the third fold (size $n_\text{test} = 30$). Both the stability threshold $\alpha$ and the clustering cutoff height $h$ (for CSS) are chosen based on validation MSE from the sets $\{0.7, 0.75, 0.8, 0.85, 0.9, 0.95\}$ and $\{0.1, 0.3, 0.5\}$, respectively. Performance are evaluated on the test fold. To further evaluate the output stability metric, we bootstrapped $M=50$ samples from the first fold.

\begin{figure}[!ht]
    \centering
    \includegraphics[width=\textwidth]{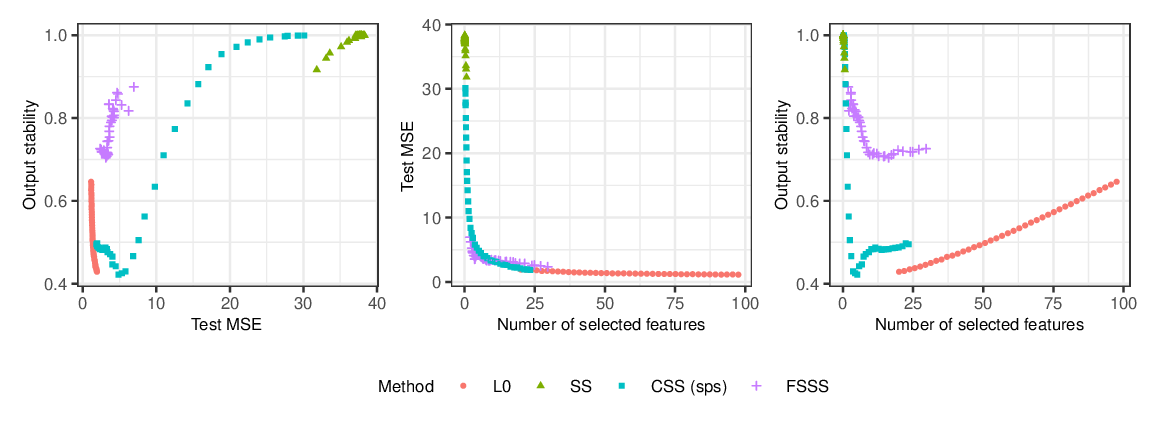}
    \caption{\footnotesize Performance of $\ell_0$-regression, SS, sparsity (sps) CSS, and greedy FSSS on the breast cancer gene expression data. The subsampling uses $\ell_0$-regression as the base procedure. 
    }
    \label{fig:breast_perform}
\end{figure}

Figure~\ref{fig:breast_perform} compares the test MSE, the number of selected features, and output stability for each stability procedure, as well as for the base procedure, across different choices of the tuning parameter $s_0$. We observe that FSSS is the only method that achieves both high output stability and a nontrivial number of selected features, while effectively stabilizing the base procedure and keeping test MSE low. In contrast, SS seldom selects any feature, whereas CSS cannot achieve low MSE and and high output stability simultaneously. 

Next, in order to obtain multiple selection sets, we apply Algorithm~\ref{algo:all-path} on the entire data with parameters $K = 50$, $s_0 = 50$ and $\alpha = 0.7$. 
All selection sets have stability $\pi\in[0.7, 0.747]$, and they appear to be scientifically meaningful. As an example, one selection set we obtain includes the following six features: \textit{ALDH3B2}, \textit{IGK$_\text{3}$}, \textit{CEP55}, \textit{S100A14}, \textit{ADAMTS5} and \textit{SETD8}. Among these, \textit{CEP55} and \textit{SETD8} are strongly associated with breast cancer, with \textit{CEP55} linked to higher tumor grade and poor prognosis, and \textit{SETD8} promoting proliferation via histone methylation \citep{sinha2019strategic, huang2017monomethyltransferase}. Additional selection sets are provided in  \ref{sec:sup-real-data}.

\begin{figure}[!ht]
    \centering
    \includegraphics[width=0.9\textwidth]{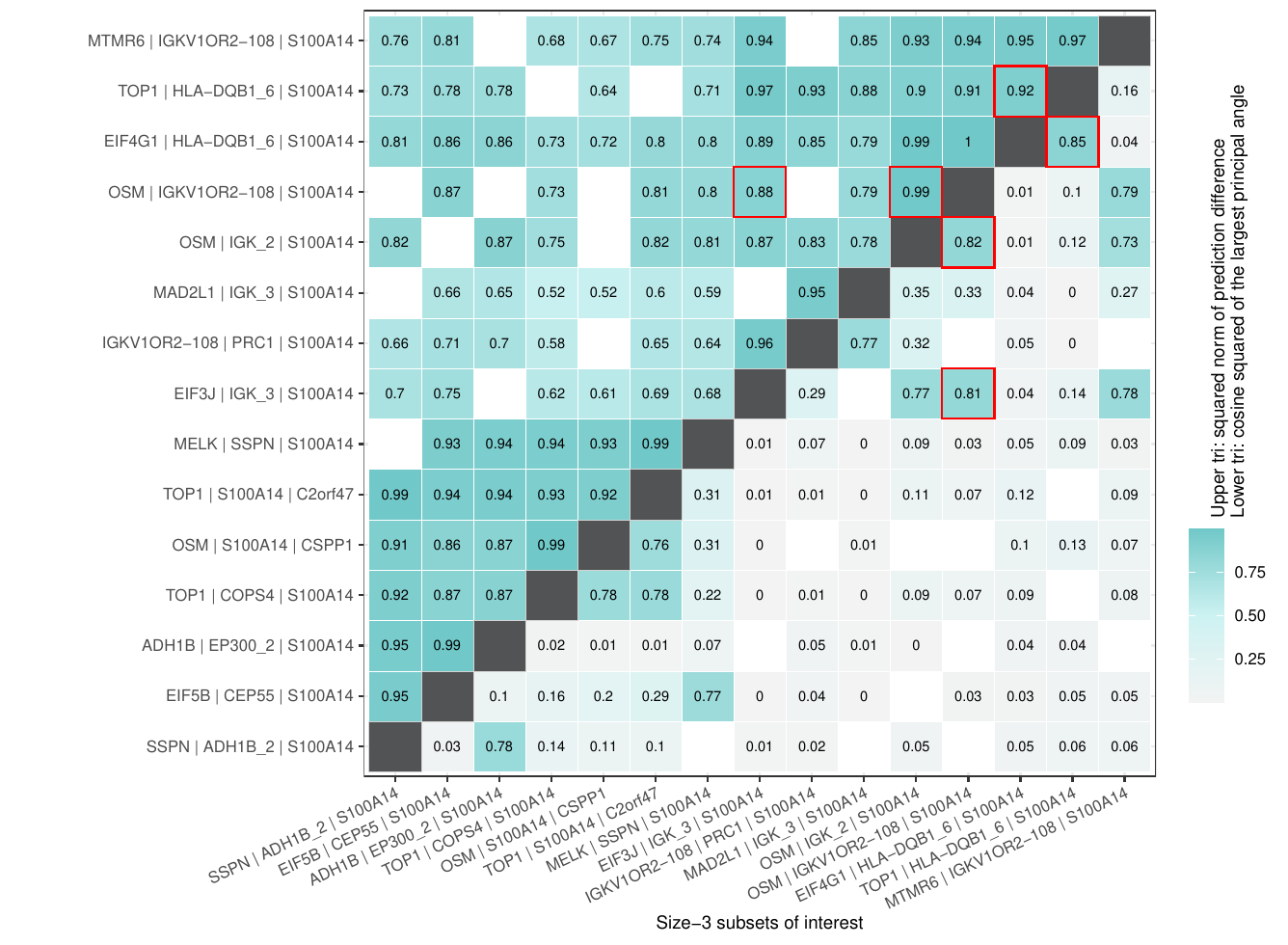}
    \caption{\small A tile plot of several size-3 subsets from the FSSS-selected models in the real dataset.
    For any two sets $S_1$ and $S_2$, the upper  triangular matrix in the tile plot displays $\tau^y(S_1, S_2)$ as introduced in Remark~\ref{rem:y-similarity}, and the lower triangular displays the similarity $\tilde{\tau}(S_1, S_2)$ in \eqref{eqn:conservative_similarity}. Blank entries indicate subset pairs that is stable together. The entries highlighted with red frames correspond to pairs with $\tilde{\tau}(S_1, S_2) \geq 0.8$.}
    \label{fig:real_tile}
\end{figure}

We identify substitutable predictive models among the FSSS-selected models by enumerating all size-3 subsets. Figure~\ref{fig:real_tile} shows pairwise similarities among several highly substitutable subsets using two metrics. The upper triangle reports $\tau^y(S_1,S_2)$ from Remark~\ref{rem:y-similarity}, which compares predictability for the observed response, while the lower triangle shows $\tilde{\tau}(S_1,S_2)$ from \eqref{eqn:conservative_similarity}, which compares worst-case predictability under arbitrary shifts in $y$.

We make several observations based on Figure~\ref{fig:real_tile}. First, although all subsets are $\alpha$-stable, they are not necessarily similar in terms of column spaces. Indeed, substantially different directions can still be strongly aligned with the average projection matrix $\P_{\rm avg}$. As a result, for some pairs, we see some small values for both similarity metrics. Second, $\tilde{\tau}(S_1,S_2)$ is consistently more conservative than $\tau^y(S_1,S_2)$, suggesting that models can agree in prediction for the observed response (high $\tau^y$) yet diverge when the response changes (low $\tilde{\tau}$). Finally, sub-models with large $\tilde{\tau}$ values (red frames in Figure~\ref{fig:real_tile}) show similar predictive for all possible future response data. For example, {\textit{EIF4G1}, \textit{HLA-DQB1$_6$}, \textit{S100A14}} and {\textit{TOP1}, \textit{HLA-DQB1$_6$}, \textit{S100A14}} have $\tilde{\tau}=0.85$, indicating comparable utility for breast cancer prognosis; they differ only by \textit{EIF4G1} and \textit{TOP1}, whose correlation is 0.926. Models differing by two feature pairs, such as {\textit{EIF3J}, \textit{IGK$3$}, \textit{S100A14}} and {\textit{OSM}, \textit{IGKV1OR2${108}$}, \textit{S100A14}} ($\tilde{\tau}=0.81$), also involve highly correlated features (with correlation coefficients 0.91 and 0.90, respectively). Overall, substitutable models in this breast cancer dataset consist of highly correlated gene pairs, suggesting that gene expression features exhibit a clustering structure rather than a more complex dependency pattern.

\section{Discussion}
We proposed a subspace framework to quantify similarity and stability of models in highly correlated settings. We also presented an algorithm that leverages our subspace framework to return multiple stable feature sets and described theoretical control on false positive error for this procedure. 

Our work opens several avenues for future research. First, while repeated runs of Algorithm~\ref{algo:all-path} can, in principle, yield all stable models, this process can be inefficient when the number of features is large. It would be interesting to develop more efficient approaches to identify the number of stable models, and to enumerate them. Second, extending our methodology to generalized additive models with splines---where the response variable depends linearly on unknown spline functions of the predictors---would be of practical interest. Finally, building on the previous point, a general nonlinear version of our framework may be possible by measuring similarity via non-parametric canonical correlation analysis.

\section{Funding}
This work was partially supported by funding from NSF under grant DMS-2413074.


\addtolength{\textheight}{-.2in}%

\singlespacing
\small
\bibliography{refs.bib}

\newpage
\singlespacing
\appendix
\small

\renewcommand{\thesection}{\Alph{section}}

\setcounter{section}{0}
\section{Our subspace framework and comparisons to other perspectives}

\subsection{Connection to model selection over partially ordered sets}
\label{sec:sup-poset}
There is a strong connection between our proposed subspace framework and the model selection framework proposed in \cite{Taeb2023ModelSO}. We first present the perspective in \cite{Taeb2023ModelSO}, and then describe how our proposed model selection framework is a special case. \\

\textbf{Framework in \cite{Taeb2023ModelSO}}: Variable selection models are organized as a partially ordered set (poset). The poset is structured hierarchically as shown in Figure \ref{fig:poset}, starting from its least element (the empty set $\varnothing$), and each model on it is formed by including a new feature to a model at the preceding level. This poset inherits a rank function (formally known as a graded poset) that captures the complexity of a model. In the variable selection poset, the rank of a model $S$ is its cardinality $|S|$, which is also the length of each path from $\varnothing$ to $S$. 

\begin{figure}[!ht]
    \centering
    \includegraphics[width=0.4\textwidth]{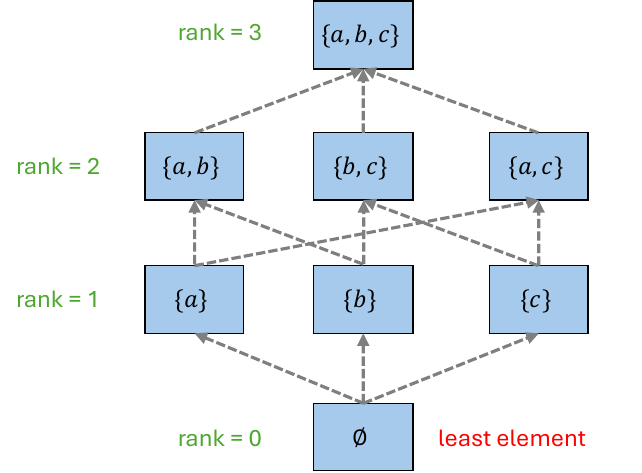}
    \caption{\small A variable selection poset over features $\{a,b,c\}$.} 
    \label{fig:poset}
\end{figure}

To assess the similarity between two models $S$ and $\widetilde{S}$, a generic ``similarity valuation'' metric $\rho(S, \widetilde{S})$ was proposed in \cite{Taeb2023ModelSO}. This metric can take different forms depending on the specific problem domain. Based on this similarity measure, the amount of true positives, false positive error, and false negative error are defined as:
\begin{align}
\begin{split}
    \mathrm{TP}(\widehat{S}, S^\star) &:= \rho(\widehat{S}, S^\star), \\
    \mathrm{FPE}(\widehat{S}, S^\star) &:= \rank(\widehat{S}) - \rho(\widehat{S}, S^\star),\\
        \mathrm{FNE}(\widehat{S}, S^\star) &:= \rank({S}^\star) - \rho(\widehat{S}, S^\star).
\end{split}
\label{eq:posets_fpe}
\end{align}
To obtain a model with small false positive error, the poset framework views any selected model $\widehat{S}$ as one that is obtained by following a path that begins at the empty set, and sequentially includes one variable at each level until $\widehat{S}$ is reached. Suppose that the path forming $\widehat{S}$ is given by $(S_0 = \varnothing, S_1, \dots, S_{t-1}, S_t = \widehat{S})$, then the $\mathrm{FPE}(\widehat{S}, S^\star)$ can be written as the telescoping sum
$$
\mathrm{FPE}(\widehat{S}, S^\star) = \sum_{i=0}^{t-1} 1 - \left[ \rho(S_{i+1}, S^\star) - \rho(S_{i}, S^\star) \right],
$$
where the term $1 - [\rho(S_i, S^\star) - \rho(S_{i-1}, S^\star)]$ represents the ``additional false discovery'' incurred by moving from the model $S_{i-1}$ to the model $S_i$ on the poset. Based on the telescoping sum decomposition, a greedy algorithm based on subsampling was proposed in \cite{Taeb2023ModelSO} to control this additional FD term, by guaranteeing the following at each step from $S_i$ to $S_{i+1}$:
\begin{align}\label{eq:poset-cond}
    \frac{1}{B} \sum_{\ell=1}^B {\rho(S_{i+1}, \widehat{S}^{(\ell)}) - \rho(S_i, \widehat{S}^{(\ell)}) } \geq \alpha ,
\end{align}
where $\{\widehat{S}^{(\ell)}\}_{\ell = 1}^B$ are selection sets obtained from subsampling. The left-hand-side of \eqref{eq:poset-cond} represent the additional similarity that $S_{i+1}$ exhibits with $\{\widehat{S}^{(\ell)}\}_{\ell = 1}^B$ relative to $S_i$, and is thus a measure of stability associated with the move from model $S_{i}$ to $S_{i+1}$. The quantity $\alpha \in (0,1)$ represents a stability threshold. The algorithm described above is intuitively approximating $1 - \left[ \rho(S_{i+1}, S^\star) - \rho(S_{i}, S^\star) \right]$ with $\frac{1}{B} \sum_{\ell=1}^B {\rho(S_{i+1}, \widehat{S}^{(\ell)}) - \rho(S_i, \widehat{S}^{(\ell)}) }$. \\

\textbf{Our proposed framework:} We take $\rho({S}, \widetilde{S}) = \tr(\P_{{S}} \P_{\widetilde{S}})$, where we use the notation $\P_{S} := \P_{\mathrm{col}(X_S)}$. As we describe in Section 2.2 of the main paper, this choice of $\rho$ is motivated by: $(i)$ the subspace $\mathrm{col}(X_S)$ is a natural and useful representation of $S$ in highly correlated settings; and $(ii)$ motivated by \cite{taeb2020false}, the trace inner product of projection matrices onto subspaces is a natural measure of similarity between two subspaces. Note that $\tr(\P_{{S}} \P_{\widetilde{S}})$ formally qualifies as similarity valuation (as defined in \cite{Taeb2023ModelSO}) as long as the matrices $X_S$ and $X_{\widetilde{S}}$ both consist of linearly independent columns. 

With our choice of similarity valuation $\rho$, Definition~\ref{def:fd-power} is a special case of \eqref{eq:posets_fpe}. Furthermore, our proposed algorithm FSSS is similar to the algorithm in \cite{Taeb2023ModelSO} in that it searches for a selection set by growing a path in this poset, starting from the empty set $S_0 = \varnothing$, and adding one element at a time to reach $S_i$. At each step, we require $\tr((\P_{S_i} - \P_{S_{i-1}}) \P_{\rm avg}) \geq \alpha$, which is a special variate of (\ref{eq:poset-cond}). The final model is denoted by $S_t$. As a result, the FPE for $S_t$ can be controlled as
\begin{align*}
\begin{split}
    \mathrm{FPE}(\widehat{S}, S^\star) 
    &= \sum_{i=1}^{t-1} 1 - \left[ \tr(\P_{S_{i+1}} \P_{S^\star}) - \tr( \P_{S_i} P_{S^\star} ) \right] \\
    &\approx \sum_{i=1}^{t-1} 1 - \left[ \tr(\P_{S_{i+1}} \P_{\rm avg}) - \tr( \P_{S_i} P_{\rm avg} ) \right] \leq |S_t| (1 - \alpha),
\end{split}
\end{align*}
where the approximation uses $\P_{\rm avg}$ as an estimate for $\P_{S^\star}$. This suggests that $\alpha$-stable models have small false positive error, especially when $\alpha \approx 1$. We provide rigorous false positive error guarantees in Section \ref{sec:theory}.

\vspace{0.2in}
\subsection{Extension to generalized linear models}
\label{sec:sup-glm}
Let $X\in\mathbb{R}^{n\times p}$ denote the design matrix, $\beta\in\mathbb{R}^p$ denote the coefficient, and $y\in\mathbb{R}^n$ denote the response variable. For generalized linear models (GLMs), the conditional expectation of the response is given by $\mathbb{E}[y \mid X] = \mu = g^{-1}(X\beta)$, where $g$ is a link function. Due to the linear component $X\beta$ in GLMs, we can define $\mathrm{TP}, \mathrm{FPE}, \mathrm{FNE}$ exactly like Definition~\ref{def:fd-power} and stability $\pi$ exactly like Definition~\ref{defn:stability}. Additionally, Algorithm 1 (FSSS) can be applied without any modifications, provided that the base procedure is compatible with the response type--for example, using $\ell_0$-penalized logistic regression for binary response $y \in \{0,1\}^n$.

However, when assessing substitutability via subspaces, the choice of link function and the type of response variable becomes relevant. In particular, when the model is not linear, $\P_{\mathrm{col}(X_S)} (y)$ does not represent the component of the response captured by the features in $S$. One can instead estimate $g(\mu)$---the actual linear combination of $X\beta$ in the GLM framework---and replace all terms involving $\P_{\mathrm{col}(X_S)} (y)$ in the substitutability metric with $\P_{\mathrm{col}(X_S)} (g(\widehat{\mu}))$. A possible strategy for estimating $g(\mu_i) = g(\mathbb{E}[y_i \mid X^{(i)}])$ is to use a non-parametric $k$-nearest neighbors (KNN) approach. For the $i$-th observation $X^{(i)}$, we identify its $k$ nearest neighbors, denoted by $N(i)$, and estimate $\widehat{\mu}_i$ as $\frac{1}{ |N(i)| } \sum_{j\in N(i)} y_j  $. When the link function is identity, we set $k = 1$ so that the estimate of $\mu_i$ is simply each $y_i$ itself. In contrast, when the link function is the logit function, $k$ must be necessarily greater than 1 in order for $\widehat{\mu}\in(0,1)$, as values of 0 or 1 are not valid under the logistic transformation.

\vspace{0.2in}
\subsection{Why the proposal in \texorpdfstring{\citet{g2013false}}{G'Sell et al. (2013)} is not suitable for assessing false positive error}
\label{sec:sup-g2013false}
First, we define how \citet{g2013false} measure the false positive error. Letting $\widehat{S}$ be a selected set of features and supposing the linear mechanism $y = X\beta^\star + \epsilon$ with $X \in \mathbb{R}^{n \times p}$ and $y \in \mathbb{R}^n$, \citet{g2013false} project $X\beta^\star$ onto the span of the columns of $X$ in $\widehat{S}$ (denoted by $X_{\widehat{S}}$) and obtain a projected mean $X_{\widehat{S}}\widehat{\beta}_{\widehat{S}}$. Then, the false positive error (dubbed false variable selection) associated with the selection set $\widehat{S}$ is measured to be the number of zeros in $\widehat{\beta}_{\widehat{S}}$. 

While in some highly correlated settings, this new metric provides more reasonable values than the standard notion, it can yield smaller false positive errors than one might expect. For example, suppose $S^\star = \{1\}$, and $\widehat{S} = \{2\}$, where $X_2$ is a null feature that is nearly, but not completely, orthogonal to $X_1$. Then, a simple calculation shows that $\widehat{\beta}_{\widehat{S}}$ is a scalar that will be nonzero. Hence, \citet{g2013false} would claim that there is no false positive error with this selection set. However, including a null feature instead of a nearly orthogonal signal feature should incur a false positive close to one, which is what our subspace definition outputs.

The shortcoming of \citet{g2013false}'s measure in the previous example arises from its sharp binary decision: namely, that if $\widehat{\beta}_{\widehat{S}} = 0$, then we incur a false positive error of one; otherwise, if $\widehat{\beta}_{\widehat{S}} \neq 0$, we incur no false positive errors. As a result, as long as there is any nonzero correlation between $X_1$ and $X_2$, \citet{g2013false} measures no error. On the other hand, our subspace definition of false positive error varies smoothly depending on the amount of correlation among the features: $\mathrm{FPE}(\widehat{S},S^\star) = 1$ when $X_1$ and $X_2$ are orthogonal, and $\mathrm{FPE}(\widehat{S},S^\star)$ decreases gradually as the correlation level increases, with $\mathrm{FPE}(\widehat{S},S^\star) = 0$ when $X_1$ and $X_2$ are perfectly correlated.

\vspace{0.2in}
\subsection{A potential formulation for similarity incorporating \texorpdfstring{$y$}{y}}
\label{sec:sup-y-similarity}
We propose a response-aware definition of similarity that involves $y$.
\begin{definition}
    Let $S_1, S_2 \subseteq \{1, 2, \dots, p\}$ be two sets of features, and let $\mathcal{B}$ be a perturbation set for the normalized response variable $y$. Then the response-aware similarity between $S_1$ and $S_2$ is defined as:
    $$
    \tau_{\mathcal{B}}(S_1, S_2) := 1 - \sup_{y'\in \mathcal{B}, \ \|y'\|_2 = 1 } \left\|\P_{X_{S_1}} y' - \P_{X_{S_2}} y' \right\|_2^2.
    $$
\end{definition}

We note that $\tau_\mathcal{B}(S_1, S_2) \in [0,1]$.
Here are some choices of the perturbation set $\mathcal{B}$ and the resulting $\tau_{\mathcal{B}}$ metric:
\begin{enumerate}[label = (\arabic*)]
    \item Take $\mathcal{B}_1 = \{y': y' = y / \|y\| \}$. The resulting metric is given by:
    $$
    \tau_{\mathcal{B}_1}(S_1, S_2) = 1 - \frac{ \left\| P_{X_{S_1}} y - P_{X_{S_2}} y \right\|^2_2 }{\|y\|_2^2}.
    $$
    The perturbation set $\mathcal{B}_1$ contains the singleton normalized response, yielding a metric that quantifies the distance between the best linear maps from $y$ to $X_{S_1}$ and $X_{S_2}$, respectively. One can show that
    $$
    \tau_{\mathcal{B}_1}(S_1, S_2) = 1 - \sum_{j=1}^{\max\{|S_1|, |S_2|\}} \mathrm{corr^2}(y, v_j) \cdot \sin^2\theta_j  . 
    $$
    where $v_j$'s are the eigenvectors of $P_{X_{S_1}} - P_{X_{S_2}}$, and $\theta_j$'s are the corresponding principal angles between $\mathrm{col}(X_{S_1})$ and $\mathrm{col}(X_{S_2})$. Specifically, we order $(v_j, \theta_j)$ such that $\theta_1 \leq \theta_2 \leq \dots \leq \theta_{\max\{|S_1|, |S_2|\} }$. In other words, $1 - \tau_{\mathcal{B}_1}(S_1, S_2)$ can be expressed as a linear combination of squared sines of the principal angles, where the weights correspond to the correlations between $y$ and the principal vectors $v_j$. Consequently, the response-aware similarity $\tau_{\mathcal{B}_1}(S_1, S_2)$ can remain large even when some principal angles are large, as long as the corresponding directions are not predictive of $y$. Note that $\tau_{\mathcal{B}_1}$ is exactly $\tau^y$ defined in Remark \ref{rem:y-similarity}.
    
    \item Take $\mathcal{B}_2 = \{y' : \|y'\|_2 = 1\} $. The resulting metric is given by:
    $$
    \tau_{\mathcal{B}_2}(S_1, S_2) = 1 - \sin^2 \theta_{\max\{|S_1|, |S_2|\} } = \cos^2\theta_{\max\{|S_1|, |S_2|\}}.
    $$
    The perturbation set $\mathcal{B}_2$ contains all directions, yielding the most conservative metric that ignores $y$. The cosine squared of the largest principal angle becomes zero whenever $|S_1| \neq |S_2|$. Note that $\tau_{\mathcal{B}_2}$ reduces to $\tilde{\tau}$ defined in Section~\ref{sec:similarity}.
    
    \item Take $\mathcal{B}_3 = \{y': \angle\langle y / \|y\|_2, y'/\|y'\|_2 \rangle \leq \eta \}$. Suppose $v_j$'s are the eigenvectors of $P_{X_{S_1}} - P_{X_{S_2}}$. Let $L = \max \{|S_1|, |S_2|\}$. The resulting metric is given by:
    \begin{enumerate}[label = (\roman*)]
            \item If $\mathrm{corr}(v_L, y/\|y\|_2) \geq \cos\eta$,
            then $\tau_{\mathcal{B}_3}(S_1, S_2) = 1 - \sin^2 \theta_{L}$.

            \item If $\mathrm{corr}(v_L, y/\|y\|_2) < \cos\eta$, then
            \begin{align*}
            \begin{split}
               & \tau_{\mathcal{B}_3}(S_1, S_2)  \\
            ={}& 1- \cos^2(\angle\langle v_L, y/\|y\|_2 \rangle -\eta ) \cdot \sin^2\theta_{L} - \sin^2(\angle\langle v_1, y/\|y\|_2 \rangle -\eta ) \cdot \sum_{j=1}^{L-1}\frac{  \mathrm{corr}^2(y, v_j) }{\sum_{j'=1}^{L-1} \mathrm{corr^2}(y, v_{j'}) } \cdot \sin^2 \theta_j.
            \end{split}
            \end{align*}         
    \end{enumerate}
    The perturbation set $\mathcal{B}_3$ consists of directions whose deviation from the normalized response $y$ is at most $\eta$, forming a neighborhood around the normalized $y$. The metric $1-\tau_{\mathcal{B}_3}$ can be expressed as a linear combination of the squared sines of the principal angles, with weights depending on both the correlations between $y$ and the principal vectors $v_j$ and the perturbation angle $\eta$. When $\eta = 0$, corresponding to no perturbation, $\tau_{\mathcal{B}_3}$ reduces to $\tau_{\mathcal{B}_1}$. When $\eta = \pi/2$, so that the perturbation spans the entire space, $\tau_{\mathcal{B}_3}$ reduces to $\tau_{\mathcal{B}_2}$.
    
\end{enumerate}

\vspace{0.2in}
\section{Additional theoretical properties of FSSS}

\emph{Notation}: For simplicity, we write the norm $\|x\|_{\ell_q}$ as $\|x\|_q$. We adopt the notation in Section \ref{sec:algorithm} of the paper, namely that for a vector $w \in \mathbb{R}^n$, $\P_{w} = \P_{\mathrm{span}(w)}$ and for any subset $S \subseteq \{1,\dots,p\}$, $\P_{X_S} = \P_{\mathrm{col}(X_S)}$.

This section is organized as follows: Section \ref{sec:sup-cluster} presents theoretical properties of $\ell_0$-penalized regression as a base procedure under the clustering setup; Section \ref{sec:sup-bpquality}, also under the clustering setup, provides further discussion on the quality terms $\gamma_\mathcal{W}$ and $\gamma_\mathcal{U}$ introduced in Section \ref{sec:theory}; Section \ref{sec:sup-prediction-similarity} discusses the similarity in predictability among the ``equally good models'' in $\mathcal{S}$; and finally, Section \ref{sec:sup-complex} extends our theoretical analysis to setups with complex dependency structures.

\vspace{0.2in}
\subsection{\texorpdfstring{$\ell_0$}{l\_0}-penalized base procedure}
\label{sec:sup-cluster}
While any feature selection algorithm can serve as the base procedure for FSSS, an $\ell_0$-penalized regression is particularly recommended for selecting highly correlated variables. Given a design matrix $X \in \mathbb{R}^{n \times p}$ and a response variable $y \in \mathbb{R}^n$, the solution is given by
\begin{align}\label{eq:lasso_s0}
    \widehat{\beta} \in \arg\min_{\beta\in\mathbb{R}} \|y - X\beta\|_2^2, \quad
\st \|\beta\|_0 \leq s_0.
\end{align}

The following theorem is built on the clustering setup described in Section \ref{sec:clustering_setup}.

\vspace{0.1in}
\begin{theorem} [Base procedure] \label{prop:clust-consist-s0}
    Suppose that $\max_{j\in \mathcal{K}}\mathrm{Corr}(X_j, \epsilon) \leq \mathcal{O}(\sqrt{\log p / n})$ and $\max_{j\in \mathcal{V}}\mathrm{Corr}(\delta_j, \epsilon) \leq \mathcal{O}(\sqrt{\log p / n})$. Furthermore, assume that $|\mathcal{C}_k| = D$ for all $k\in\mathcal{K}$, and
    \begin{align}\label{eq:beta-min-cluster}
    \begin{split}
        \min_{j\in S^\star} \left|\beta^\star_j\right|^2 > 
        \frac{D\eta_1^2}{1 + \eta_1^2}  \|\beta^\star\|_2^2 + \mathcal{O}\left( s_0 D^2 \sqrt{\log p /n} \right) (\|\beta^\star\|_1 + \|\epsilon\|_2)^2.
    \end{split}
    \end{align}
    Then when $s_0 \geq s^\star$,
    \begin{enumerate}[label = \normalfont(\arabic*)]
        \item The solution $\widehat{\beta}$ to (\ref{eq:lasso_s0}) does not miss any signal groups: there must exist $\widetilde{S} \in \mathcal{S}$ such that $\widetilde{S} \subseteq \supp(\widehat{\beta})$.

        \item When $s_0 = s^\star$, the solution $\widehat{\beta}$ achieves cluster feature selection consistency: $\supp(\widehat{\beta}) \in\mathcal{S}$.

        \item When $s_0 = s^\star$ and furthermore
        \begin{align}\label{eq:beta-min-cluster2}
        \frac{\eta_1^2}{1 + \eta_1^2}  \|\beta^\star\|_2^2 
        > 
        \mathcal{O}\left( s^\star \sqrt{ \log p /n} \right) (\|\beta^\star\|_1 + \|\epsilon\|_2)^2,
        \end{align}
        the solution $\widehat{\beta}$ achieves exact feature selection consistency: $\supp(\widehat{\beta}) = S^\star$.
    \end{enumerate}
\end{theorem}

We prove Theorem~\ref{prop:clust-consist-s0} in Section \ref{sec:sup-proofthm3}. Theorem \ref{prop:clust-consist-s0} demonstrates the effectiveness of $\ell_0$-regression in the clustering setup. To be specific, under the beta-min condition (\ref{eq:beta-min-cluster}), when $s_0 \geq s^\star$, $\ell_0$-regularization manages to select at least one feature per signal group. Under the more restrictive condition $s_0 = s^\star$, exactly one feature is selected per signal group, and no noise clusters can be selected. However, exact feature selection consistency may require a much larger sample size, as indicated by the beta-min condition (\ref{eq:beta-min-cluster2}).

Note that Theorem \ref{prop:clust-consist-s0} relies on the conditions $\max_{j\in \mathcal{K}}\mathrm{Corr}(X_j, \epsilon) \leq \mathcal{O}(\sqrt{\log p / n})$ and $\max_{j\in \mathcal{V}}\mathrm{Corr}(\delta_j, \epsilon) \leq \mathcal{O}(\sqrt{\log p / n})$. Under the linear model $y = X\beta^\star + \epsilon$ with $\epsilon \sim \mathcal{N}(0, \sigma^2 / nI_n)$, these conditions hold with high probability, as suggested by the concentration inequality $\mathbb{P}\left[ \max_{j\in \mathcal{K}}\mathrm{Corr}(X_j, \epsilon) \leq t, \max_{j\in \mathcal{V}}\mathrm{Corr}(\delta_j, \epsilon) \leq t  \right] \geq 1 - ap \exp\{-bnt^2\}$ for some $a,b>0$. As a result, and as noted in Remark \ref{rem:l0-base}, the $\ell_0$-regularized regression is very likely to satisfy the condition (i) in Theorem \ref{thm:1}, and hence serves as a good candidate of base procedure for FSSS to achieve promising results.

\vspace{0.2in}
\subsection{Quality of base procedure}
\label{sec:sup-bpquality}

Definition \ref{def:quality_base_procedure} characterizes the quality of a base procedure via the terms $\gamma_\mathcal{W}$ and $\gamma_\mathcal{U}$. In this section, we analyze the behavior of these quantities. \\

\textbf{Regarding the term $\gamma_\mathcal{W}$}: The quality term appeared in Theorem \ref{thm:1}. Recall that the collection of all ``noise directions'' is defined as $\mathcal{W} := \{ \delta_j : j\in\mathcal{V} \} \cup \bigcup_{j\in \mathcal{K} \cap N^\star} \{X_j\}$. We also define the index set $W := \mathcal{V} \cup (\mathcal{K} \cap N^\star)$, and associate each $j \in W$ with a direction $w_j$, where $w_j = \delta_j$ for $j \in \mathcal{V}$ and $w_j = X_j$ for $j \in \mathcal{K} \cap N^\star$. That is, we write $\mathcal{W} = \{ w_j \}_{j \in W}$. We further denote $w_j = X_j$ for any $j \in [p]\setminus W$.
With these notations in place, we introduce two standard assumptions that are commonly used in the study of stability-based methods, to enable a more refined analysis of $\gamma_\mathcal{W}$.

\begin{assumption}[Better than random guessing]\label{assum:better}
    For all subsamples $\ell\in \{1, \dots,B\}$,
    $$
    \frac{ \sum_{j\in W} \mathbb{E}\left[\tr(\P_{w_j} \P_{X_{\widehat{S}^{(\ell)}}} ) \right]  }{p - s^\star}
    \leq
    \frac{ \sum_{j\in [p]\setminus W} \mathbb{E}\left[\tr(\P_{w_j} \P_{X_{\widehat{S}^{(\ell)}}} ) \right]  }{s^\star}.
    $$
\end{assumption}

\begin{assumption}[Exchangeability]\label{assum:exchange}
    The noise directions are exchangeable in the following sense: for all $j\in W$,
    $
    \inf_{T\in\mathscr{T}} \mathbb{E}\left[ \tr(\P_{w_j} \P_{X_{\widehat{S}(T)}}) \right] \equiv l ,
    $
    where $\widehat{S}(T)$ is the selection set of the base procedure on subsample $T$.
\end{assumption}

Informally, Assumption \ref{assum:better} states that the normalized power of the base procedure on signal directions exceeds its normalized false discovery rate on noise directions. Assumption \ref{assum:exchange}, on the other hand, assumes that the minimum false discoveries made by the base procedure is uniformly distributed across all the noise directions. Remark \ref{rem:assump_better} provides insights into scenarios in which these assumptions are likely to hold.

Under the two assumptions above, Proposition \ref{prop:assum-inter} below establishes an upper bound on the quality term $\sum_{j\in W} \sup_{T\in\mathscr{T}} \mathbb{E}\left[ \tr(\P_{w_j} \P_{X_{\widehat{S}(T)}}) \right]$, which serves as a proxy for $p\gamma_\mathcal{W}$. By comparing these two expressions, we find that $\gamma_\mathcal{W}$ is approximately $s_0 / p + 1/p \sum_{j\in[p] } r_j + \mathcal{O}(s_0 \sqrt{\log p / n})$. In particular, under the special case described in Remark \ref{rem:assump_better}, this quantity simplifies exactly to $s_0 / p$. 

\vspace{0.1in}
\begin{proposition}\label{prop:assum-inter}
    Let 
    $
    r_j = \sup_{T\in\mathscr{T}} \mathbb{E}\left[ \tr(\P_{w_j} \P_{X_{\widehat{S}(T)}}) \right] - l
    $ for any $j\in [p]$. Under the Assumptions \ref{assum:better}--\ref{assum:exchange}, we have
    \begin{align*}
    \begin{split}
        \sum_{j\in W} \sup_{T\in\mathscr{T}} \mathbb{E}\left[ \tr(\P_{w_j} \P_{X_{\widehat{S}(T)}}) \right] 
        \leq s_0 + \sum_{j\in [p]} r_j + \mathcal{O}\left(s_0 p \sqrt{\log p / n} \right),
    \end{split}
    \end{align*}
    and
    \begin{align*}
    \begin{split}
     &   \sum_{j\in W} \sup_{T\in\mathscr{T}} \mathbb{E}\left[ \tr(\P_{w_j} \P_{X_{\widehat{S}(T)}}) \right]^2 
     \leq \frac{s_0^2}{p} + \sum_{j\in[p]} \left[ r_j^2 + \frac{2s_0}p{ r_j} \right] + \mathcal{O}\left(s_0^2 p \sqrt{\log p / n} \right).
    \end{split}
    \end{align*}
\end{proposition}

Before moving on to the second perspective, we present an immediate corollary of Proposition \ref{prop:assum-inter}. As shown in Remark \ref{rem:assump_better}, the resulting FPE upper bound simplifies to that of stability selection, up to an additive constant factor related to the length of the perturbation directions.

\vspace{0.1in}
\begin{corollary}\label{cor:pfer}
    Let 
    $
    r_j = \sup_{T\in\mathscr{T}} \mathbb{E}\left[ \tr(\P_{w_j} \P_{X_{\widehat{S}(T)}}) \right] - l
    $ for any $j\in [p]$. Under the conditions of Theorem \ref{thm:1} and Assumptions \ref{assum:better}--\ref{assum:exchange}, for every $S$ returned by Algorithm \ref{algo:all-path}, 
    \begin{align*}
    \begin{split}
       \mathbb{E} \left[ \mathrm{FPE}(S,S^\star) \right] \leq \frac{1}{2\alpha-1} \left[ \frac{s_0^2}{p} + \sum_{j\in[p]} \left(r_j^2 + \frac{2s_0}{p} r_j \right) \right] + f(\eta_1) + \mathcal{O}\left( s_0^2 p \sqrt{\log p / n}  \right),
    \end{split}
    \end{align*}
    where 
    $
    f(\eta_1) =  [ 2(s_0 + \sum_{j\in[p]} r_j) \sqrt{\eta_1^2 / (1 + \eta_1^2) } + 4p\eta_1^2 / (1 + \eta_1^2)  ] / (2\alpha - 1) +  s_0 \eta_1^2 / (1 + \eta_1^2).
    $
\end{corollary}

Both Proposition \ref{prop:assum-inter} and Corollary \ref{cor:pfer} are proved in Section \ref{sec:sup-proofpropcor}.

\vspace{0.2in}

\textbf{Regarding the term $\gamma_\mathcal{U}$}: This term appears in Theorem \ref{thm:2}. Recall that $\mathcal{U}:= \bigcup_{k\in S^\star} \{\delta_j - \delta_l: j\neq l \in \mathcal{V}_k \} \cup \bigcup_{k\in \mathcal{K} \cap S^\star} \{\delta_j: j\in\mathcal{V}_k\} \cup \bigcup_{k\in \mathcal{K} \cap N^\star} \{ X_j: j\in\mathcal{C}_k \}$. Our heuristic approximation on $\gamma_\mathcal{U}$ is based on the following idea: when the sample size $n$ is sufficiently large, both the full data set and any subsample would approximate an idealized setting, in which the directions of representative features, perturbations, and noise are exactly orthogonal. 
Therefore, analyzing the behavior of $\gamma_\mathcal{U}$ under this idealized setup provides valuable insight into its asymptotic behavior.

In order to construct such an idealized setup, we distinguish between two key components of the FSSS procedure: \textit{index estimation} and \textit{subspace estimation}. For index estimation, we apply the base procedure to disjoint pairs of subsamples $\mathscr{T} :=\{T^{(2\ell-1)}, T^{(2\ell)}\}_{\ell=1}^{B/2}$, where $T^{(2\ell-1)} \cap T^{(2\ell)} = \varnothing$, yielding $B$ estimates $\{\widehat{S}^{(2\ell-1)}, \widehat{S}^{(2\ell)}\}_{\ell=1}^{B/2}$. For subspace estimation, we compute the average projection matrix $\P_{\rm avg} = \frac{1}{B} \sum_{\ell=1}^B \P_{\widehat{S}^{(\ell)}}$ based on a design matrix restricted to a subset of observations $\mathscr{D} \subseteq\{1, \dots, n\}$. Importantly, the sets $\{T^{(2\ell-1)}, T^{(2\ell)}\}_{\ell=1}^{B/2}$ need not be complementary subsamples of equal size, and the evaluation set $\mathscr{D}$ need not coincide with the full dataset. No intrinsic relationship is required between the subsampling sets used for index estimation and the data used for subspace estimation.

\begin{figure}[!ht]
    \centering
    \includegraphics[width=\textwidth]{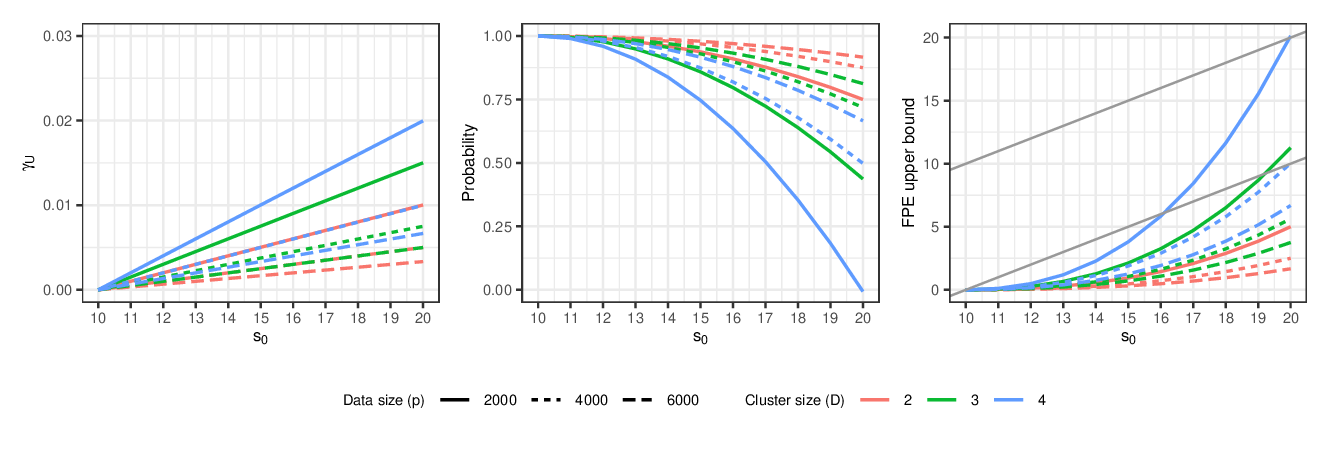}
    \caption{\small A numerical experiment illustrating the behavior of $\gamma_\mathcal{U}$ and its related quantities, with $\alpha_0 = 0.1$, $\eta_1 = 0.01$, $s^\star=10$, and $p = 2000$. Left panel: $\gamma_\mathcal{U}$ values for different $s_0$. Middle panel: the probability $1 - ( s^\star \binom{D-1}{2} + p-s^\star) \cdot \gamma_\mathcal{U}^2 / (1 - 2\alpha_0)$ for different $s_0$. Right panel: FPE upper bound for different $s_0$, where the gray lines indicate the upper bound and lower bound for $\ell_0$-base procedure respectively.} 
    \label{fig:assump_clust}
\end{figure}

Building on this generalization, we now assume that the subspace estimation indices $\mathscr{D}$ and the index estimation indices $\mathscr{T}$ form a partition of $\{1, \dots, n\}$. Within each partition, we further assume that the components $X_\mathcal{K}$, $\delta_\mathcal{V}$ and $\epsilon$ are perfectly orthogonal; that is, $Z_i \perp Z_j$ for all $Z_i\neq Z_j \in X_\mathcal{K} \cup \delta_\mathcal{V}\cup \{\epsilon\}$. Additionally, we assume that each cluster has size $D > 1$. In this setting,
\begin{align}\label{eq:assump_cluster}
    \gamma_\mathcal{U} = \max\left\{ \frac{s_0 - s^\star}{p - s^\star}, \quad
    1 - \prod_{i=1}^{D-1} \frac{p-s_0 - i}{ p-s^\star-i }
    \right\},
\end{align}
which is about $s_0/p$ when $s_0$ and $p$ are sufficiently large (see Section \ref{sec:sup-proofeq} for a proof). Figure \ref{fig:assump_clust} shows a numerical experiment examining the behavior of $\gamma_\mathcal{U}$, the probability of $\mathfrak{S} \subseteq \mathcal{S}$, and the FPE upper bound. We observe a meaningful region where FSSS achieves uniformly lower FPE upper bounds than the base procedure, highlighting the benefits of subsampling. \\

As a final remark, we point out that this idealized setting can also be applied to $\gamma_\mathcal{W}$ analysis. Remark \ref{rem:assump_better} presents a special scenario, in the same spirit of the perfectly orthogonal setting, under which the Assumptions \ref{assum:better}--\ref{assum:exchange} hold.

\vspace{0.1in}
\begin{remark}\label{rem:assump_better}
    Consider such a special example: (i) the subspace estimation indices $\mathscr{D}$ and the index estimation indices $\mathscr{T}$ are a partition of $\{1, \dots, n\}$; (ii) For each partition, $X_\mathcal{K}$, $\delta_\mathcal{V}$ and $\epsilon$ are perfectly orthogonal to each other. In this case, the RHS of Assumption \ref{assum:better} attains 1, which is a trivial upper bound of the LHS.

    On top of the two conditions above, we further require that all non-singleton clusters are signal clusters (meaning that $\left\{ k\in\mathcal{K}: |\mathcal{C}_k| \neq 1 \right\} \subseteq S^\star$). In this case, for $\ell_0$-base procedure, $\sup_{T\in\mathscr{T}} \mathbb{E}\left[ \tr(\P_{w_j} \P_{X_{\widehat{S}(T)}}) \right] = \inf_{T\in\mathscr{T}} \mathbb{E}\left[ \tr(\P_{w_j} \P_{X_{\widehat{S}(T)}}) \right] \equiv (s_0 - s^\star) / (p - s^\star)$ for any $j\in W$. Hence,
    $$
    \mathbb{E}\left[ \mathrm{FPE}(S, S^\star) \right] \leq \frac{1}{2\alpha - 1} \left[\frac{s_0^2}{p} + 2 s_0 \sqrt{\frac{\eta_1^2}{1 + \eta_1^2}} + \frac{4p \eta_1^2}{1 + \eta_1^2} \right] + \frac{s_0 \eta_1^2}{1 + \eta_1^2} .
    $$
    Figure \ref{fig:assump_clust2} presents FPE upper bounds for this special example, which are uniformly lower than $s_0 - s^\star$, the lower bound for FPE incurred by $\ell_0$-base procedure.

    \begin{figure}[!ht]
    \centering
    \includegraphics[width=0.6\textwidth]{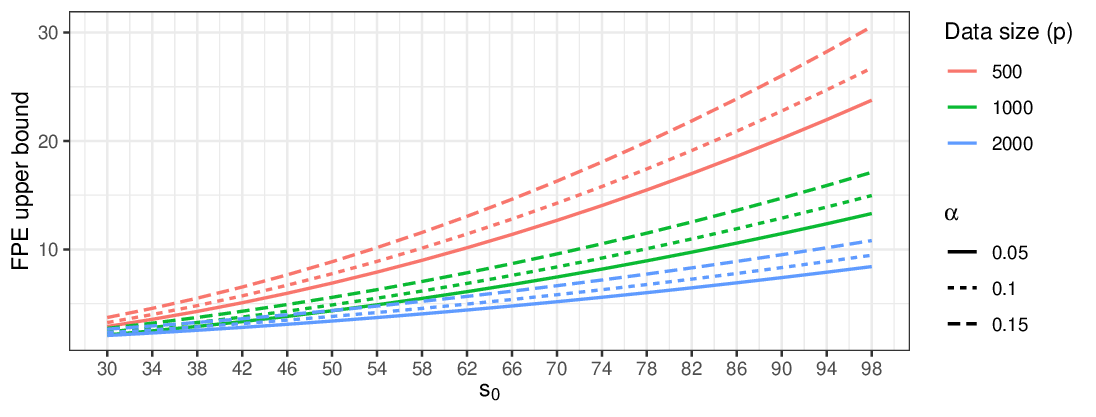}
    \caption{\small A numerical experiment for the FPE upper bound under Assumptions \ref{assum:better} and \ref{assum:exchange}, with $\alpha = 0.9$, $\eta_1 = 0.01$, and $s^\star=10$. } 
    \label{fig:assump_clust2}
    \end{figure}
    
\end{remark}

\vspace{0.2in}
\subsection{Similarity in prediction}
\label{sec:sup-prediction-similarity}

In Theorem~\ref{thm:2}, we showed that under certain conditions, FSSS will likely return
models in the set of ``equally good models" $\mathcal{S}$. Particularly, if $\mathfrak{S}$ contains all $\alpha$-stable sets (which it will with high probability if it is run for a large number
of iterations), then, with high probability, $\mathfrak
{S} = \mathcal{S}$.

It is noteworthy to point out that all models in $\mathcal{S}$ have similar in-sample prediction error. Specifically, assuming $\epsilon \sim \mathcal{N}(0,\sigma^2/nI_{n})$ for some $\sigma > 0$, we show in Appendix~\ref{sec:sup-proof-prediction} that, with high probability, 
\begin{align}\label{eq:prediction-similarity}
\begin{split}
    \max_{S_1 \neq S_2 \in \mathcal{S} } \left| \mathrm{pe}(S_1) - \mathrm{pe}(S_2) \right| 
\leq \frac{\eta_1^2}{1 + \eta_1^2} \|\beta^\star\|_2^2  + \mathcal{O}\left( s^\star \sqrt{\log p / n} \right) (\|\beta^\star\|_1 + \sigma)^2,
\end{split}
\end{align} 
where $\mathrm{pe}(S) = \left\| \P_{\mathrm{col}(X_S)^\perp}(y) \right\|_2^2 $ is the in-sample prediction error for the model $S$. As expected, smaller perturbations $\eta_1$ lead to more similar in-sample predictions errors among models in $\mathcal{S}$.

\vspace{0.2in}
\subsection{Beyond clustering: the complex dependency setup}
\label{sec:sup-complex}

We consider a design matrix composed of block structures. Let $\mathcal{K}$ be the index set of \textit{block parents}. For any $k\in\mathcal{K}$, we assume
for simplicity that $X_k$ is normalized so that $\|X_k\|_2 = 1$. Suppose that the \textit{child} of $\mathcal{K}$ is given by $X_c = \sum_{j\in\mathcal{K}} X_j + \delta$, where $\delta$ is a perturbation with $\|\delta_j\|_2 = \eta_1 \ll 1$. The entire block is denoted by $\mathcal{B}:= \mathcal{K} \cup \{c\}$. Additionally, there are \textit{individual features} indexed by $\mathcal{I}$, where for any $k\in\mathcal{I}$, $X_k$ is normalized so that $\|X_k\|_2 = 1$. We further assume that all $X_k$, $k\in\mathcal{K} \cup \mathcal{I}$, and the perturbation $\delta$ are nearly orthogonal; that is, for all $k\neq k' \in \mathcal{K} \cup \mathcal{I}$, the correlation $|\cos\langle X_k, X_{k'} \rangle|$ and $\cos\langle X_k, \delta \rangle|$ are upper bounded by $\mathcal{O}(\sqrt{\log p/ n})$.

The response variable is generated from block parents and individual features:
$
y = \sum_{k\in\mathcal{K} \cup \mathcal{I} } X_k \beta^\star_k + \epsilon.
$
Here $\epsilon \in \mathbb{R}^n$ is a random vector with independent and identically distributed entries.
We define the \textit{signal block parents} as $S^\star_\mathcal{B}: = \{k\in\mathcal{K}: \beta^\star_k \neq 0\}$, the \textit{noise block parents} as $N^\star_\mathcal{B}: = \{k\in\mathcal{K}: \beta^\star_k = 0\}$, the \textit{signal individual parents} as $S^\star_\mathcal{I}: = \{k\in\mathcal{I}: \beta^\star_k \neq 0\}$, and the \textit{noise individual parents} as $N^\star_\mathcal{I}: = \{k\in\mathcal{I}: \beta^\star_k = 0\}$. Moreover, let $K:= |\mathcal{K}|$, $S^\star := S^\star_\mathcal{B} \cup S^\star_\mathcal{I}$, and $s^\star:= \left|S^\star \right|$.

Finally, note that since $\eta_1 \ll 1$, when $N^\star_\mathcal{B} = \varnothing$, models that select $K$ features from $\mathcal{K}$ are nearly indistinguishable. In other words, models selecting all parents in $\mathcal{B}$, and models selecting any $K-1$ parents and the child in $\mathcal{B}$ span similar subspaces. As a result, targeting the signal set $S^\star$ is not meaningful. Instead, we define a class of ``block consistency models'':
$$
\mathcal{S}_0:= \left\{ \supp(\beta): \beta\in\mathbb{R}^p;
\ \left| \supp(\beta) \cap \mathcal{B} \right| = K;
\ \beta_j \neq 0 \ \forall j\in S^\star_\mathcal{I};
\ \beta_j = 0 \ \forall j\in N^\star_\mathcal{I}
\right\}.
$$
Note that all models in $\mathcal{S}_0$ span similar subspaces when $N^\star_\mathcal{B} = \varnothing$. 
We further denote
$
\mathcal{S}_1:= \left\{ S^\star \right\},
$
and
$$
\mathcal{S} := \left\{ S: S\in\mathcal{S}_0 \text{ if $N^\star_\mathcal{B} = \varnothing$}; 
\ S\in\mathcal{S}_1 \text{ if $N^\star_\mathcal{B} \neq \varnothing$}
\right\} .
$$

\vspace{0.1in}

We first present the properties of $\ell_0$-penalized base procedure (\ref{eq:lasso_s0}) in the complex dependency setup. Specifically, under a certain beta-min condition, when $s_0 \geq s^\star$, the $\ell_0$-base procedure does not miss any signal directions. When $s_0 = s^\star$, the ``block selection consistency'' is achieved. We proof Theorem \ref{prop:block-consist-s0} in Section \ref{sec:sup-proof-blockbase}.

\vspace{0.1in}
\begin{theorem}[Base procedure] \label{prop:block-consist-s0}
    Suppose $\max_{j\in \mathcal{K \cup \mathcal{I}} } \mathrm{Corr}(X_j, \epsilon) \leq \mathcal{O}(\sqrt{\log p / n})$ and $\mathrm{Corr}(\delta, \epsilon) \leq \mathcal{O}(\sqrt{\log p / n})$.
    Furthermore, assume that 
    \begin{align}\label{eq:complex-beta-min1}
    \begin{split}
     &   \min\left\{ \frac{\min_{j\neq k \in S^\star_\mathcal{B}} |\beta^\star_j - \beta^\star_k|^2   + 2 \eta_1^2 \min_{j\in S^\star_\mathcal{B} } |\beta^\star_j|^2}
    { 1 + K + (1-K) \mathds{1}_{\{N^\star_\mathcal{B} = \varnothing \}}  + \eta_1^2},\ 
    \frac{  (\eta_1^2 + 1) \min_{j\in S^\star_\mathcal{B} } |\beta^\star_j|^2 \mathds{1}_{\{N^\star_\mathcal{B} \neq \varnothing\}} }{K + \eta_1^2},\ 
    \min_{j\in S^\star_\mathcal{I}} |\beta_j^\star|^2 \right\} \\
    >{}& \frac{\eta_1^2}{1 + \eta_1^2} \|\beta^\star\|_2^2 + \mathcal{O}\left( (s_0 + K)\sqrt{\log p / n} \right) (\|\beta^\star\|_1 + \|\epsilon\|_2)^2,
    \end{split}
    \end{align}
    and additionally
    $$
    \min_{j\in S^\star_\mathcal{B}} |\beta^\star_j|^2 > \mathcal{O}\left( s_0 \sqrt{\log p / n} \right)\left( \|\beta^\star\|_1 + \|\epsilon\|_2 \right)^2 . $$
    Then when $s_0 \geq s^\star$,
    \begin{enumerate}[label = \normalfont(\arabic*)]
        \item The solution $\widehat{\beta}$ to (\ref{eq:lasso_s0}) does not miss any signals. That is, there must exist $\widetilde{S} \in \mathcal{S}_0 \cup \mathcal{S}_1 $ such that $\widetilde{S} \subseteq \supp(\widehat{\beta})$.

        \item When $s_0 = s^\star$, the solution $\widehat{\beta}$ achieves block feature selection consistency: $\supp(\widehat{\beta}) \in\mathcal{S}$.
    \end{enumerate}    
\end{theorem}

Analogously to the clustering setup, the following theorem demonstrates how FSSS, coupled with an appropriate base procedure, returns models in $\mathcal{S}$ even if the base procedure chooses many redundant or noise features.
Denote the quality of base procedure by
    $$
    \gamma := \max_{u \in \mathcal{U}} \max_{T\in\mathscr{T}} \mathbb{E}\left[\tr(\P_u \P_{X_{\widehat{S}(T)}}) \right] ,
    $$
    where $\mathscr{T}$ is the collection of subsample indices; additionally, the ``undesired direction'' is given by $\mathcal{U}:= \{\delta \} \cup \bigcup_{j\in N^\star_\mathcal{I}} \{ X_j \}$ if $N^\star_\mathcal{B} = \varnothing$, and $\mathcal{U}:= \{X_c \} \cup \bigcup_{ N^\star_\mathcal{B} \cup N^\star_\mathcal{I} } \{X_j\} $ otherwise.

\vspace{0.1in}
\begin{theorem}[Subsampling consistency] \label{prop:block-consist-subsample}
    Suppose that
    \begin{enumerate} [label = \normalfont(\roman*)]
        \item The base procedure does not miss any signal: for any $\ell \in \{1, \dots, B\}$, there must exist $S \in\mathcal{S}_0 \cup \mathcal{S}_1$ such that $S \subseteq \widehat{S}^{(\ell)}$.

        \item $ \max\left\{ \frac{\eta_1^2}{1 + \eta_1^2},\ \frac{|S^\star_\mathcal{B}|}{K + \eta_1^2} \right\} + \mathcal{O}(s^\star (K+s^\star) \sqrt{\log p / n}) < \frac{1}{2} $, and $s^\star \leq (n / \log p)^{1/8}$.

        \item $K^2 \sqrt{\log p /n} \rightarrow 0$ as $K,p,n \rightarrow  \infty$. 
    \end{enumerate}
    Then for any $\alpha_0 \in (0, 1/2)$, with the threshold
    $$
    f_1(\eta_1) - \alpha_0 + \mathcal{O}\left( s^\star \sqrt{\log p / n} \right) \leq  \alpha \leq  f_2(s^\star, \eta_1, s_0) + \mathcal{O}\left( s^\star (s_0 + B) \sqrt{\log p / n} \right),
    $$
    where 
    $
    f_1(\eta_1) := 1 +20  \sqrt{ \eta_1^2 / (1 + \eta_1^2) }  + 8 \eta_1^2 / (1 + \eta_1^2),
    $
    and 
    $
    f_2(s^\star, \eta_1, s_0) := 1- (3(s^\star)^2 + 2s^\star s_0) \sqrt{\eta_1^2 / (1 + \eta_1^2)} + (3 s^\star + s_0) \eta_1^2 / (1 + \eta_1^2) + \left(\eta_1^2 / (1 + \eta_1^2) \right)^{3/2},
    $
    we have $\mathfrak{S} \subseteq \mathcal{S}$ with probability $1 - \left( p - s^\star \right) \cdot \gamma^2 / (1 - 2 \alpha_0) $.

\end{theorem}

We proof Theorem \ref{prop:block-consist-subsample} in Section \ref{sec:sup-proof-blockfsss}. Theorem \ref{prop:block-consist-subsample} reveals that, under certain conditions, the subspaces spanned by FSSS selection sets closely resemble the oracle subspace $\spa(X_{S^\star})$ with high probability. Consequently, false discovery can be immediately controlled, as established in Corollary \ref{cor:pfer2}.

\vspace{0.1in}
\begin{corollary}\label{cor:pfer2}
    Under the assumptions of Theorem \ref{prop:block-consist-subsample}, we further assume that the FSSS algorithm stops before it selects $s_0$ features. Then, every FSSS selection set $S$ satisfies the FPE bound: 
    \begin{align*}
    \begin{split}
        \mathbb{E} \left[\mathrm{FPE}(S,S^\star) \right]
        &\leq  \frac{\eta_1^2}{1 + \eta_1^2} \left[ 1 - \left( p - s^\star \right) \cdot \gamma^2 / (1 - 2 \alpha_0) \right]  + \mathcal{O}\left( (s^\star)^2 \sqrt{\log p / n} \right) \\
        & \quad  + 
    s_0 \left(  p - s^\star \right) \cdot \gamma^2 / (1 - 2 \alpha_0).
    \end{split}
    \end{align*}
\end{corollary}
We proof Corollary \ref{cor:pfer2} in Section \ref{sec:sup-proof-blockcor}.

\vspace{0.1in}
\begin{remark}\label{rem:theory-block}
    We study the behavior of $\gamma$ in Theorem \ref{prop:block-consist-subsample} throught a very special case. Suppose that the subspace estimation indices $\mathscr{D}$ and the index estimation indices $\mathscr{T}$ are a partition of $\{1, \dots, n\}$. For each partition, $\mathcal{K}$, $\{\delta\}$, $\mathcal{I}$ and $\epsilon$ are perfectly orthogonal: $Z_i \perp Z_j$ for all $Z_i\neq Z_j \in X_\mathcal{K} \cup \{\delta\} \cup \mathcal{I}\cup \{\epsilon\}$. In this case,
    $$
    \gamma = \frac{s_0 - s^\star}{p - s^\star} + 
    \max\left\{ 
    \sum_{m=0}^{|N^\star_\mathcal{B}|} \frac{( |S^\star_\mathcal{B}| +m )  }{ K + \eta_1^2 } \frac{\binom{|N^\star_\mathcal{B}| }{m} \binom{p - s^\star -|N^\star_\mathcal{B}| -1 }{s_0 - s^\star - m}  }{ \binom{p-s^\star}{s_0 - s^\star} } , \quad
    \sum_{m=0}^{|N^\star_\mathcal{B}| - 1} \frac{ 1 }{ |N^\star_\mathcal{B}| - m + \eta_1^2 } \frac{\binom{|N^\star_\mathcal{B}|-1 }{m} \binom{p - s^\star -|N^\star_\mathcal{B}| -1 }{s_0 - s^\star - m}  }{ \binom{p-s^\star}{s_0 - s^\star} }
    \right\}.
    $$
    See Section \ref{sec:sup-proof-blockrem} for a proof.
    Figure \ref{fig:assump_block} shows a numerical experiment of $\gamma$, the probability of $\mathfrak{S} \subseteq \mathcal{S}$, and the FPE upper bound. We observe a meaningful region where FSSS achieves uniformly lower FPE upper bounds than the base procedure, attributable to the benefits of subsampling.

    \begin{figure}[!ht]
    \centering
    \includegraphics[width=\textwidth]{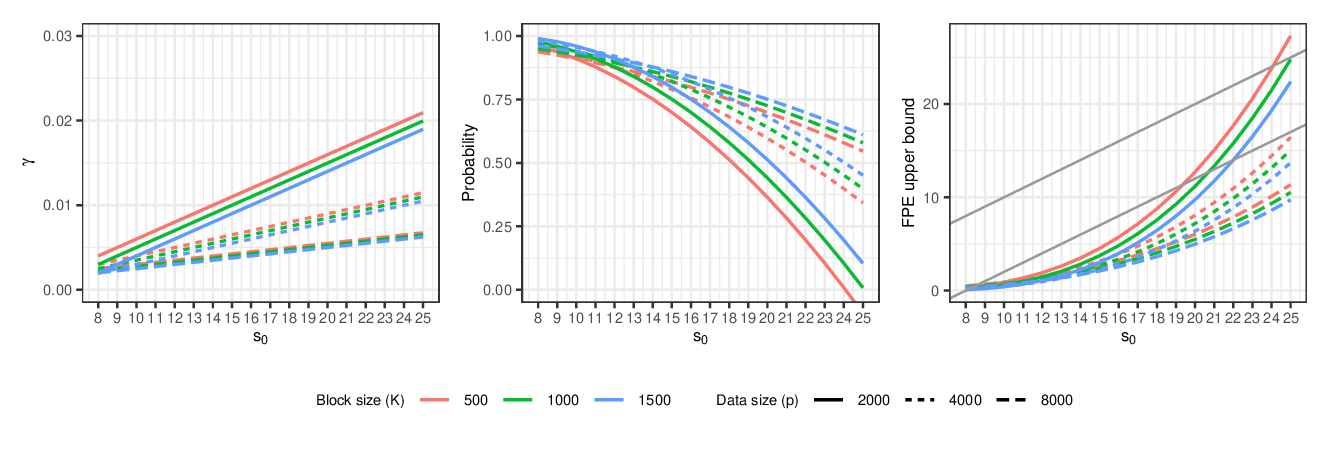}
    \caption{\small A numerical experiment illustrating the behavior of $\gamma$ and its related quantities in Theorem \ref{prop:block-consist-subsample}, with $\alpha_0 = 0.1$, $\eta_1 = 0.01$, $|S^\star_\mathcal{I} |=5$, and $|N^\star_{\mathcal{B}}| = \lfloor 0.998 K \rfloor$. Left panel: $\gamma$ values for different $s_0$. Middle panel: the probability $1 - (p-s^\star) \cdot \gamma^2 / (1 - 2\alpha_0)$ for different $s_0$. Right panel: the FPE upper bounds for different $s_0$, where the gray lines indicate the upper bound and lower bound for $\ell_0$-base procedure respectively.} 
    \label{fig:assump_block}
    \end{figure}

\end{remark}

\vspace{0.2in}
\section{Simulation and data application details}

\subsection{Setup of Figure \ref{fig:stability}}
\label{sec:sup-figure1}

Figure \ref{fig:stability} presents stability paths of each individual feature on a synthetic dataset $(X\in\mathbb{R}^{n\times p}, y\in\mathbb{R}^n)$. This dataset contains $n = 100$ observations and $p = 82$ features, organized into 8 cluster blocks, 2 dependency blocks, and 50 individual features. Specifically, for each cluster $\mathcal{C}_k = \{k, k+1, k+2\}$ where $k\in\{1, 4, 7, 10, 13, 16, 19, 22\}$, the representative feature is generated as $X_k \overset{iid}{\sim} \mathcal{N}(0, I_n)$. The proxies are defined as $X_j = X_k + \delta_j$ for $j\in\{k+1, j+2\}$, where the perturbation directions follow $\delta_j \overset{iid}{\sim} \mathcal{N}(0, 0.2^2 I_n)$. For the first dependency block $\mathcal{B}_1 = \{25,26,27,28\}$, the parents $X_{25}$ and $X_{26}$ are iid drawn from $\mathcal{N}(0, I_n)$. Their first child is defined as $X_{27} = X_{25} + X_{26} + \delta_{27}$, and the second as $X_{28} = X_{25} - X_{26} + \delta_{28}$, where the perturbations follows $\delta_j \overset{iid}{\sim} \mathcal{N}(0, 0.2^2)$, $j\in\{27,28\}$. The second dependency block $\mathcal{B}_2 = \{29, 30, 31, 32\}$ is generated in the same way as $\mathcal{B}_1$, with $X_{29}$ and $X_{30}$ as the parents, and $X_{31}$ and $X_{32}$ as the children. The remaining 50 individual features are iid drawn from $\mathcal{N}(0, I_n)$. All representative features, perturbation directions, and individual features are independently generated.

We assign coefficients one to all representative features, while setting the coefficients of their proxies as zero. For each dependency block, the two parent features are assigned coefficients of 1.5 and 1, respectively, whereas the child features receive zero coefficients. All remaining individual features are treated as noise with zero coefficients. Letting $\beta^\star$ denote the coefficient vector, the response variable is generated as $y = X\beta^\star + \epsilon$ where $\epsilon \overset{iid}{\sim} \mathcal{N}(0, 0.2^2 I_n)$. Under this setup, the ``signal'' features in the legend of Figure \ref{fig:stability} refer to all representative features in clusters, and all parents features from dependency blocks. The ``correlated signal'' refers to the cluster proxies and the children in dependency blocks. Finally, the ``noise'' features refer to those individual features with zero coefficients.

We clarify the quantities computed for each panel in Figure \ref{fig:stability}. We perform subsampling using $\ell_0$-penalized regression as the base procedure, with a total of 100 subsamples. The tuning parameter is the number of features selected in each subsample.
In the left panel, corresponding to standard stability selection, the stability value for each feature refers to its selection proportion across the subsamples. In the middle panel, corresponding to cluster stability selection, the stability value refers to the selection proportion of the cluster to which each feature belongs. Clusters are determined using hierarchical clustering with a distance metric defined as $1 - \mathrm{cor}(X_j, X_k)$ for any two features $X_j$ and $X_k$, and a cutoff height of 0.2.
In the right panel, under our proposed subspace framework, the stability of each feature is measured by the subspace-based quantity $\pi(X_j)$.

\begin{figure}[!ht]
    \centering
    \includegraphics[width=\textwidth]{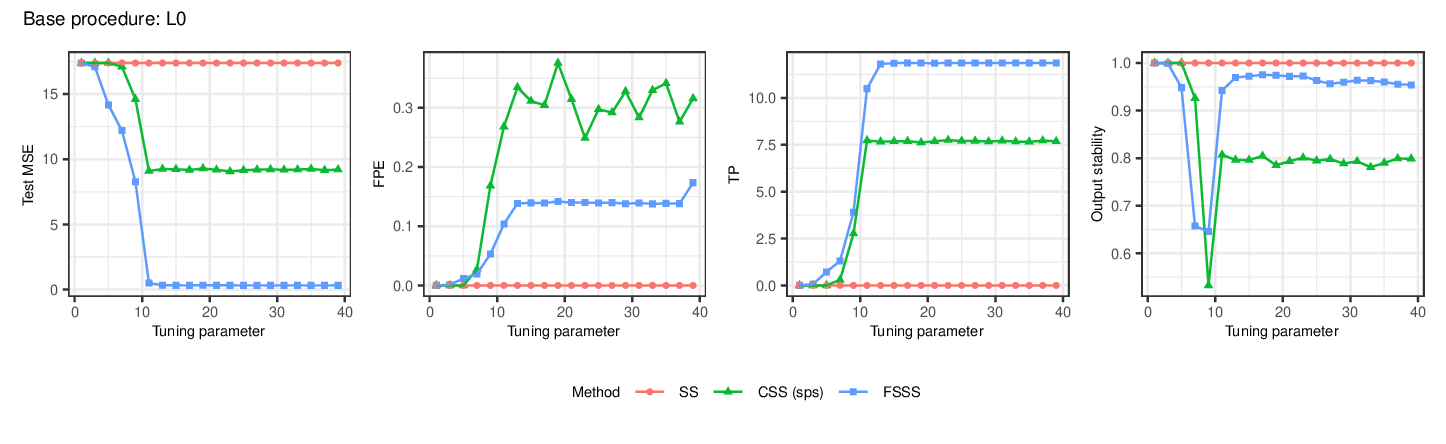}
    \caption{\small Performance of SS, sparsity (sps) CSS, and FSSS on the dataset of Figure \ref{fig:stability} with $\alpha = 0.8$. The subsampling uses $\ell_0$-regression as the base procedure, and takes $B=100$ subsamples. The tuning parameter is the number of selected features for each subsample. Points on the plot represent averages across 2 model sizes to smooth the plot.} 
    \label{fig:teaser_perform}
\end{figure}

Lastly, we evaluate the performance of stability selection, sparse cluster stability selection, and FSSS on this dataset with $\alpha = 0.8$. Specifically, we apply feature selection and coefficient estimation on this dataset, and report the test MSE, FPE, TP, and output stability using a separate test set drawn from the same distribution. We repeat this procedure $M = 50$ times. As shown in Figure \ref{fig:teaser_perform}, stability selection hardly makes any selection, and FSSS uniformly outperforms CSS in terms of test MSE, FPE, TP, and output stability.

\vspace{0.2in}
\subsection{More on the simulations}
\label{sec:sup-synthetic}

In Section \ref{sec:simulation}, we compared the performance of different methods across tuning parameters $s_0 \in \{2, 3, \dots, 41\}$. Figure \ref{fig:synthetic_perform_l0} and \ref{fig:synthetic_perform_lasso} display the evaluation metrics for all $s_0$.

\begin{figure}[!ht]
    \centering
    \includegraphics[width=\textwidth]{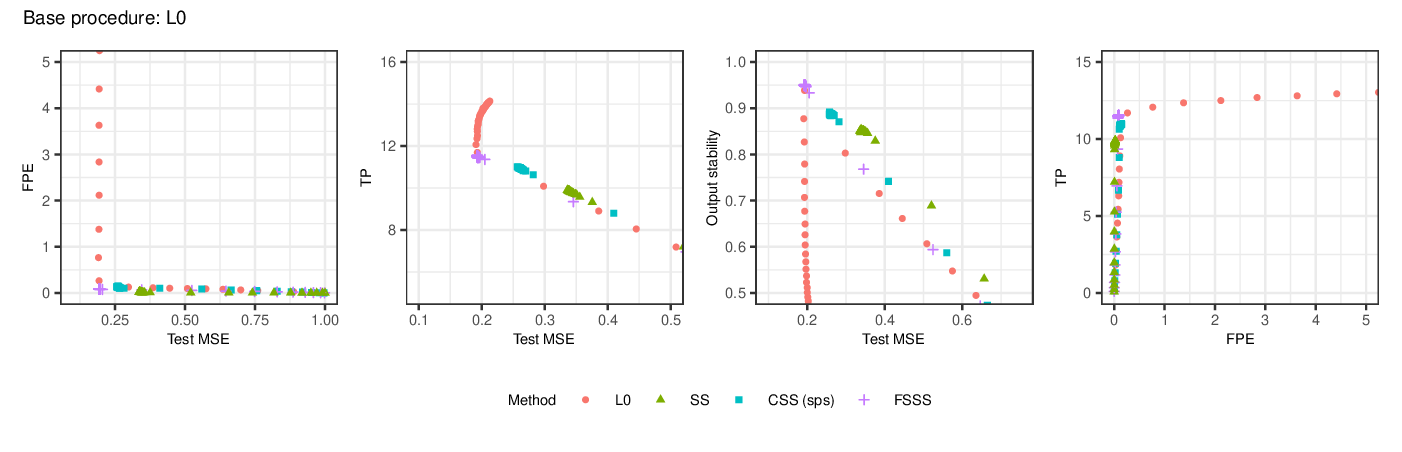}
    \caption{\small Performance of SS, sparsity (sps) CSS, and FSSS on synthetic datasets. The subsampling uses $\ell_0$-regression as the base procedure, and takes $B=200$ subsamples. The tuning parameter is the number of selected features for base procedure.} 
    \label{fig:synthetic_perform_l0}
\end{figure}

\begin{figure}[!ht]
    \centering
    \includegraphics[width=\textwidth]{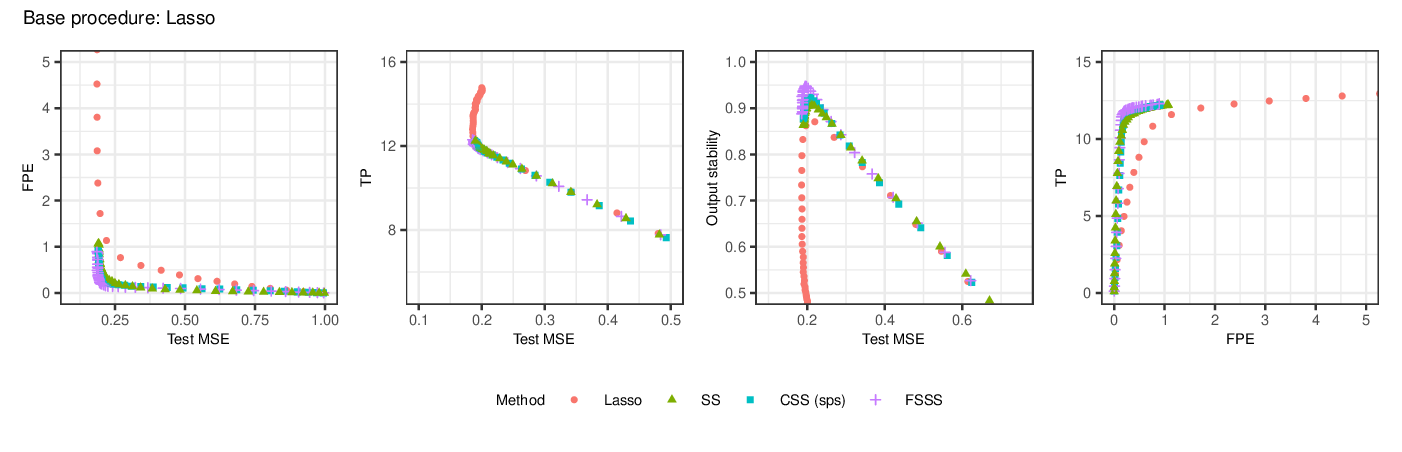}
    \caption{\small Performance of SS, sparsity (sps) CSS, and FSSS on synthetic datasets. The subsampling uses Lasso as the base procedure, and takes $B=200$ subsamples. The tuning parameter is the number of selected features for base procedure.} 
    \label{fig:synthetic_perform_lasso}
\end{figure}
\FloatBarrier

\vspace{0.2in}
\subsection{More on the breast cancer dataset}
\label{sec:sup-real-data}

All 50 selection sets in $\mathfrak{S}$ are listed in Table \ref{tab:breast-selection-set}.

\vspace{0.1in}
\begin{table}[!ht]
\small
\centering
\begin{tabular}{|l|l|}
\hline
$HEPH, ACSL1, TBC1D9$                           & $ABI1, TFAP2B, DNAJC12$                       \\ \hline
$PRKAR1A, EIF5B, PTPRC_2, CEP55, S100A14$       & $ALOX5AP, TRBC1_2, H3F3A, CEP55$              \\ \hline
$CA12_2, WASL, H2BFS, ASPN$                     & $COL10A1, G3BP2, IGHM, ADH1B_2$               \\ \hline
$AMD1, TOP2A, ELOVL2, IGK_3$                    & $DDX24, CKS2, GABRP, ATP8B1, HLA\text{-}DQB1_6$      \\ \hline
$ALDH3B2, COL18A1, RRM2_2, ATP5G2P1$            & $ADIPOQ, CA12_3, KDM4B, GREM1$                \\ \hline
$KCNK1, TFF1, ADIPOQ, ARF1$                     & $HES1, MTMR6, IGKV1OR2\text{-}108, PRC1, S100A14$    \\ \hline 
$SSPN, ADH1B_2, OSM, IGKV1OR2\text{-}108, S100A14$    & $NFIB_2, RRM2_2, MYCN, PLA2G12A$        \\ \hline
$GABRP, BLNK, CCNB1, ABHD2$                     & $FOXA1, MYCN, MFAP4, KIF4A$                   \\ \hline
$ALCAM, NEK2, IGLC1, UBE2G1$                    & $ESR1, DSC3, ADH1B, EP300_2, S100A14$         \\ \hline
$FHL1, LY96, TOP1, COPS4, S100A14$              & $SFRP1, GULP1, TRBC1_2, RHOB$                 \\ \hline
$SNRPD1, DUSP4_2, PTPRC$                        & \begin{tabular}[c]{@{}l@{}}$SPARC, ANXA3, IL6ST_3, HLA\text{-}DQB1_6, CEP55,$ \\ $S100A14$ \end{tabular} \\ \hline
$THBS4, CA12_3, TAF9B$                          & $AEBP1, SFRP1_2, BCL2, IFI44L, EIF4G1$        \\ \hline
$INSIG1, CA12, CHRDL1, RBM25, HLA\text{-}DQB1_6$      & $SLC7A5, MSH3, MUC1, ID4$               \\ \hline
$HSPA8, IGLC1_3, S100A14, PDGFD$                & $SF1, IGK, EXOC5$                             \\ \hline
$GBP1, GABRP, MXRA8, GM2A$                      & $LPL, ISLR, G3BP2, TBC1D9$                    \\ \hline
$IAPP, BUB1, HLA\text{-}DQB1_6, RRN3$           & $MMP11_2, NFIB_2, DCN_2, HLA\text{-}DQB1_6, CKLF$    \\ \hline
$UBE2S, MTUS1, NOTCH2NL, CA12_5$                & $WASL, CD36, NAT1, S100A14$                   \\ \hline
$MELK, SSPN, HLA\text{-}DQB1_6, FBXO3, S100A14$ & $HTRA1, NR3C1, HLA\text{-}DQB1_4, N4BP2L2_2$  \\ \hline
$CCL5, LOC101928189, S100A14, EAF2"$            & $MAD2L1, EIF3J, GOLGA8A_2, IGK_3, S100A14$    \\ \hline
$RRM2, ERBB4, IGK_3$                            & $HTRA1, OSM, IGK_2, S100A14, CSPP1$           \\ \hline
$CD55, CRIP1, PLIN1, UBXN4$                     & $FOSB, C1QB, COPS4, S100A14$                  \\ \hline
$TOP1, DKK3, HLA\text{-}DQB1_6, S100A14, C2orf47$     & $CCL5, H2AFZ, SFRP1, VCAN_3, KDM4B$               \\ \hline
$ACAP1, KRT7, CALM1, ASPN$                      & $CA12_2, ADH1B_2, SLC35A3, MTUS1, HLA\text{-}DQB1_6$    \\ \hline
$CAV1, COL11A1, HLA\text{-}DQB1, CA12_3, VENTXP1$     & $EIF4G1, PTGER3, HLA\text{-}DQB1_6, S100A14, CKLF$       \\ \hline
\begin{tabular}[c]{@{}l@{}}$ALDH3B2, IGK_3, CEP55, S100A14, ADAMTS5,$ \\ $SETD8$ \end{tabular}        & $SOX9, CALU, DNAJC12$          \\ \hline
\end{tabular}
\caption{\small The 50 selection sets obtained by FSSS from the breast cancer dataset.}
\label{tab:breast-selection-set}
\end{table}

\clearpage
\section{Proofs}\label{sec:sup-proof}
\subsection{Proof of equation \texorpdfstring{\eqref{eq:stability-explain}}{(6)}}
\label{sec:equivalence}
We prove the result that
$$\pi(S) = \min_{z \in \mathrm{col}(X_S)} \frac{1}{B} \sum_{\ell=1}^B \allowbreak\mathrm{trace}\allowbreak(\mathcal{P}_{\mathrm{col}(z)}\mathcal{P}_{\mathrm{col}(\widehat{S}^{(\ell)})}),
$$
when $X_S$ are linearly independent. This relation follows from the next lemma.
\begin{lemma}
    For a subspace $T\subseteq \mathbb{R}^n$, and a positive semi-definite matrix $M$, we have $\sigma_{\mathrm{dim}(T)}(\P_T M \P_T) = \min_{z\in T} \tr(\P_{\mathrm{col}(z)} M ).$ By plugging in $M = \mathcal{P}_{\mathrm{avg}}$, we have $\sigma_{\mathrm{dim}(T)}(\P_T \mathcal{P}_\mathrm{avg} \P_T) = \min_{z\in T} \tr(\P_{\mathrm{col}(z)}\mathcal{P}_\mathrm{avg}  ) = \min_{z\in T} \frac{1}{B}\sum_{\ell=1}^B\tr(\P_{\mathrm{col}(z)}\P_{\mathrm{col}(X_{\widehat{S}^{(\ell)}})})$. 
\end{lemma}

\begin{proof}
    Notice that
\begin{align*}
\begin{split}
    \sigma_{\mathrm{dim}(T)}(\P_T M \P_T) 
&= \min_{z\in T, \|z\|_2 = 1} z^\top \P_{T}M\P_T{z}\\
&= \min_{z \in T,\|z\|_2=1} z^\top{M}z \\
&= \min_{z \in T, \|z\|_2=1} \tr(zz^\top{M})\\
&= \min_{z \in T} \tr(\P_{\mathrm{col}(z)}{M}).
\end{split}
\end{align*}
To see why the first equality is true, let $\lambda_1\geq \lambda_2\geq \dots \geq \lambda_{\mathrm{dim}(T)} \geq 0$ be the first $\mathrm{dim}(T)$ singular values of $\P_T{M}\P_T$. Suppose $\lambda_{\mathrm{dim}(T)}>0$. Let $u_1,\dots,u_{\mathrm{dim}(T)}$ be the corresponding eigenvectors. Then, any $z \in T$ with $\|z\|_2 = 1$ can be expressed as the linear combination $z = \sum_{i= 1}^{\mathrm{dim}(T)}\alpha_i{u}_i$ with $\sum_{i}\alpha_i^2 = 1$. Then, a straightforward calculation yields $z^\top \P_{T}M\P_T{z} = \sum_{i}\alpha_i^2\lambda_i$. Thus, $\min_{z\in T, \|z\|_2 = 1} z^\top \P_{T}{M}\P_T{z}= \lambda_{\mathrm{dim}(T)}$, as desired. Now suppose $\lambda_{\mathrm{dim}(T)} = 0$. This shows that there exists $z \in T$ such that $\P_{T}{M}\P_T{z} = 0$. Since $\mathcal{P}_T{M}\P_{T}$ is positive semi-definite, we have that in this setting, $\min_{z\in T, \|z\|_2 = 1} z^\top \P_{T}{M}\P_T{z} = 0 = \lambda_{\mathrm{dim}(T)}$.
\end{proof}

\vspace{0.2in}
\subsection{The clustering setup}
\label{sec:sup-proof-cluster}

\emph{Notation:} We additionally define an operator that identifies the representative features associated with a given set $S$, defined as $\mathrm{RF}(S):= \{k \in \mathcal{K}: S \cap \mathcal{C}_k \neq \varnothing \}$. 

\subsubsection{Preliminaries}
\begin{lemma}\label{lem:cluster-angle}
    Suppose that the following correlation terms are upper bounded by $\eta_0$: $\max_{k\neq k' \in \mathcal{K} } \mathrm{Corr}(X_k, X_{k'})$,
    $\max_{k\in\mathcal{K}, j\in\mathcal{V}} \mathrm{Corr}(X_k, \delta_j)$,
    $\max_{j\neq j' \in\mathcal{V}} \mathrm{Corr}(\delta_j, \delta_{j'})$, 
    $\max_{k\in\mathcal{K} } \mathrm{Corr}(X_k, \epsilon)$, and
    $\max_{ j\in\mathcal{V}} \mathrm{Corr}(\delta_j, \epsilon)$.
    We assume $\eta_0 = \mathcal{O}(\sqrt{\log p / n})$.
    Then we have
    \begin{enumerate}[label = \normalfont(\arabic*)]
        \item For $k \in \mathcal{K}$, $j \in \mathcal{V}_{k'}$, $k\neq k'$, $\tr(\P_{X_k} \P_{X_j}) \in \left[ 0, \ 
        \frac{ \eta_0^2 (1 + \eta_1)^2 }{ 1 - 2\eta_1\eta_0 } \right] $. Additionally, $\tr(\P_{X_k}\P_{X_j}) \leq \mathcal{O}(\log p / n)$.
        
        \item For $j\in \mathcal{V}_{k_1}$, $l\in\mathcal{V}_{k_2}$, $k_1\neq k_2$, $\tr(\P_{X_j} \P_{X_l}) \in \left[ 0,\ \frac{ \eta_0^2 (1 + \eta_1)^4 }{ (1 - 2\eta_1 \eta_0)^2 } \right]$. Additionally, $\tr(\P_{X_j}\P_{X_l}) \leq \mathcal{O}(\log p / n)$.

        \item For $j\in \mathcal{V}_{k_1}$, $l\in\mathcal{V}_{k_2}$, $k_1\neq k_2$, $\tr(\P_{X_j} \P_{\delta_l}) \in \left[ 0,\ \frac{ \eta_0^2 (1 + \eta_1)^2 }{ 1 - 2\eta_1 \eta_0 } \right]$. Additionally, $\tr(\P_{X_j} \P_{\delta_l}) \leq \mathcal{O}(\log p / n)$.

        \item For $k\in\mathcal{K}$, $j \in \mathcal{V}_k$, $\tr(\P_{X_k} \P_{X_j}) \in \left[ \frac{(1-\eta_0\eta_1)^2}{ 1 + \eta_1^2 + 2\eta_0\eta_1},\ 1 \right]$. Additionally, $\tr(\P_{X_k} \P_{X_j}) \geq \frac{1}{1+\eta_1^2} - \mathcal{O}(\sqrt{\log p / n}) $.

        \item For $j\neq l \in \mathcal{V}_k$, $\tr(\P_{X_j} \P_{X_l}) \in \left[ \frac{(1 - 2\eta_0\eta_1 - \eta_0\eta_1^2)^2}{ (1 + \eta_1^2 + 2\eta_0\eta_1)^2 }, \ 1  \right]$. Additionally, $\tr(\P_{X_j} \P_{X_l}) \geq \frac{1}{(1 + \eta_1^2)^2} - \mathcal{O}(\sqrt{\log p / n}) $.

        \item For $j\neq l \in \mathcal{V}_k$, $\tr(\P_{X_j} \P_{\delta_l}) \in \left[ 0, \ 
        \frac{ \eta_0^2 (1 + \eta_1)^2 }{ 1 - 2\eta_1\eta_0 } \right]$. Additionally, $\tr(\P_{X_j} \P_{\delta_k}) \leq \mathcal{O}(\log p / n)$.

        \item For $j \in \mathcal{V}_k$, $\tr(\P_{X_j} \P_{\delta_j}) \in \left[ 0, \  \frac{ (\eta_1 + \eta_0)^2 }{1 + \eta_1^2 - 2\eta_1 \eta_0} \right]$. Additionally, $\tr(\P_{X_j} \P_{\delta_j}) \leq \frac{\eta_1^2}{1 + \eta_1^2} + \mathcal{O}(\sqrt{\log p / n})$.

        \item For $j\in \mathcal{V}_k$, $\tr(\P_{X_j} \P_\epsilon) \in \left[ 0,\ \frac{ \eta_0^2 (1 + \eta_1)^2 }{ (1 - 2\eta_1 \eta_0)^2 } \right]$. Additionally, $\tr(\P_{X_j} \P_\epsilon) \leq \mathcal{O}(\log p / n)$.

    \end{enumerate}
\end{lemma}
\begin{proof}

    \noindent\textbf{Part (1).} Let $\xi_j = \delta_j / \|\delta_j\|_2$, then
    $$
    \tr(\P_{X_k} \P_{X_j}) 
    = \frac{|X_k^\top X_j|^2}{\|X_k\|^2_2 \cdot \|X_j\|^2_2} 
    = \frac{|X_k^\top (X_{k'} + \delta_j)|^2}{\|X_{k'} + \delta_j\|_2^2}
    \leq \frac{ (|X_k^\top X_{k'}| + |X^\top_k \xi_j| \cdot \|\delta_j\|_2 )^2 }{ 1 - 2 |X_{k'}^\top \xi_j| \cdot \|\delta_j\|_2  }
    \leq \frac{ \eta_0^2 (1 + \eta_1)^2 }{1 - 2 \eta_0\eta_1}.
    $$

    \noindent\textbf{Part (2).}
    Let $\xi_j = \delta_j / \|\delta_j\|_2$, and $\xi_l = \delta_l / \|\delta_l\|_2$. Then
    \begin{align*}
    \begin{split}
        \tr(\P_{X_j}\P_{X_l}) 
    &= \frac{|X_j^\top X_l|^2 }{ \|X_j\|^2_2 \cdot \|X_l\|_2^2 }
    = \frac{ |(X_{k_1} + \delta_j)^\top (X_{k_2} + \delta_l)|^2  }{ \|X_{k_1} + \delta_j\|^2_2 \cdot \|X_{k_2} + \delta_l\|^2_2 }  \\
    &\leq \frac{( |X_{k_1}^\top X_{k_2}| + |X_{k_1}^\top \xi_l| \cdot \|\delta_l\|_2 + |\xi_j^\top X_{k_2} | \cdot \|\delta_j\|_2 + |\xi_j^\top \xi_l| \cdot \|\delta_j\|_2 \|\delta_l\|_2 )^2  }{ (1 - 2 |X_{k_1}^\top \xi_j| \cdot \|\delta_j\|_2) (1 - 2 |X_{k_2}^\top \xi_l| \cdot \|\delta_l\|_2) }
    \leq \frac{ \eta_0^2 (1 +\eta_1)^4 }{(1 - 2\eta_1 \eta_0)^2}.
    \end{split}
    \end{align*}

    \noindent\textbf{Part (3).}
    Let $\xi_j = \delta_j / \|\delta_j\|_2$, and $\xi_l = \delta_l / \|\delta_l\|_2$. Then
    \begin{align*}
    \begin{split}
        \tr(\P_{X_j}\P_{\delta_l}) 
    &= \frac{|X_j^\top \delta_l|^2 }{ \|X_j\|^2_2 \cdot \|\delta_l\|_2^2 }
    = \frac{ |(X_{k_1} + \delta_j)^\top \delta_l|^2  }{ \|X_{k_1} + \delta_j\|^2_2 \cdot \|\delta_l\|^2_2 }  
    \leq \frac{( |X_{k_1}^\top \xi_{l}| \cdot \|\delta_l\|_2 + |\xi_{j}^\top \xi_l| \cdot \|\delta_j\|_2 \cdot \|\delta_l\|_2  )^2  }{ (1 - 2 |X_{k_1}^\top \xi_j| \cdot \|\delta_j\|_2) \cdot \eta_1^2 } \\
    &\leq \frac{ \eta_0^2  (1 +\eta_1)^2 }{1 - 2\eta_1 \eta_0}.
    \end{split}
    \end{align*}

    \noindent\textbf{Part (4).}
    Let $\xi_j = \delta_j / \|\delta_j\|_2$. Then
    \begin{align*}
    \begin{split}
        \tr(\P_{X_k}\P_{X_j}) 
    &= \frac{|X_k^\top X_j|^2 }{ \|X_k\|^2_2 \cdot \|X_j\|_2^2 }
    = \frac{ |X_{k}^\top (X_k + \delta_j)|^2  }{ \|X_{k} \|^2_2 \cdot \|X_k + \delta_j\|^2_2 }  
    \geq \frac{( 1 - |X_{k}^\top \xi_j| \cdot \|\delta_j\|_2  )^2  }{ 1 + \eta_1^2 + 2 |X_{k}^\top \xi_j| \cdot \|\delta_j\|_2  }
    \geq \frac{ (1-\eta_1\eta_0)^2 }{ 1 + \eta_1^2 + 2\eta_1 \eta_0}.
    \end{split}
    \end{align*}

    \noindent\textbf{Part (5).}
    Let $\xi_j = \delta_j / \|\delta_j\|_2$, and $\xi_l = \delta_l / \|\delta_l\|_2$. Then
    \begin{align*}
    \begin{split}
        \tr(\P_{X_j}\P_{X_l}) 
    &= \frac{|X_j^\top X_l|^2 }{ \|X_j\|^2_2 \cdot \|X_l\|_2^2 }
    = \frac{ |(X_{k} + \delta_j)^\top (X_k + \delta_l)|^2  }{ \|X_{k} + \delta_j \|^2_2 \cdot \|X_k + \delta_l\|^2_2 }  \\
    &\geq \frac{( 1 - |X_{k}^\top \xi_l| \cdot \|\delta_l\|_2 - |X_{k}^\top \xi_j| \cdot \|\delta_j\|_2 - |\xi_j^\top \xi_l| \cdot \|\delta_j\|_2 \cdot \|\delta_l\|_2  )^2  }{ (1 + \eta_1^2 + 2 |X_{k}^\top \xi_j| \cdot \|\delta_j\|_2) (1 + \eta_1^2 + 2 |X_{k}^\top \xi_l| \cdot \|\delta_l\|_2)  }
    \geq \frac{ (1-2\eta_0\eta_1 - \eta_0\eta_1^2)^2 }{ (1 + \eta_1^2 + 2\eta_0\eta_1)^2 }.
    \end{split}
    \end{align*}

    \noindent\textbf{Part (6).}
    Let $\xi_j = \delta_j / \|\delta_j\|_2$, and $\xi_l = \delta_l / \|\delta_l\|_2$. Then
    \begin{align*}
    \begin{split}
        \tr(\P_{X_j}\P_{\delta_l}) 
    &= \frac{|X_j^\top \delta_l|^2 }{ \|X_j\|^2_2 \cdot \|\delta_l\|_2^2 }
    = \frac{ |(X_{k} + \delta_j)^\top \delta_l|^2  }{ \|X_{k} + \delta_j \|^2_2 \cdot \eta_1^2 } 
    \leq \frac{(|X_{k}^\top \xi_l| \cdot \|\delta_l\|_2 + |\xi_j^\top \xi_l| \cdot \|\delta_j\|_2 \cdot \|\delta_l\|_2   )^2  }{ (1 - 2 |X_{k}^\top \xi_j| \cdot \|\delta_j\|_2) \cdot \eta_1^2  }
    \leq \frac{ \eta_0^2 (1 + \eta_1)^2 }{ 1 -  2\eta_1 \eta_0}.
    \end{split}
    \end{align*}

    \noindent\textbf{Part (7).}
    Let $\xi_j = \delta_j / \|\delta_j\|_2$. Then
    \begin{align*}
    \begin{split}
        \tr(\P_{X_j}\P_{\delta_j}) 
    &= \frac{|X_j^\top \delta_j|^2 }{ \|X_j\|^2_2 \cdot \|\delta_j\|_2^2 }
    = \frac{ |(X_{k} + \delta_j)^\top \delta_j|^2  }{ \|X_{k} + \delta_j \|^2_2 \cdot \eta_1^2 }  
    \leq \frac{(|X_{k}^\top \xi_j| \cdot \|\delta_j\|_2 + \eta_1^2   )^2  }{ (1 + \eta_1^2 - 2 |X_{k}^\top \xi_j| \cdot \|\delta_j\|_2) \cdot \eta_1^2  }
    \leq \frac{ (\eta_0 + \eta_1)^2 }{ 1 + \eta_1^2 -  2\eta_1 \eta_0}.
    \end{split}
    \end{align*}

    \noindent\textbf{Part (8).}
    Let $\xi_j = \delta_j / \|\delta_j\|_2$, and $\varepsilon = \epsilon / \|\epsilon\|_2$. Then
    \begin{align*}
    \begin{split}
        \tr(\P_{X_j}\P_{\epsilon}) 
    &= \frac{|X_j^\top \epsilon|^2 }{ \|X_j\|^2_2 \cdot \|\epsilon\|_2^2 }
    = \frac{ |(X_{k} + \delta_j)^\top \epsilon|^2  }{ \|X_{k} + \delta_j \|^2_2 \cdot \|\epsilon\|^2 }  
    \leq \frac{(|X_{k}^\top \varepsilon| \cdot \|\epsilon\|_2 + |\xi_j^\top \varepsilon| \cdot \|\delta_j\|_2 \cdot \|\epsilon\|_2   )^2  }{ (1 - 2 |X_{k}^\top \xi_j| \cdot \|\delta_j\|_2) \cdot \|\epsilon\|_2^2  }
    \leq \frac{ \eta_0^2 (1 + \eta_1)^2 }{ 1 -  2\eta_1 \eta_0}.
    \end{split}
    \end{align*}

\end{proof}

\vspace{0.1in}
\begin{lemma}\label{lem:proj-decomp}
    Let $\Xi = \spa(\{Z_j\}_{j\in S})$, and $\Delta = \P_{\Xi} - \sum_{j\in S} \P_{Z_j}$. Then 
    $$
    \|\Delta\|_F^2 \leq 2 \sum_{j < l \in S} \tr(\P_{Z_j} \P_{Z_l}).
    $$
\end{lemma}
\begin{proof}
    Note that
    \begin{align*}
    \begin{split}
        \|\Delta\|_F^2 
    &=\tr\left( \P_{\Xi} - \sum_{j\in S} \P_{Z_j} \right)\left( \P_{\Xi} - \sum_{j\in S} \P_{Z_j} \right) \\
    &= \tr(\P_{\Xi}) - 2 \sum_{j\in S} \tr(\P_{\Xi} \P_{Z_j}) + \sum_{j\in S} \tr(\P_{Z_j}) + 2 \sum_{j < l \in S} \tr(\P_{Z_j} \P_{Z_l}) 
    = 2 \sum_{j < l \in S} \tr(\P_{Z_j} \P_{Z_l}).
    \end{split}
    \end{align*}
\end{proof}

\vspace{0.1in}
\begin{lemma}\label{lem:between-cluster}
    Suppose that $Z_1 \in \spa( \{ X_{\mathcal{C}_{k_1}} \})$ and $Z_2 \in \spa( \{ X_{\mathcal{C}_{k_2} } \})$ with $k_1 \neq k_2$ and $|\mathcal{C}_{k_1}| = |\mathcal{C}_{k_1}| = D$. Under the condition of Lemma \ref{lem:cluster-angle}, 
    \begin{enumerate}[label = \normalfont(\arabic*)]
        \item $\left\|\P_{Z_1} \P_{X_{k_2}} \right\|_2 \leq 2D \eta_0$.

        \item $\left\|\P_{Z_1} \P_{\epsilon} \right\|_2 \leq 2D \eta_0$.

        \item $\left\| \P_{Z_1} \P_{Z_2} \right\|_2 \leq 2D^2 (\eta_0^2 + \eta_0) $.
    \end{enumerate}
\end{lemma}
\begin{proof}
    Let $\mathcal{C}_{k_1} = \{{k_1}, {j_1}, \dots, {j_{k_1}}\}$ and $\mathcal{C}_{k_2} = \{{k_2}, {l_1}, \dots, {l_{k_2}}\}$. Denote $\Delta_1 = \P_{\{ X_{k_1}, \delta_{j_1}, \dots, \delta_{j_{k_1}} \}} - \P_{X_{k_1}} - \sum_{i\in\{j_1, \dots, j_{k_1}\} } \P_{\delta_{i}}$, and $\Delta_2 = \P_{\{ X_{k_2}, \delta_{l_1}, \dots, \delta_{l_{k_2}} \}} - \P_{X_{k_2}} - \sum_{i\in\{l_1, \dots, l_{k_2}\} } \P_{\delta_{i}}$. By Lemma \ref{lem:proj-decomp}, we have
    $$
    \|\Delta_1\|_2^2 \leq \|\Delta_1\|_F^2 \leq D^2 \eta_0^2, \quad
    \|\Delta_2\|_2^2 \leq \|\Delta_2\|_F^2 \leq D^2 \eta_0^2.
    $$
    Therefore,
    \begin{align*}
    \begin{split}
        \left\| \P_{Z_1} \P_{X_{k_2}} \right\|_2 
    &\leq \left\| \P_{\{ X_{k_1}, \delta_{j_1}, \dots, \delta_{j_{k_1}} \}} \P_{ X_{k_2} } \right\|_2 
    = \left\| ( \P_{X_{k_1}} + \P_{\delta_{j_1}} + \dots + \P_{\delta_{j_{k_1}}} + \Delta_1 )  \P_{X_{k_2}}  \right\|_2 \\
    &\leq D \eta_0 +  \|\Delta_1\|_2 
    \leq 2D \eta_0 .
    \end{split}
    \end{align*}
    Similarly, $\left\|\P_{Z_1} \P_{\epsilon} \right\|_2 \leq 2D \eta_0$. Moreover,
    \begin{align*}
    \begin{split}
        \left\| \P_{Z_1} \P_{Z_2} \right\|_2
    &\leq \left\| \P_{\{ X_{k_1}, \delta_{j_1}, \dots, \delta_{j_{k_1}} \}} \P_{\{ X_{k_2}, \delta_{l_1}, \dots, \delta_{l_{k_2}} \}} \right\|_2 \\
    &= \left\| ( \P_{X_{k_1}} + \P_{\delta_{j_1}} + \dots + \P_{\delta_{j_{k_1}}} + \Delta_1 ) ( \P_{X_{k_2}} + \P_{\delta_{l_1} } + \dots + \P_{\delta_{l_{k_2}}} + \Delta_2 ) \right\|_2 \\
    &\leq D^2 \eta_0^2 + D ( \|\Delta_1\|_2 + \|\Delta_2\|_2 ) + \|\Delta_1\|_2 \|\Delta_2\|_2 
    \leq 2D^2 (\eta_0^2 + \eta_0).
    \end{split}
    \end{align*}
    
\end{proof}

\vspace{0.1in}
\begin{lemma}\label{lem:eta1-gap}
    Under the condition of Lemma \ref{lem:cluster-angle}, for $j\in \mathcal{V}_k$, $\left\| \P_{X_j^\perp} X_k \right\|_2^2 \geq \frac{\eta_1^2 (1 - \eta_0^2)}{1 + \eta_1^2 + 2\eta_1\eta_0}$.
\end{lemma}
\begin{proof}    
    When $\langle\delta_j, X_k \rangle \leq \pi/2$, let $\langle\delta_j, X_k \rangle = \pi/2 -\theta$ as shown in Figure \ref{fig:project}. Denote $h = \|X_k -\P_{X_j} X_k\|_2$. Then we have $\|\delta_j\|_2^2 \cdot (\cos\theta)^2 \cdot \|X_k\|^2_2 = h^2 \cdot [ \|\delta_j\|_2^2 \cos^2\theta + (\|X_k\|_2 + \|\delta_j\|_2 \sin\theta)^2]$, which implies
    $
    h^2 = \frac{\eta_1^2 \cos^2\theta}{1 + \eta_1^2 + 2\eta_1 \sin\theta}.
    $
    Moreover, note that $\sin^2\theta = \cos^2\langle \delta_j, X_k \rangle \leq \eta_0^2$, we have $h^2 \geq \frac{\eta_1^2 (1 - \eta_0^2)}{1 + \eta_1^2 2 \eta_1 \eta_0}$.
    
    \begin{figure}[!ht]
    \centering
    \includesvg[width=0.45\textwidth]{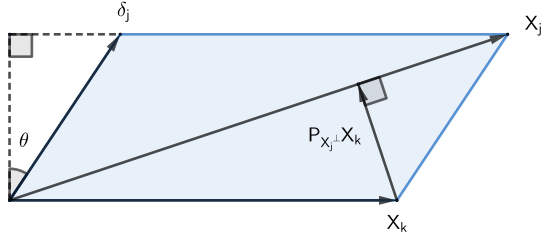}
    \caption{\small The projection of $X_k$ to $X_j$.}
    \label{fig:project}
    \end{figure}
    
    A similar argument can be shown when $\langle\delta_j, X_k \rangle \in (\pi, \pi/2]$.
\end{proof}

\vspace{0.1in}
\begin{lemma}\label{lem:sigmamin-quasi-proj}
    Suppose that $A = \sum_{i=1}^k X_i X_i^\top$ where $X_i\in\mathbb{R}^n$ with $n > k$ and $\|X_i\|_2 = 1$. In addition, $|X_i^\top X_j| \leq\eta_0 $ is small enough for any $i,j = 1, \dots, k$. Then when $\eta_0 < \frac{1}{k^3(k-1) }$,
    $$
    \lambda_k(A) \geq \frac{1 + \sqrt{1 - 4k(k-1)\eta_0 }}{2},\quad
    \text{and} \quad
    \lambda_{j}(A) = 0, \ \forall j > k.
    $$
\end{lemma}
\begin{proof}
    Consider the equality 
    $
    A^\top A = A + \Delta
    $
    where
    \begin{align*}
    \begin{split}
        \| \Delta \|_2 
    &= \left\| A^\top A - A \right\|_2
    = \left\| \sum_{i=1}^k \sum_{j=1}^k X_i X_i^\top X_j X_j^\top - \sum_{i=1}^k X_i X_i^\top \right\|_2
    = \left\| \sum_{i \neq j} X_i X_i^\top X_j X_j^\top \right\|_2 \\
    &\leq 2 \sum_{i<j} \left\|X_i X_i^\top X_j X_j^\top \right\|_2
    \leq 2 \sum_{i<j} \eta_0 = k(k-1)\eta_0.
    \end{split}
    \end{align*}
    Suppose that $A v_i = \lambda_i v_i$, then
    $$
    v_i^\top A^\top A v_i = v_i^\top A v_i + v_i^\top \Delta v_i
    \quad\Longrightarrow\quad
    \lambda_i^2 = \lambda_i + v_i^\top \Delta v_i
    \quad\Longrightarrow\quad
    -\|\Delta\|_2 \leq \lambda_i(\lambda_i - 1) \leq \|\Delta\|_2.
    $$
    This implies that either
    $
    \frac{ 1 + \sqrt{1 - 4\|\Delta\|_2} }{2} \leq \lambda_i \leq \frac{ 1 + \sqrt{1 + 4\|\Delta\|_2} }{2}
    $
    or
    $
    \frac{ 1 - \sqrt{1 + 4\|\Delta\|_2} }{2} \leq \lambda_i \leq \frac{ 1 - \sqrt{1 - 4\|\Delta\|_2} }{2}.
    $
    Note that $\sum_{i=1}^n \lambda_i = \tr(A) = \sum_{i=1}^k \tr(X_i X_i^\top) = k$. Now we proof the range for $\lambda_k(A)$ by contradiction. Suppose that 
    $\lambda_k(A) \leq  \frac{ 1 - \sqrt{1 - 4\|\Delta\|_2} }{2}$, then 
    \begin{align*}
    \begin{split}
        \sum_{i=1}^k \lambda_i 
    &\leq \sum_{i=1}^{k-1} \frac{1 + \sqrt{1 + 4 \|\Delta\|_2}}{2} +  \frac{ 1 - \sqrt{1 - 4 \|\Delta\|_2} }{2} 
    = \frac{k-1}{2} (1 + \sqrt{1 + 4\|\Delta\|_2}) + \frac{1}{2} (1 - \sqrt{1 - 4\|\Delta\|_2}) \\
    &\leq \frac{k}{2} + \frac{k-1}{2} (1 + \sqrt{4\|\Delta\|_2}) - \frac{1}{2} (1 - \sqrt{4\|\Delta\|_2}) 
    \leq k-1 + \frac{k}{2} \sqrt{4 k(k-1)\eta_0 } \\
    &< (k-1) + 1 = k,
    \end{split}
    \end{align*}
    which leads to a contradiction. Hence we must have $\lambda_k(A) \geq \frac{1 + \sqrt{1 - 4\|\Delta\|_2}}{2}$.
    In addition, we know $\rank(A) \leq k$, which implies that $\lambda_j(A) = 0$ for all $j > k$.

\end{proof}

\vspace{0.1in}
\begin{lemma}\label{lem:U-set-property}
    Given any $S \subseteq \widetilde{S} \in \mathcal{S}$ and the set of ``undesired'' directions $\mathcal{U}$ in Definition \ref{def:quality_base_procedure}:
    \begin{enumerate}[label = \normalfont(\arabic*)]
        \item For any features $X_j$ where $j\in \bigcup_{k\in \mathcal{K} \cap N^\star} \mathcal{C}_k$ or $j\in \bigcup_{k\in \mathrm{RF}(S) } \mathcal{C}_k \setminus S$, denote $r_j :=X_j - \P_{X_S} X_j$. Then there must exist $u \in \mathcal{U}$ such that $r_j = u - \P_{X_S} u$ or $-r_j = u - \P_{X_S} u$.

        \item The $\mathcal{U}$ set is far away from $\spa(X_S)$: $\max_{u\in \mathcal{U},\ S\in\mathcal{S}} \tr(\P_u \P_{X_S}) \leq \frac{\eta_1^2}{1 + \eta_1^2} + \mathcal{O}(s^\star \sqrt{\log p / n})$.
    \end{enumerate}
\end{lemma}
\begin{proof}
    \textbf{Part (1).}
    First, we argue that for any ``undesired'' features $X_j$, there must exist $u \in \mathcal{U}$ such that $r_j = u - \P_{X_S} u$ or $-r_j = u - \P_{X_S} u$. (1) If the feature $X_j$ is from a noise cluster, meaning that $j\in\bigcup_{k\in \mathcal{K} \cap N^\star} \mathcal{C}_k$, then it is trivial that $u = X_j \in \bigcup_{k\in \mathcal{K} \cap N^\star} \{ X_j: j\in\mathcal{C}_k \} \subseteq \mathcal{U}$. (2) If the feature $X_j$ is a proxy feature of a selected cluster by $S$ and its representative feature is in $S$, meaning that $S\cap \mathcal{C}_k = \{k\}$ for some $k\in S^\star$ while $j\in \mathcal{V}_k$, then we have $r_j = X_j - \P_{X_S} X_j = (X_k + \delta_j) - \P_{X_S} (X_k + \delta_j) = \delta_j - \P_{X_S} \delta_j $, and hence $u = \delta_j \in \bigcup_{k\in \mathcal{K} \cap S^\star} \{\delta_j: j\in\mathcal{V}_k\} \subseteq \mathcal{U}$. (3) If the feature $X_j$ is the representative feature of a cluster selected by $S$, meaning that $j\in\mathcal{K}$ while $S\cap \mathcal{C}_j = \{l\}\subseteq \mathcal{V}_j$, then we have $r_j = X_j - \P_{X_S} X_j = (X_l - \delta_l) - \P_{X_S} (X_l - \delta_l) = -(\delta_l - \P_{X_S}\delta_l)$, and hence $u = \delta_l \in \bigcup_{k\in \mathcal{K} \cap S^\star} \{\delta_j: j\in\mathcal{V}_k\} \subseteq \mathcal{U}$. (4) If the feature $X_j$ is a proxy feature of a selected cluster by $S$ and its representative feature is not in $S$, meaning that $S\cap \mathcal{C}_k = \{l\} \subseteq\mathcal{V}_k$ for some $k\in S^\star$, $l\neq k$, and $j\in \mathcal{V}_k$ as well, then we have $r_j = X_j - \P_{X_S} X_j = (X_l - \delta_l + \delta_j) - \P_{X_S} (X_l - \delta_l + \delta_j) = (\delta_j - \delta_l) - \P_{X_S} (\delta_j - \delta_l)$, and hence $u = \delta_j - \delta_l \in \bigcup_{k\in S^\star} \{\delta_j - \delta_l: j\neq l \in \mathcal{V}_k \} \subseteq \mathcal{U}$.

    \vspace{0.1in}
    \noindent\textbf{Part (2).}
    We used Lemma \ref{lem:cluster-angle} repeatedly in this proof. Let $\Delta = \P_{X_S} - \sum_{j\in S} \P_{X_j}$ with $\|\Delta\|_F \leq \mathcal{O}(s^\star \sqrt{\log p / n})$ by Lemma \ref{lem:proj-decomp}.
    (1) When $u = \delta_j \in \bigcup_{k\in \mathcal{K} \cap S^\star} \{ \delta_j: j \in \mathcal{V}_k \}$, we have $\tr(\P_u \P_{X_S}) = \sum_{j\in S}\tr(\P_{\delta_j} \P_{X_j}) + \tr(\P_{\delta_j} \Delta) \leq \tr(\P_{\delta_j} \P_{X_j}) + \mathcal{O}(s^\star \sqrt{\log p/ n}) \leq \eta_1^2 / (1 + \eta_1^2) + \mathcal{O}(s^\star \sqrt{\log p / n}) $. (2) When $u = \delta_j - \delta_l  \in \{\delta_j - \delta_l: j\neq l \in \mathcal{V}_k, k\in S^\star\} $, we have $\tr(\P_u \P_{X_S}) \leq \sum_{k\in S} \tr(\P_{\delta_j - \delta_l } \P_{X_k}) + \tr(\P_{\delta_j - \delta_k} \Delta) \leq \eta_1^2 / (2 + 2\eta_1^2) + \mathcal{O}(s^\star \sqrt{\log p / n}) $. (3) When $u = X_j \in \bigcup_{k\in \mathcal{K} \cap N^\star} \{X_j: j \in \mathcal{C}_k\} $, we have $\tr(\P_u \P_{X_S}) = \sum_{k\in S}\tr(\P_{X_j} \P_{X_k}) + \tr(\P_{X_j} \Delta) \leq \mathcal{O}(s^\star \sqrt{\log p/ n})$. The conclusion then follows.
    
\end{proof}

\vspace{0.2in}
\subsubsection{Proof of Theorem \ref{thm:1}}
\label{sec:sup-proofthm1}

        
        



\begin{proof}
    In this proof, let $\widehat{S}$ be the selection set returned by FSSS.
    We partition the feature set into four disjoint subsets: $\{1, \dots, p\} =W_\mathcal{V} \cup W_1 \cup W_2 \cup S^\star $, each defined below. For each feature, let $w_j$ denote its associated “direction”, and define $\xi_j$ as the projection of $w_j$ onto either $\spa(X_{\widehat{S}})$ or its orthogonal complement, whichever is closer to $w_j$.
    Finally, note that the notations in Table \ref{tab:notation-W} are aligned with $W$ and $\mathcal{W}$ defined in Section \ref{sec:sup-bpquality}, with the following relationship: $W= W_{\mathcal{V}} \cup W_1 \cup W_2 $, and $\mathcal{W} = \{w_j\}_{j\in W}$.
    
\begin{table}[!ht]
\small
\centering
\begin{tabular}{c|c|c|c}
\hline
                                & Set notation                                                                   & $w_j$      & $\xi_j$                     \\ \hline
Perturbation directions         & $W_\mathcal{V}= \mathcal{V}$                                                   & $\delta_j$ & $\P_{X_{\widehat{S}}^\perp} w_j$ \\
Selected representative features in noise groups   & $W_1:= \{ k\in\mathcal{K} \cap N^\star: \mathcal{C}_k \cap \widehat{S} \neq \varnothing \}$ & $X_j$      & $\P_{X_{\widehat{S}}} w_j$       \\
Unselected representative features in noise groups & $W_2:= \{ k\in\mathcal{K} \cap N^\star: \mathcal{C}_k \cap \widehat{S} = \varnothing \}$    & $X_j$      & $\P_{X_{\widehat{S}}^\perp} w_j$ \\
Signals                         & $S^\star$                                                                      & $X_j$      & Undefined                   \\ \hline
\end{tabular}
\caption{\small Notation illustration for proof of Theorem \ref{thm:1}.}
\label{tab:notation-W}
\end{table}

    \underline{Step 1: show that $| \mathcal{C}_k \cap \widehat{S} | \in \{0,1\}$ for all $k\in\mathcal{K}$}. Note that for any $j,k \in \mathcal{C}_k$, by condition (ii), we have $\tr(\P_{X_j} \P_{X_k}) \geq \frac{1}{ (1 + \eta_1^2)^2 } + \mathcal{O} (\sqrt{\log p / n}) > \frac{1}{2}  $ by Lemma \ref{lem:cluster-angle} part (5). 
    Therefore, under condition (i), at most one feature per cluster can be selected by $\widehat{S}$.

    \vspace{0.1in}
    \underline{Step 2: find the FPE upper bound}. 
    We start by bounding a term associated with the slackness arising from decomposing $\P_{{X_S^\star}^\perp}$ into the sum of projections onto individual vectors.
    Let $\Delta^\star := \P_{X_{S^\star}^\perp} - \sum_{j\in W} \P_{w_j}$. Then 
    \begin{align*}
    \begin{split}
        &\sum_{j\in W_2 \cup W_\mathcal{V}} \tr(\P_{X_{\widehat{S}}} \P_{w_j} ) + \tr(\P_{X_{\widehat{S}}} \Delta^\star)
    =\tr( \P_{X_{\widehat{S}}} ( \P_X - \P_{X_S^\star} - \sum_{j\in W_1} \P_{w_j} ) ) \\
    \leq{}& \tr(\P_{X_{\widehat{S}}}) - \sum_{j\in W_1} \tr(\P_{X_{\widehat{S}}} \P_{w_j} ) \leq |\widehat{S} | \left[1 - \min_{j\in W_1} \tr(\P_{X_{\widehat{S}}} \P_{w_j}) \right]  
    \leq s_0  \left[\frac{\eta_1^2}{1 + \eta_1^2} + \mathcal{O}\left( \sqrt{ \frac{\log p}{n} } \right) \right] = a.
    \end{split}
    \end{align*}
    Next, we bound a term that characterizes the distance between $\tr(\P_{w_j} \P_{X_{\widehat{S}^{(\ell)}}})$ and $\tr(\P_{\xi_j} \P_{X_{\widehat{S}^{(\ell)}}})$. Note that
    $
    \min_{j\in W_1} \tr(\P_{\xi_j} \P_{w_j}) \geq \frac{1}{1+\eta_1^2} - \mathcal{O} (s_0\sqrt{\log p / n} )
    $,
    $
    \min_{j\in W_2} \tr(\P_{\xi_j} \P_{w_j}) \geq 1 - \mathcal{O} (\sqrt{\log p / n} )
    $, and
    $
    \min_{j\in W_\mathcal{V}} \tr(\P_{\xi_j} \P_{w_j}) \geq \frac{1}{1 + \eta_1^2} - \mathcal{O} (s_0 \sqrt{\log p / n} )
    $.
    Therefore, 
    for any $j\in W$ and $\ell\in\{1, \dots, B\}$, we have 
    \begin{align*}
    \begin{split}
       & \tr(\P_{\xi_j} \P_{X_{\widehat{S}^{(\ell)}}}) - \tr(\P_{w_j} \P_{X_{\widehat{S}^{(\ell)}}} )
    = \tr( (\P_{\xi_j} - \P_{w_j}) \P_{X_{\widehat{S}^{(\ell)}}} )
    \leq 2 \|\P_{\xi_j} - \P_{w_j}\|_2 \\
    \leq{}& 2 \sqrt{1 - \min_{j\in W} \tr(\P_{\xi_j } \P_{w_j})} 
    \leq 2 \sqrt{1 - \frac{1}{1 + \eta_1^2} + \mathcal{O}\left(s_0 \sqrt{\frac{\log p}{n}} \right) }  
    = 2  \sqrt{\frac{\eta_1^2}{1 + \eta_1^2}} + \mathcal{O}\left(s_0 \sqrt{\frac{\log p}{n}} \right) = b.
    \end{split}
    \end{align*} 
    As a result, for $\alpha \in (1/2,1)$,
    \begin{align*}
    \begin{split}
    &\mathbb{E}\left[ \tr(\P_{X_{\widehat{S}}} \P_{X_{S^\star}^\perp} ) \right] 
    = \mathbb{E}\left[ \sum_{j\in W} \tr( \P_{X_{\widehat{S}}} \P_{w_j} ) + \tr(\P_{X_{\widehat{S}}} \Delta^\star) \right]  \\
    \leq{}& \mathbb{E}\left[ \sum_{j\in W_1} \tr( \P_{X_{\widehat{S}}} \P_{\xi_j} ) + \sum_{j\in W_2 \cup W_\mathcal{V}} \tr( \P_{X_{\widehat{S}}} \P_{w_j} )  + \tr(\P_{X_{\widehat{S}}} \Delta^\star) \right] \\
    \leq{}& \sum_{j\in W}  \mathbb{E}\left[ \tr( \P_{X_{\widehat{S}}} \P_{\xi_j} ) \right] + a 
    = \sum_{j\in W} \mathbb{E}\left[ \mathds{1}( \xi_j \in \spa(X_{\widehat{S}}) ) \right] + a \\
    \leq{}& \sum_{j\in W} \mathbb{E}\left[ \mathds{1}( \tr(\P_{\xi_j} \P_{\rm avg}) \geq \alpha ) \right] + a \\
    \leq{}& \sum_{j\in W} \mathbb{P}\left[ \frac{1}{B/2} \sum_{\ell=1}^{B/2} \prod_{i\in\{0,1\}} \tr(\P_{\xi_j} \P_{X_{\widehat{S}^{(2\ell-i)}} }) \geq 2\alpha-1  \right] + a\\
    \leq{}& \sum_{j\in W} \frac{1}{2\alpha - 1}  \frac{1}{B/2} \sum_{\ell=1}^{B/2}  \mathbb{E}\left[ \tr(\P_{\xi_j} \P_{X_{\widehat{S}^{(2\ell)}}})  \tr(\P_{\xi_j} \P_{X_{\widehat{S}^{(2\ell-1)}}}) \right] + a \\
    \leq{}& \sum_{j\in W}\frac{1}{2\alpha - 1}  \frac{1}{B/2} \sum_{\ell=1}^{B/2}  \mathbb{E}\left[ \left( \tr(\P_{w_j} \P_{X_{\widehat{S}^{(2\ell)}}}) + b \right) \left( \tr(\P_{w_j} \P_{X_{\widehat{S}^{(2\ell-1)}}}) + b \right) \right] + a \\
    \leq{}& \frac{1}{2\alpha - 1} \sum_{j\in W} \left\{  \sup_{T\in\mathscr{T}} \mathbb{E}\left[  \tr(\P_{w_j} \P_{X_{\widehat{S}(T)}} )  \right] + b \right\}^2 + a.
    \end{split}
    \end{align*}
The first inequality follows from $\tr(\P_{X_{\widehat{S}}}\P_{w_j} ) \leq \tr(\P_{X_{\widehat{S}}}\P_{\xi_j} ) = 1$ for any $j\in W_1$. The third inequality follows from the fact that, for any $j\in W$, $\tr(\P_{X_{\widehat{S}}}\P_{\xi_j} )$ takes value only 0 or 1; if $\tr(\P_{X_{\widehat{S}}}\P_{\xi_j} ) = 1$, then we must have $\xi_j \in \spa(X_{\widehat{S}})$, implying that $\tr(\P_{\xi_j} \P_{\rm avg} ) \geq \alpha$. The fourth inequality follows from the fact that, if $\tr(\P_{\xi_j} \P_{\rm avg}) \geq \alpha$, then
\begin{align*}
\begin{split}
    \frac{1}{B/2} \sum_{\ell=1}^{B/2} \prod_{i\in\{0,1\}} \tr (\P_{\xi_j} \P_{\widehat S^{(2\ell-i)}}) 
    &\geq \frac{1}{B/2} \sum_{\ell=1}^{B/2} \sum_{i\in\{0,1\} } \tr (\P_{\xi_j} \P_{\widehat S^{(2\ell-i)}}) - 1 
    = \frac{2}{B} \sum_{l=1}^B \tr (P_{\xi_j} P_{\widehat S^{(\ell)}}) - 1 
    \geq 2\alpha - 1.
\end{split}
\end{align*}
The fifth inequality follows from the Markov inequality. 

Finally, we conclude that
\begin{align*}
\begin{split}
    \mathbb{E}\left[ \tr(\P_{X_{\widehat{S}}} \P_{X_{S^\star}^\perp} ) \right]
&\leq \frac{1}{2\alpha - 1} \sum_{j\in W} \left\{  \sup_{T\in\mathscr{T}} \mathbb{E}\left[  \tr(\P_{w_j} \P_{X_{\widehat{S}(T)}} )  \right] + b \right\}^2 + a \\
&\leq \frac{1}{2\alpha - 1} |W| (\gamma_\mathcal{\mathcal{W}} + b)^2 + a
\leq \frac{p (\gamma_\mathcal{\mathcal{W}} + b)^2 }{2\alpha - 1} + a.
\end{split}
\end{align*}

\end{proof}

\vspace{0.2in}
\subsubsection{Proof of Theorem \ref{thm:2}}
\label{sec:sup-proofthm2}

The notation $\mathrm{RF}(S)$ used in this proof is defined at the beginning of Section \ref{sec:sup-proof-cluster}.

\begin{proof}
By condition (i), we know that for all $\ell\in \{1, \dots, B\}$ and $k\in S^\star$, $\widehat{S}^{(\ell)} \cap \mathcal{C}_k \neq \varnothing$; that is, each base procedure $\widehat{S}^{(\ell)}$ selects at least one feature per signal cluster. Now consider a set of new basis vectors $\{Z_j^{(\ell)}\}_{j\in \widehat{S}^{(\ell)}} $ for $\spa(X_{\widehat{S}^{(\ell)}} )$ and construct its ``quasi-representative features'' $\mathfrak{K}^{(\ell)}$ as follows: (1) If $|\widehat{S}^{(\ell)} \cap \mathcal{C}_k| = 1 $ for some $k$ and $j = \widehat{S}^{(\ell)} \cap \mathcal{C}_k$, then take $Z^{(\ell)}_j = X_j$ and $\mathfrak{K}^{(\ell)} \cap \mathcal{C}_k = \{j\}$; (2) If $|\widehat{S}^{(\ell)} \cap \mathcal{C}_k| > 1$ for some $k$ and $\widehat{S}^{(\ell)} \cap \mathcal{C}_k = \{ k, j_1, \dots j_k \}$, then take $\left[ Z_k^{(\ell)}, Z_{j_1}^{(\ell)}, \dots, Z_{j_k}^{(\ell)} \right] = \left[ X_k, \delta_{j_1}, \dots, \delta_{j_k} \right]$, and $\mathfrak{K}^{(\ell)} \cap \mathcal{C}_k = \{k\}$; (3) If $|\widehat{S}^{(\ell)} \cap \mathcal{C}_k| > 1$ for some $k$, and $\widehat{S}^{(\ell)} \cap \mathcal{C}_k = \{j_1, \dots j_k \}$ where $k \not\in \widehat{S}^{(\ell)} \cap \mathcal{C}_k$, then take $\left[ Z^{(\ell)}_{j_1}, Z^{(\ell)}_{j_2} \dots, Z_{j_k}^{(\ell)} \right] = \left[ X_{j_1}, \P_{X_{j_1}^\perp} X_{j_2}, \dots, \P_{\{ X_{j_1}, \dots, X_{j_{k-1}} \}^\perp} X_{j_k} \right]$, and $\mathfrak{K}^{(\ell)} \cap \mathcal{C}_k = \{ j_1 \}$. We define the ``quasi-proxies'' as $\mathfrak{V}^{(\ell)} = \widehat{S}^{(\ell)} \setminus \mathfrak{K}^{(\ell)}$.

Note that for each $k\in \mathrm{RF}(\widehat{S}^{(\ell)})$, the representative features and its ``quasi'' version should be very close. Define $q_k^{(\ell)} := \mathcal{C}_k \cap \mathfrak{K}^{(\ell)}$, then by Lemma \ref{lem:cluster-angle}, $\|\P_{X_k} - \P_{Z_{q_k^{(\ell)}}} \|_2 \leq \sqrt{  1 - \tr(\P_{X_k}  \P_{Z_{q_k^{(\ell)}}} )} \leq \sqrt{\frac{\eta_1^2}{1 + \eta_1^2}} + \mathcal{O}(\sqrt{\log p / n}) $.

\vspace{0.1in}
\underline{Step 1: find a lower bound for $\sigma_{|S|} (\P_{X_S} \P_{\rm avg}\P_{X_S} )$ for any $S \subseteq \widetilde{S} \in \mathcal{S}$}. 

Let $\Delta_S^{(\ell)} = \P_{X_{\widehat{S}^{(\ell)}}} - \sum_{k\in \mathrm{RF}(\widehat{S}^{(\ell)}) \cap \mathrm{RF}(S) } \P_{X_k} - 
\sum_{k\in \mathrm{RF}(\widehat{S}^{(\ell)}) \setminus \mathrm{RF}(S)} \P_{Z_{q_k}^{(\ell)}} - \sum_{j\in \mathfrak{V}^{(\ell)}} \P_{Z_j^{(\ell)}}$. We have $\|\Delta_S^{(\ell)} \|_2 \leq \left\| \P_{X_{\widehat{S}^{(\ell)}}} - \sum_{j \in \widehat{S}^{(\ell)} } \P_{Z_j^{(\ell)}} \right\|_2 + \left\| \sum_{k\in \mathrm{RF}(\widehat{S}^{(\ell)}) \cap \mathrm{RF}(S)} \P_{X_k} - \sum_{k\in \mathrm{RF}(\widehat{S}^{(\ell)}) \cap \mathrm{RF}(S)}  \P_{ Z_{q_k^{(\ell)}}}\right\|_2 
\leq s^\star \sqrt{\frac{\eta_1^2}{1 + \eta_1^2}} + \mathcal{O} (s_0 D^2\sqrt{\log p / n} ) $.
Moreover, let $\Delta = \P_{X_S} - \sum_{k\in \mathrm{RF}(S) } \P_{X_k} $, and $\|\Delta\|_2 \leq \| \P_{X_S} - \sum_{j\in S}\P_{X_j} \|_2 + \|\sum_{j\in S} \P_{X_j} - \sum_{k\in \mathrm{RF}(S) } \P_{X_k} \|_2 \leq s^\star \sqrt{\frac{ \eta_1^2 }{1 + \eta_1^2}} + \mathcal{O}(s^\star \sqrt{\log p / n}) $. Therefore,
\begin{align*}
\begin{split}
   &  \P_{X_S} \P_{\rm avg} \P_{X_S}  \\
={}& \left( \sum_{k\in \mathrm{RF}(S) } \P_{X_k} + \Delta \right) 
\frac{1}{B}\sum_{\ell=1}^B \left( \sum_{k\in \mathrm{RF}(S)} \P_{X_k} + \sum_{k\in \mathrm{RF}(\widehat{S}^{(\ell)}) \setminus \mathrm{RF}(S)} \P_{Z_{q_k^{(\ell)}}} + \sum_{j\in \mathfrak{V}^{(\ell)}} \P_{Z_j^{(\ell)}} + \Delta_S^{(\ell)} \right)
\left( \sum_{k\in \mathrm{RF}(S) } \P_{X_k} + \Delta \right)  \\
={}& \sum_{k\in \mathrm{RF}(S) } \P_{X_k} + \widetilde{\Delta},
\end{split}
\end{align*}
where $\| \widetilde{\Delta} \|_2 \leq \sum_{i=1}^{12} \mathfrak{T}_i(s_0, n, p) $. Specifically,

\noindent $
\mathfrak{T}_1 = \left\| \sum_{j,k,l\in \mathrm{RF}(S)^3 \setminus \{j = k= l\} } \P_{X_i} \P_{X_j} \P_{X_l} \right\|_2 \leq (s^\star)^3 \mathcal{O}(\sqrt{\log p / n})
$; 

\noindent $\mathfrak{T}_2 = \left\| \sum_{j,k \in \mathrm{RF}(S)} \P_{X_j} \frac{1}{B} \sum_{\ell=1}^B \left( \sum_{k\in \mathrm{RF}(\widehat{S}^{(\ell)}) \setminus \mathrm{RF}(S)} \P_{Z_{q_k}^{(\ell)}} + \sum_{j \in \mathfrak{V}^{(\ell)}} \P_{Z_j^{(\ell)}} \right) \P_{X_k} \right\|_2 \\
\leq s^\star s_0 \sqrt{\frac{\eta_1^2}{1 + \eta_1^2}} + \mathcal{O}((s^\star)^2 s_0 D \sqrt{  \log p / n }) $; 

\noindent $\mathfrak{T}_3 = \left\| \sum_{j,k \in \mathrm{RF}(S)} \P_{X_j} \P_{X_k} \Delta \right\|_2 \leq (s^\star)^2 \|\Delta\|_2 \leq (s^\star)^3 \sqrt{\frac{\eta_1^2}{1 + \eta_1^2}} + \mathcal{O}((s^\star)^3 \sqrt{\log p / n} ) $;

\noindent $\mathfrak{T}_4 = \left\| \sum_{j \in \mathrm{RF}(S)} \P_{X_j} \frac{1}{B} \sum_{\ell=1}^B \left( \sum_{k\in \mathrm{RF}(\widehat{S}^{(\ell)}) \setminus \mathrm{RF}(S)} \P_{Z_{q_k}^{(\ell)}} + \sum_{j \in \mathfrak{V}^{(\ell)}} \P_{Z_j^{(\ell)}} \right) \Delta \right\|_2 \leq  s^\star s_0 \|\Delta\|_2 \leq (s^\star)^2 s_0 \sqrt{\frac{\eta_1^2}{1 + \eta_1^2}} + \mathcal{O}((s^\star)^2 s_0 \sqrt{\log p / n} )$;

\noindent $\mathfrak{T}_5 = \left\| \sum_{j,k\in\mathrm{RF}(S)} \P_{X_j} \left(\frac{1}{B} \sum_{\ell=1}^B \Delta_S^{(\ell)} \right) \P_{X_k} \right\|_2 \leq (s^\star)^2 \frac{1}{B}\sum_{\ell=1}^B \|\Delta_S^{(\ell)}\|_2 \leq (s^\star)^3 \sqrt{\frac{\eta_1^2}{1 + \eta_1^2}} + \mathcal{O}( (s^\star)^2 s_0 D^2 \sqrt{\log p / n} )$;

\noindent $\mathfrak{T}_6 = \left\| \sum_{j\in\mathrm{RF}(S)} \P_{X_j} \left(\frac{1}{B} \sum_{\ell=1}^B \Delta_S^{(\ell)} \right) \Delta \right\|_2 \leq s^\star \left(\frac{1}{B} \sum_{\ell=1}^B \| \Delta_S^{(\ell)} \|_2 \right) \|\Delta\|_2 \leq (s^\star)^3 \frac{\eta_1^2}{1 + \eta_1^2} + \mathcal{O}( (s^\star)^2 s_0 D^2 \sqrt{\log p / n} )$;

\noindent $\mathfrak{T}_7 = \left\| \sum_{j,k\in\mathrm{RF}(S)} \Delta \P_{X_j} \P_{X_k} \right\|_2 \leq (s^\star)^2 \|\Delta\|_2 \leq (s^\star)^3 \sqrt{\frac{\eta_1^2}{1 + \eta_1^2}} + \mathcal{O}((s^\star)^3 \sqrt{\log p / n} ) $;

\noindent $\mathfrak{T}_8 = \left\|\sum_{j\in\mathrm{RF}(S)}  \Delta \P_{X_j} \Delta \right\|_2 \leq s^\star \|\Delta\|_2^2 \leq (s^\star)^3 \frac{\eta_1^2}{1 + \eta_1^2} + \mathcal{O}( (s^\star)^3 \sqrt{\log p / ns} )$;

\noindent $\mathfrak{T}_9 = \left\| \sum_{j\in\mathrm{RF}(S)}  \Delta \left( \sum_{k\in \mathrm{RF}(\widehat{S}^{(\ell)}) \setminus \mathrm{RF}(S)} \P_{Z_{q_k}^{(\ell)}} + \sum_{j \in \mathfrak{V}^{(\ell)}} \P_{Z_j^{(\ell)}} \right)  \P_{X_j} \right\|_2 \leq s^\star s_0  \|\Delta\|_2
\\
\leq (s^\star)^2 s_0 \sqrt{\frac{\eta_1^2}{1 + \eta_1^2}} + \mathcal{O}((s^\star)^2 s_0 \sqrt{\log p / n} )$;

\noindent $\mathfrak{T}_{10} = \left\|  \Delta \left( \sum_{k\in \mathrm{RF}(\widehat{S}^{(\ell)}) \setminus \mathrm{RF}(S)} \P_{Z_{q_k}^{(\ell)}} + \sum_{j \in \mathfrak{V}^{(\ell)}} \P_{Z_j^{(\ell)}} \right)  \Delta \right\|_2 \leq s_0 \|\Delta\|_2^2 
\leq s_0 (s^\star)^2 \frac{\eta_1^2}{1 + \eta_1 ^2} + \mathcal{O} ((s^\star)^2 s_0 \sqrt{\log p / n } )
$;

\noindent $\mathfrak{T}_{11} = \left\| \sum_{j\in\mathrm{RF}(S)} \Delta \left(\frac{1}{B}\sum_{\ell=1}^B \Delta_S^{(\ell)} \right)  \P_{X_k} \right\|_2 \leq s^\star \|\Delta\|_2 \left(\frac{1}{B} \sum_{\ell=1}^B \|\Delta_S^{(\ell)}\|_2 \right)
\leq (s^\star)^3 \frac{\eta_1^2}{1 + \eta_1^2} + \mathcal{O}( (s^\star)^2 s_0 D^2 \sqrt{\log p / n} )
$;

\noindent $\mathfrak{T}_{12} = \left\| \Delta \left(\frac{1}{B} \sum_{\ell=1}^B \Delta_S^{(\ell)} \right) \Delta \right\|_2 \leq \|\Delta\|_2^2 \left(\frac{1}{B} \sum_{\ell=1}^B \|\Delta_S^{(\ell)}\|_2 \right) 
\leq (s^\star)^3 (\frac{\eta_1^2}{1 + \eta_1^2})^{3/2} + \mathcal{O}( (s^\star)^2 s_0 D^2 \sqrt{\log p / n} ).
$
Combining them together, we have $\|\widetilde{\Delta}\|_2 \leq (s^\star s_0 + 2(s^\star)^2 s_0 + 3(s^\star)^3) \sqrt{\frac{\eta_1^2}{1 + \eta_1^2}} + ((s^\star)^2 s_0 + 3 (s^\star)^3 ) \frac{\eta_1^2}{1 + \eta_1^2} + (s^\star)^3 ( \frac{\eta_1^2}{1 + \eta_1^2} )^{3/2} + \mathcal{O}((s^\star)^2 s_0 \sqrt{\log p / n}) $. 

As a result, by Weyl's inequality, 
$
\sigma_{|S|} (\P_{X_S} \P_{\rm avg} \P_{X_S}) \geq \sigma_{|S|}( \sum_{k\in \mathrm{RF}(S) } \P_{X_k} ) - \|\widetilde{\Delta}\|_2 \geq 1 - (s^\star s_0 + 2(s^\star)^2 s_0 + 3(s^\star)^3) \sqrt{\frac{\eta_1^2}{1 + \eta_1^2}} - ((s^\star)^2 s_0 + 3 (s^\star)^3 ) \frac{\eta_1^2}{1 + \eta_1^2} - (s^\star)^3 ( \frac{\eta_1^2}{1 + \eta_1^2} )^{3/2} + \mathcal{O}((s^\star)^2 s_0 \sqrt{\log p / n}) =: L_3,
$
where the last inequality follows from Lemma \ref{lem:sigmamin-quasi-proj} and condition (iii).

\vspace{0.1in}
\underline{Step 2: find an upper bound for the stability of newly added direction} when adding an ``undesired feature'' to $S \subseteq \widetilde{S} \in \mathcal{S}$. The aim is to ensure that no extra features in a selected cluster by $S$, or features in a noise cluster can be added to $S$ with high probability.  

Define $r_j :=X_j - \P_{X_S} X_j$. By Lemma \ref{lem:U-set-property}, we know that any ``undesired'' features $X_j$, there must exist $u \in \mathcal{U}$ such that $r_j = u - \P_{X_S} u$ or $-r_j = u - \P_{X_S} u$. Now for any ``undesired'' feature $X_j$ and any subsample $\ell\in\{1, \dots, B\}$, consider
$$
\tr(\P_{X_{\widehat{S}^{(\ell)}}} \P_{r_j}) 
= \frac{ r_j^\top \P_{X_{\widehat{S}^{(\ell)}}} r_j }{r_j^\top r_j}
= \frac{ (u - \P_{X_S} u)^\top \P_{X_{\widehat{S}^{(\ell)}}} (u - \P_{X_S} u) }
{ (u - \P_{X_S} u)^\top (u - \P_{X_S} u) }.
$$
Define $\mathfrak{K}_S^{(\ell)} := \bigcup_{k\in\mathrm{RF}(S)} q_k^{(\ell)} $. Let $\nabla = \P_{X_{\mathfrak{K}_S^{(\ell)}} } - \P_{X_S}$, then $\|\nabla\|_F \leq \|\P_{X_S} - \sum_{j\in  S} \P_{X_j}\|_F + \|\P_{X_{\mathfrak{K}_S^{(\ell)}}} - \sum_{j\in \mathfrak{K}_S^{(\ell)}} \P_{X_j}\|_F + \| \sum_{j\in  S} \P_{X_j} - \sum_{k\in \mathfrak{K}_S^{(\ell)}} \P_{X_k} \|_F \leq 2s^\star \frac{ \sqrt{ \eta_1^2 (\eta_1^2 + 2) } }{ \eta_1^2 + 1 } + \mathcal{O}(s^\star \sqrt{\log p / n}) $. For the numerator, 
\begin{align*}
\begin{split}
    & (u - \P_{X_S} u)^\top \P_{X_{\widehat{S}^{(\ell)}}} (u - \P_{X_S} u) 
    = u^\top \P_{X_{\widehat{S}^{(\ell)}}} u - 2 u^\top \P_{X_S} \P_{X_{\widehat{S}^{(\ell)}}} u + u^\top \P_{X_S} \P_{X_{\widehat{S}^{(\ell)}}} \P_{X_S} u  \\
={}&   u^\top \P_{X_{\widehat{S}^{(\ell)}}} u - 2 u^\top (\P_{X_{\mathfrak{K}_S^{(\ell)}} } - \nabla ) \P_{X_{\widehat{S}^{(\ell)}}} u + u^\top (\P_{X_{\mathfrak{K}_S^{(\ell)}} } - \nabla ) \P_{X_{\widehat{S}^{(\ell)}}} (\P_{X_{\mathfrak{K}_S^{(\ell)}} } - \nabla ) u  \\
={}& u^\top \P_{X_{\widehat{S}^{(\ell)}}} u - u^\top \P_{X_{\mathfrak{K}_S^{(\ell)}} } u + 2 u^\top \nabla \P_{X_{\widehat{S}^{(\ell)}}} u - 2 u^\top \nabla \P_{X_{\mathfrak{K}_S^{(\ell)}} } u + u^\top \nabla \P_{X_{\widehat{S}^{(\ell)}}} \nabla u  \\
={}& u^\top \P_{X_{\widehat{S}^{(\ell)}}} u - u^\top \P_{X_S} u - u^\top\nabla u + 2 u^\top \nabla \P_{X_{\widehat{S}^{(\ell)}}} u - 2 u^\top \nabla \P_{X_{\mathfrak{K}_S^{(\ell)}} } u + u^\top \nabla \P_{X_{\widehat{S}^{(\ell)}}} \nabla u \\
\leq{}& u^\top \P_{X_{\widehat{S}^{(\ell)}}} u - u^\top \P_{X_S} u  + 5\|\nabla\|_F \|u\|_2^2 + \|\nabla\|^2_F \|u\|_2^2.
\end{split}
\end{align*}
Therefore,
\begin{align*}
\begin{split}
    \tr(\P_{X_{\widehat{S}^{(\ell)}}} \P_{r_j}) 
&\leq \frac{ \tr(\P_u \P_{X_{\widehat{S}^{(\ell)}}} ) - \tr( \P_u \P_{X_S} )  }{ 1 - \tr( \P_u \P_{X_S} ) } + \frac{ 5\|\nabla\|_F + \|\nabla\|_F^2 }{ 1 - \tr( \P_u \P_{X_S} )  } \\
&\leq \tr(\P_u \P_{X_{\widehat{S}^{(\ell)}}} )  + 10 \|\nabla\|_F + 2 \|\nabla\|_F^2 \\
&\leq \tr(\P_u \P_{X_{\widehat{S}^{(\ell)}}} ) + 20 s^\star \frac{ \sqrt{\eta_1^2 (\eta_1^2 + 2)} }{ \eta_1^2 + 1 } + 8(s^\star)^2 \frac{ \eta_1^2 (\eta_1^2 + 2) }{(1 + \eta_1^2)^2} + \mathcal{O}((s^\star)^2 \sqrt{\log p / n}),
\end{split}
\end{align*}
where the second inequality follows from the fact that $\tr(\P_u \P_{X_{\widehat{S}^{(\ell)}}} ) \leq 1$ and $\max_{u\in\mathcal{U}}\tr( \P_u \P_{X_S} ) \leq \frac{1}{2}$ by Lemma \ref{lem:U-set-property} and condition (ii). 

Finally, for the given $\alpha_0\in (0, \frac{1}{2})$, we have
\begin{align*}
\begin{split}
    &\mathbb{P}\left[\frac{1}{B} \sum_{\ell=1}^B \tr(\P_u \P_{X_{\widehat{S}^{(\ell)}}} )  \geq 1-\alpha_0 \right] 
    \leq \mathbb{P}\left[\frac{1}{B/2} \sum_{\ell=1}^{B/2} \prod_{i\in\{0,1\} } \tr(\P_u \P_{X_{\widehat{S}^{(2\ell-i)}}} )  \geq 1-2\alpha_0 \right] \\
    \leq{}& \frac{1}{1 - 2\alpha_0} \frac{1}{B/2} \sum_{\ell=1}^{B/2} \mathbb{E}\left[ \tr(\P_u \P_{X_{\widehat{S}^{(2\ell)}}}) \right] \mathbb{E}\left[ \tr(\P_u \P_{X_{\widehat{S}^{(2\ell-1)}}}) \right] 
    \leq \frac{\gamma_\mathcal{U}^2}{1 - 2\alpha_0},
\end{split}
\end{align*}
and hence $\mathbb{P}\left[ \max_{u \in \mathcal{U}} \frac{1}{B} \sum_{\ell=1}^B \tr(\P_u \P_{X_{\widehat{S}^{(\ell)}}} )  \geq 1-\alpha_0 \right] \leq |\mathcal{U}| \gamma_\mathcal{U}^2 / (1 - 2\alpha_0) = \left( s^\star \binom{D-1}{2} + p - s^\star \right) \gamma_\mathcal{U}^2 / (1 - 2\alpha_0)$.

In summary, with probability at least $1 - \left( s^\star \binom{D-1}{2} + p - s^\star \right) \gamma_\mathcal{U}^2 / (1 - 2\alpha_0)$, for any ``undesired'' feature $X_j$, we have
\begin{align*}
\begin{split}
    \tr(\P_{r_j} \P_{\rm avg}) 
    &\leq \frac{1}{B} \sum_{\ell=1}^B \max_{u\in \mathcal{U}} \tr(\P_{u} \P_{X_{\widehat{S}^{(\ell)}}} ) + 20 s^\star \frac{ \sqrt{\eta_1^2 (\eta_1^2 + 2)} }{ \eta_1^2 + 1 } + 8(s^\star)^2 \frac{ \eta_1^2 (\eta_1^2 + 2) }{(1 + \eta_1^2)^2} + \mathcal{O}((s^\star)^2 \sqrt{\log p / n})  \\
    &\leq 1 - \alpha_0 + 20 s^\star \frac{ \sqrt{2\eta_1^2 (\eta_1^2 + 2)} }{ \eta_1^2 + 1 } + 8 (s^\star)^2 \frac{\eta_1^2 (\eta_1^2 + 2) }{(1 + \eta_1^2)^2} + \mathcal{O}((s^\star)^2 \sqrt{\log p / n}) =:U_3.
\end{split}
\end{align*}

\underline{Step 3: Draw the conclusion}. 
When $U_3 < \alpha < L_3$, the only selection sets can be returned by our algorithm are those in $\{S: S \subseteq \widetilde{S} \in \mathcal{S}\}$. Note that $ \sigma_{|S|}(\P_{X_S} \P_{\rm avg} \P_{X_S} ) \leq \inf_{Z\in \spa(X_S)}\tr(\P_{Z} \P_{\rm avg})$, implying that for any $S \in\mathcal{S}$, both the ``new direction condition'' and the ``$\sigma_{\min}$-condition" will be satisfied and hence only $S$ rather than its subsets will be returned.

\end{proof}

\vspace{0.2in}
\subsubsection{Proof of Theorem \ref{prop:clust-consist-s0}}
\label{sec:sup-proofthm3}
The notation $\mathrm{RF}(S)$ used in this proof is defined at the beginning of Section \ref{sec:sup-proof-cluster}.

\begin{proof}
    We begin by giving a brief outline of the proof. Note that for any given selection set $S$, the objective function of $\ell_0$-penalized regression can be re-written as its prediction error $\mathrm{pe}(S)$:
    \begin{align*}
    \begin{split}
        \mathrm{pe}(S) 
    &= \left\| \P_{{X_S}^\perp} y \right\|_2^2 
    = \left\| \sum_{k\in\mathcal{K}} \beta_k^\star  \P_{{X_S}^\perp} X_k +  \P_{{X_S}^\perp} \epsilon \right\|_2^2 \\
    &= \sum_{k\in S^\star} \left| \beta^\star_k \right|^2 \left\| \P_{X_S^\perp}X_k \right\|_2^2 + \left\|\P_{X_S^\perp} \epsilon \right\|_2^2 + 2 \sum_{k\neq l \in S^\star} \beta_k^\star \beta_l^\star X_k^\top \P_{X_S^\perp} X_l + 2 \sum_{k\in S^\star} \beta^\star_k \epsilon^\top \P_{X_S^\perp} X_k.
    \end{split}
    \end{align*}
    Notice that $\mathrm{pe}(S)$ decreases as $S$ expands, so the solution should always select the maximum number of features, i.e. $s_0$ features. Accordingly, we develop the following proofs by comparing the objective function values for different selection sets $S$ with $|S| = s_0$, and find the selection set that minimizes the objective function.
    
    \vspace{0.1in}
    \textbf{Part (1). } Denote by $S^\star_+:= \mathrm{RF}(S) \cap S^\star$ the signal representative features that are captured by $S$, and denote by $S^\star_- := S^\star \setminus S^\star_+$ the signal representative features that are missed by $S$. 

    Note that multiple features within one cluster may be selected by $S$. We create a set of new basis vectors $\{Z_1, \dots, Z_{|S|}\}$ for $\spa(X_S)$ as follows: (1) If $|S \cap \mathcal{C}_k| = 1 $ for some $k$ and $j = S \cap \mathcal{C}_k$, then take $Z_j = X_j$; (2) If $|S \cap \mathcal{C}_k| > 1$ for some $k$ and $S \cap \mathcal{C}_k = \{ k, j_1, \dots j_k \}$, then take $\left[ Z_k, Z_{j_1}, \dots, Z_{j_k} \right] = \left[X_k, \delta_{j_1}, \dots, \delta_{j_k}\right] $; (3) If $|S \cap \mathcal{C}_k| > 1$ for some $k$, and $S \cap \mathcal{C}_k = \{j_1, \dots j_k \}$ where $k \not\in S \cap \mathcal{C}_k$, then take $\left[ Z_{j_1}, Z_{j_2} \dots, Z_{j_k} \right] = \left[ X_{j_1}, \P_{X_{j_1}^\perp} X_{j_2}, \dots, \P_{\{ X_{j_1}, \dots, X_{j_{k-1}} \}^\perp} X_{j_k} \right]$. Now let $\Delta = \P_{X_S} - \sum_{j\in S} \P_{Z_j}$. By Lemma \ref{lem:proj-decomp} and \ref{lem:between-cluster}, we have $\|\Delta\|_F \leq \left( 2 \sum_{j<l \in S} \tr(\P_{Z_j} \P_{Z_l}) \right)^{1/2} \leq  |S| \max_{j<l \in S} \|\P_{Z_j} \P_{Z_l}\|_2  \leq  \mathcal{O} (s_0 D^2\sqrt{\log p / n}  )$.
    
    \vspace{0.1in}
    \underline{Step 1: consider the scenario when $S$ misses no signal groups}, which is $S^\star = S^\star_+$. 
    
    For any $k\in S^\star$, there must be one $j_0\in S$ such that $\mathrm{RF}(j_0) =k$ and $Z_{j_0} = X_{j_0}$. Hence, $\left\| \P_{{X_S}^\perp} X_k \right\|_2^2 = X_k^\top \P_{{X_S}^\perp} X_k = \left\|\P_{X_{j_0}^\perp} X_k \right\|_2^2 - \sum_{j\in S\setminus\{j_0\}} \left\|\P_{Z_j} X_k \right\|_2^2 - X_k^\top \Delta X_k $. By Lemma \ref{lem:cluster-angle}, we know $\left\|\P_{X_{j_0}^\perp} X_k \right\|_2^2 = 1 - \tr( \P_{X_{j_0}} \P_{X_k} ) \leq \frac{\eta_1^2}{1 + \eta_1^2} + \mathcal{O} ( \sqrt{ \log p / n } )$. By Lemma \ref{lem:between-cluster}, $\left\| \P_{Z_j} X_k \right\|_2^2 = \tr(\P_{Z_j} \P_{X_k}) \leq \mathcal{O} (D^2 \log p /n  ) $ for any $j\in S\setminus \mathcal{C}_k$. 
    For any $j\in \mathcal{C}_k \setminus\{j_0\}$, $\left\| \P_{Z_j} X_k \right\|_2 = \left\| \P_{Z_j} \P_{X_k} \right\|_2 \leq \left\| \P_{Z_j} \P_{X_{j_0}} \right\|_2 + \left\| \P_{Z_j} ( \P_{X_k} -  \P_{X_{j_0}} ) \right\|_2 \leq \left\|\P_{X_{j_0}^\perp} X_k \right\|_2 $. Furthermore, $\left|X_k^\top \Delta X_k \right| \leq \|\Delta\|_F \leq  \mathcal{O} (s_0 D^2 \sqrt{ \log p / n }  ).
    $
    As a result, $\left\| \P_{{X_S}^\perp} X_k \right\|_2^2 \leq \frac{D \eta_1^2}{1 + \eta_1^2} + \mathcal{O} (s_0 D^2 \sqrt{ \log p / n }  ) $.

    For any $k\neq l \in S^\star$, we have $\left| X_k^\top \P_{X_S^\perp} X_l \right|_2 \leq \left| X_k^\top X_l \right| + \sum_{j\in S} \left|X_k^\top \P_{Z_j} X_l \right| + \left|X_k^\top \Delta X_l \right| $. Note that for any $j\in S$, either $j \notin \mathcal{C}_k$ or $j \notin \mathcal{C}_l$; without loss of generality, we assume that $j \notin \mathcal{C}_l$. By Lemma \ref{lem:between-cluster}, we have $\left|X_k^\top \P_{Z_j} X_l \right| \leq \left\| \P_{Z_j} X_l \right\|_2  \leq \mathcal{O} (D \sqrt{\log p / n}  ) $. Furthermore, $\left| X_k^\top \Delta X_l \right| \leq \|\Delta\|_F \leq \mathcal{O} (s_0 D^2 \sqrt{\log p / n}  )$. Therefore, $\left| X_k^\top \P_{X_S^\perp} X_l \right|_2 \leq \mathcal{O} (s_0 D^2 \sqrt{\log p / n}  )$.

    For any $k \in S^\star$, we have $\left| \epsilon^\top \P_{X_S^\perp} X_k \right| \leq \left|\epsilon^\top X_k \right| + \sum_{j\in S} \left|\epsilon^\top \P_{Z_j} X_k \right| + \left|\epsilon^\top \Delta X_k \right| $. By Lemma \ref{lem:between-cluster}, for any $j\in S$, we have $\left|\epsilon^\top \P_{Z_j} X_k \right| \leq \|\epsilon\|_2 \cdot \mathcal{O} (D \sqrt{\log p / n}  )$. Furthermore, $\left|\epsilon^\top \Delta X_k \right| \leq \|\epsilon\|_2 \|\Delta\|_F \leq \|\epsilon\|_2 \cdot  \mathcal{O} ( s_0 D^2 \sqrt{\log p / n}  )$. Hence, $\left|  \epsilon^\top \P_{X_S^\perp} X_k\right| \leq \|\epsilon\|_2 \cdot  \mathcal{O} ( s_0 D^2 \sqrt{\log p / n}  )$.

    Finally, note that $\left\| \P_{X_S^\perp} \epsilon \right\|_2^2 \leq \|\epsilon\|_2^2$. In summary, when $S \in \mathcal{S}$,
    $$
    \mathrm{pe}(S) \leq U_1:= \frac{D \eta_1^2}{1 + \eta_1^2} \|\beta^\star\|_2^2 + \|\epsilon\|_2^2 + \mathcal{O}\left( s_0 D^2 \sqrt{\log p / n} \right) \left( \|\beta^\star\|_1 + \|\epsilon\|_2 \right)^2 .
    $$

    \vspace{0.1in}
    \underline{Step 2: consider the scenario when $S$ misses $r$ signal groups}, which is $|S^\star_-| = r \geq 1$.

    Now for any $k \in S^\star_+$, $\left\| \P_{{X_S}^\perp} X_k \right\|_2^2 \geq 0$. For any $k\in S^\star_-$, on the other hand, $\left\| \P_{{X_S}^\perp} X_k \right\|_2^2 = X_k^\top \P_{{X_S}^\perp} X_k = 1 - \sum_{j\in S} \left\|\P_{Z_j} X_k \right\|_2^2 - X_k^\top \Delta X_k .
    $
    By Lemma \ref{lem:between-cluster}, we have
    $
    \left\|\P_{Z_j} X_k\right\|_2^2 = \tr(\P_{Z_j} \P_{X_k}) \leq \mathcal{O}\left( D^2 \log p / n \right)
    $ for any $j\in S$. Furthermore,
    $
    \left|X_k^\top \Delta X_k \right| \leq \|\Delta\|_F \leq   \mathcal{O} ( s_0 D^2 \sqrt{ \log p / n }  ).
    $
    As a result, $\left\| \P_{{X_S}^\perp} X_k \right\|_2^2 \geq 1 - \mathcal{O} (s_0 D^2 \sqrt{ \log p / n }  ) $ for $k\in S^\star_-$.

    Similarly as in step 1, we have $\left| X_k^\top \P_{X_S^\perp} X_l \right|_2 \leq \mathcal{O} (s_0 D^2 \sqrt{\log p / n}  )$, and $\left|  \epsilon^\top \P_{X_S^\perp} X_k\right| \leq \|\epsilon\|_2 \cdot \mathcal{O} ( s_0 D^2 \sqrt{\log p / n}  )$ for any $k \neq l \in S^\star$.

    Finally, $\left\| \P_{X_S^\perp} \epsilon \right\|_2^2 = \epsilon^\top \P_{X_S^\perp} \epsilon = \|\epsilon\|_2^2 - \sum_{j\in S} \|\P_{Z_j} \epsilon\|_2^2 - \epsilon^\top \Delta \epsilon $. By Lemma \ref{lem:between-cluster}, for any $j\in S$, we have $\left\| \P_{Z_j} \epsilon \right\|^2_2 = \|\epsilon\|_2^2 \cdot  \tr(\P_{Z_j} \P_\epsilon) \leq \|\epsilon\|^2_2 \cdot \mathcal{O} ( D^2 \log p / n  )$, and $\left|\epsilon^\top \Delta \epsilon \right|_2 \leq \|\epsilon\|_2^2 \|\Delta\|_F \leq \|\epsilon\|_2^2 \cdot \mathcal{O} ( s_0 D^2 \sqrt{\log p / n}  )$. Therefore, $\left\| \P_{X_S^\perp} \epsilon \right\|_2^2  \geq \|\epsilon\|_2^2 \left[1 - \mathcal{O}( s_0 D^2 \sqrt{\log p / n} ) \right] $. 

    In summary, when $|S^\star_-| = r \geq 1$,
    $$
    \mathrm{pe}(S) \geq L_1(r):= r \min_{j\in S^\star} \left|\beta_j^\star \right|^2 + \|\epsilon\|_2^2 + \mathcal{O}\left( s_0 D^2 \sqrt{\log p / n} \right) \left( \|\beta^\star\|_1 + \|\epsilon\|_2\right)^2.
    $$

    \vspace{0.1in}
    \underline{Step 3: Find a sufficient condition for not missing any signal groups}. Note that when
    $
    \min_{j\in S^\star} \left|\beta^\star_j\right|^2 > 
        \frac{D \eta_1^2}{1 + \eta_1^2}  \|\beta^\star\|_2^2 + \mathcal{O} ( s_0 D^2 \sqrt{\log p / n}  ) \left( \|\beta^\star\|_1 + \|\epsilon\|_2 \right)^2
    $, we have $L_1(r) > U_1$ for any $r \geq 1$, meaning that no signal groups should be missed.

    \vspace{0.1in}
    \noindent\textbf{Part (2). }
    When $s_0 = s^\star$, the only solutions that miss no signal groups are those selecting exactly one feature per signal group. In other words, $\supp(\widehat\beta) \in \mathcal{S}$.

    \vspace{0.1in}
    \noindent\textbf{Part (3). }
    \underline{Step 1: consider the scenario when $S \in \mathcal{S}\setminus \{S^\star \}$}. Let $\Delta = \P_{X_S} - \sum_{j\in S} \P_{X_j}$, and by Lemma \ref{lem:cluster-angle} and \ref{lem:proj-decomp}, we have $\|\Delta\|_F \leq \mathcal{O} (s^\star \sqrt{\log p / n}  )$.

    For any $k\in S^\star$, there must be one and only one $j_0\in S$ such that $\mathrm{RF}(j_0) =k$. Hence, $\left\| \P_{{X_S}^\perp} X_k \right\|_2^2 = X_k^\top \P_{{X_S}^\perp} X_k = \left\|\P_{X_{j_0}^\perp} X_k \right\|_2^2 - \sum_{j\in S\setminus\{j_0\}} \left\|\P_{X_j} X_k \right\|_2^2 - X_k^\top \Delta X_k $. By Lemma \ref{lem:eta1-gap}, we know $\left\|\P_{X_{j_0}^\perp} X_k \right\|_2^2 \geq \frac{\eta_1^2}{1 + \eta_1^2} + \mathcal{O} ( \sqrt{ \log p / n } )$. By Lemma \ref{lem:cluster-angle}, $\left\| \P_{X_j} X_k \right\|_2^2 = \tr(\P_{X_j} \P_{X_k}) \leq \mathcal{O} (\log p  / n ) $ for any $j\in S\setminus \{j_0\}$. Furthermore, $\left|X_k^\top \Delta X_k \right| \leq \|\Delta\|_F \leq  s^\star \mathcal{O} ( \sqrt{ \log p / n }  ).
    $
    As a result, $\left\| \P_{{X_S}^\perp} X_k \right\|_2^2 \geq \frac{\eta_1^2}{1 + \eta_1^2} - s^\star \mathcal{O} (\sqrt{ \log p / n }  ) $.

    For any $k\neq l \in S^\star$, we have $\left| X_k^\top \P_{X_S^\perp} X_l \right|_2 \leq \left| X_k^\top X_l \right| + \sum_{j\in S} \left|X_k^\top \P_{X_j} X_l \right| + \left|X_k^\top \Delta X_l \right| $. Note that for any $j\in S$, either $j \notin \mathcal{C}_k$ or $j \notin \mathcal{C}_l$; without loss of generality, we assume that $j \notin \mathcal{C}_l$. By Lemma \ref{lem:cluster-angle}, we have $\left|X_k^\top \P_{X_j} X_l \right| \leq \left\| \P_{X_j} X_l \right\|_2 = \tr(\P_{X_j} \P_{X_l})^{1/2} \leq \mathcal{O} (\sqrt{\log p / n}  ) $. Furthermore, $\left| X_k^\top \Delta X_l \right| \leq \|\Delta\|_F \leq s^\star \mathcal{O} (\sqrt{\log p / n} )$. Therefore, $\left| X_k^\top \P_{X_S^\perp} X_l \right|_2 \leq s^\star \mathcal{O}(\sqrt{\log p / n} )$.

    For any $k \in S^\star$, we have $\left| \epsilon^\top \P_{X_S^\perp} X_k \right| \leq \left|\epsilon^\top X_k \right| + \sum_{j\in S} \left|\epsilon^\top \P_{X_j} X_k \right| + \left|\epsilon^\top \Delta X_k \right| $. By Lemma \ref{lem:cluster-angle}, for any $j\in S$, we have $\left|\epsilon^\top \P_{X_j} X_k \right| \leq \|\epsilon\|_2 \mathcal{O} (\sqrt{\log p / n}  )$. Furthermore, $\left|\epsilon^\top \Delta X_k \right| \leq \|\epsilon\|_2 \|\Delta\|_F \leq \|\epsilon\|_2 \mathcal{O} ( s^\star \sqrt{\log p / n}  )$. Hence, $\left|  \epsilon^\top \P_{X_S^\perp} X_k\right| \leq \|\epsilon\|_2 \mathcal{O} ( s^\star \sqrt{\log p / n} )$.

    Finally, $\left\| \P_{X_S^\perp} \epsilon \right\|_2^2 = \epsilon^\top \P_{X_S^\perp} \epsilon = \|\epsilon\|_2^2 - \sum_{j\in S} \|\P_{X_j} \epsilon\|_2^2 - \epsilon^\top \Delta \epsilon $. By Lemma \ref{lem:cluster-angle}, for any $j\in S$, we have $\left\| \P_{X_j} \epsilon \right\|^2_2 = \|\epsilon\|_2^2 \cdot \tr(\P_{X_j} \P_\epsilon) \leq \|\epsilon\|^2_2 \mathcal{O} (\log p / n  )$, and $\left|\epsilon^\top \Delta \epsilon \right|_2 \leq \|\epsilon\|_2^2 \|\Delta\|_F \leq \|\epsilon\|_2^2 \mathcal{O} ( s^\star \sqrt{\log p / n}  )$. Therefore, $\left\| \P_{X_S^\perp} \epsilon \right\|_2^2  \geq \|\epsilon\|_2^2 \left[1 - \mathcal{O} ( s^\star \sqrt{\log p / n} ) \right] $. 

    In summary, when $S \in \mathcal{S} \setminus \{S^\star\}$, 
    $$
    \mathrm{pe}(S) \geq L_2:= \min_{j\in S^\star} \left|\beta_j^\star \right|^2 \frac{\eta_1^2}{1 + \eta_1^2} + \|\epsilon\|_2^2 - \mathcal{O}\left( s^\star \sqrt{\log p / n} \right) \left( \|\beta^\star\|_1 + \|\epsilon\|_2\right)^2.
    $$

    \vspace{0.1in}
    \underline{Step 2: Find a sufficient condition for preferring $S^\star$ over $\mathcal{S}\setminus \{S^\star\} $}.
    Note that $\mathrm{pe}(S^\star) = \left\| \P_{ (X_{S^{\star}})^\perp } \epsilon \right\|_2^2$. When 
    $
    \frac{\eta_1^2}{1 + \eta_1^2}  \|\beta^\star\|_2^2 
    > 
    \mathcal{O} ( s^\star \sqrt{\log p / n} ) \left( \|\beta^\star\|_1 + \|\epsilon\|_2 \right)^2,
    $
    we have $L_2 > \mathrm{pe}(S^\star)$, meaning that $S^\star$ is more preferred by $\widehat\beta$.

\end{proof}

\vspace{0.2in}
\subsubsection{Proof of Proposition \ref{prop:assum-inter} and Corollary \ref{cor:pfer}}
\label{sec:sup-proofpropcor}

\begin{proof}
    Let $\nabla = \mathbb{E} \left[ \tr\left( (\P_X - \sum_{j\in[p]} \P_{w_j}) \P_{X_{\widehat{S}^{(\ell)}}} \right) \right]$.
    By Assumption \ref{assum:better}, we have
    $$
    \frac{ \sum_{j\in W} \mathbb{E}\left[\tr(\P_{w_j} \P_{X_{\widehat{S}^{(\ell)}}} ) \right]  }{p - s^\star}
    \leq
    \frac{ \mathbb{E}|\widehat{S}^{(\ell)}| - \sum_{j\in W} \mathbb{E}\left[\tr(\P_{w_j} \P_{X_{\widehat{S}^{(\ell)}}} ) \right] - \nabla }{s^\star},
    $$
    which, combined with Assumption \ref{assum:exchange}, implies that
    $$
    (p-s^\star) l
    \leq \sum_{j\in W} \mathbb{E}\left[ \tr(\P_{w_j} \P_{X_{\widehat{S}^{(\ell)}}})  \right] \leq \frac{ p-s^\star }{p} (\mathbb{E}|\widehat{S}^{(\ell)}| - \nabla),
    $$
    and hence $l \leq \frac{ \mathbb{E}|\widehat{S}^{(\ell)}| - \nabla }{p} \leq \frac{ s_0 - \nabla }{p}$, where $|\nabla| \leq s_0 \|\P_{X} - \sum_{j\in[p]} \P_{w_j}\|_2 \leq \mathcal{O}(s_0 p \sqrt{\log p / n})$. As a result,
    \begin{align*}
    \begin{split}
        \sum_{j\in W} \sup_{T\in\mathscr{T}} \mathbb{E}\left[ \tr(\P_{w_j} \P_{X_{\widehat{S}(T)}}) \right] \leq \sum_{j\in W} (l + r_j) = \sum_{j\in W} \left[ \frac{s_0 - \nabla}{p} + r_j \right] 
        \leq s_0 + \sum_{j\in [p]} r_j + \mathcal{O}\left(s_0 p \sqrt{\log p / n} \right),
    \end{split}
    \end{align*}
    and
    \begin{align*}
    \begin{split}
     &   \sum_{j\in W} \sup_{T\in\mathscr{T}} \mathbb{E}\left[ \tr(\P_{w_j} \P_{X_{\widehat{S}(T)}}) \right]^2 \leq \sum_{j\in W} (l + r_j)^2 = \sum_{j\in W} \left[\frac{ s_0 - \nabla }{p} + r_j \right]^2 \\
    \leq{}& \sum_{j\in W} \left[ \frac{s_0^2}{p^2} + r_j^2 + \frac{2s_0}{p} r_j + \mathcal{O}\left(s_0^2 \sqrt{\log p / n} \right)  \right] 
    \leq \frac{s_0^2}{p} + \sum_{j\in[p]} \left[ r_j^2 + \frac{2s_0}p{ r_j} \right] + \mathcal{O}\left(s_0^2 p \sqrt{\log p / n} \right).
    \end{split}
    \end{align*}
    Plugging these terms into the result of Theorem \ref{thm:1}, the result of Corollary \ref{cor:pfer} follows.
    
\end{proof}

\vspace{0.2in}
\subsubsection{Proof of equation \texorpdfstring{{\eqref{eq:assump_cluster}}}{(15)} }
\label{sec:sup-proofeq}
\begin{proof}
    In the special case, the base procedure $\widehat{S}^{(\ell)}$ would include $S^\star$ and uniformly select $(s_0 - s^\star)$ other directions from $ \bigcup_{k \in S^\star} \{\delta_j: j \in \mathcal{V}_k\}  \cup \bigcup_{k\notin S^\star } \{ X_j : j\in \mathcal{C}_k\}$.

    For any $u =\delta_j \in \bigcup_{k\in S^\star} \{\delta_j: j \in \mathcal{V}_k \}$, we have $\tr(\P_u \P_{X_{\widehat{S}^{(\ell)}}}) = 1$ w.p. $\frac{s_0 - s^\star}{ p - s^\star }$, and $\tr(\P_u \P_{X_{\widehat{S}^{(\ell)}}}) = 0$ otherwise.

    For any $u = \delta_j - \delta_l \in \bigcup_{k\in S^\star} \{\delta_j - \delta_l: j\neq l \in \mathcal{V}_k\} $, if both $j,k \in \widehat{S}^{(\ell)}$ (w.p. $\frac{ \binom{2}{2} \binom{p-s^\star - 2}{s_0 - s^\star - 2} }{ \binom{p-s^\star}{s_0 - s^\star} } = \frac{(s_0 - s^\star) (s_0 - s^\star - 1)}{ (p-s^\star) (p-s^\star - 1) } $), we have $\tr(\P_u \P_{X_{\widehat{S}^{(\ell)}}}) = 1 $. If either $j \in \widehat{S}^{(\ell)}$ or $l \in \widehat{S}^{(\ell)}$ (w.p. $\frac{ \binom{2}{1} \binom{p-s^\star - 2}{s_0 - s^\star - 1} }{ \binom{p-s^\star}{s_0 - s^\star} } = \frac{2 (s_0 - s^\star) (p - s_0)}{ (p-s^\star) (p-s^\star - 1) } $), we have $\tr(\P_u \P_{X_{\widehat{S}^{(\ell)}}}) = \frac{\eta_1^2}{2 (1 + \eta_1^2)} $. If $j,k \notin \widehat{S}^{(\ell)}$, we have $\tr(\P_u \P_{X_{\widehat{S}^{(\ell)}}}) = 0$ otherwise.

    For any $u = X_j \in \bigcup_{k\in \mathcal{K} \cap S^\star} \{X_j: j \in \mathcal{C}_k\} $, if $\mathcal{C}_{\mathrm{RF}(j)} \cap \widehat{S}^{(\ell)} \neq \varnothing$ (w.p. $1 - \frac{ \binom{D}{0} \binom{p-s^\star - D}{s_0 - s^\star } }{ \binom{p-s^\star}{s_0 - s^\star} } =1 - \frac{ (p-s_0) (p-s_0 - 1) \cdots (p - s_0 -D + 1) }{ (p-s^\star) (p-s^\star - 1) \cdots (p - s^\star - D +1) }$), we have $\tr(\P_u \P_{X_{\widehat{S}^{(\ell)}}}) \leq 1 $. Otherwise, $\tr(\P_u \P_{\widehat{S}^{(\ell)}}) = 0$.
\end{proof}

\vspace{0.2in}
\subsubsection{Proof of equation \texorpdfstring{\eqref{eq:prediction-similarity}}{(16)} }
\label{sec:sup-proof-prediction}

\begin{proof}
Note that when $\epsilon \sim \mathcal{N}(0, \sigma^2 / n I_n)$, we have $\max_{j\in \mathcal{K}}\mathrm{Corr}(X_j, \epsilon) \leq \mathcal{O}(\sqrt{\log p / n})$, $\max_{j\in \mathcal{V}}\mathrm{Corr}(\delta_j, \epsilon) \leq \mathcal{O}(\sqrt{\log p / n})$, and $\|\epsilon\|_2 \leq \sigma  +\mathcal{O}(\sigma/\sqrt{n})$ with high probability. These results follow from the concentration inequalities: for any $t > 0$, $\mathbb{P}\left[ \max_{j\in \mathcal{K}}\mathrm{Corr}(X_j, \epsilon) \leq t, \max_{j\in \mathcal{V}}\mathrm{Corr}(\delta_j, \epsilon) \leq t  \right] \geq 1 - ap \exp\{-bnt^2\}$ and $\mathbb{P}[ \left| \|\epsilon\|_2 - \sigma \right| \leq t ] \geq 1 - 2 \exp (-ct^2) $ with some $a,b,c >0 $.

\vspace{0.1in}
Now since $ \mathfrak{S} \subseteq \mathcal{S}$, the goal amounts to finding an upper bound and lower bound of $\mathrm{pe}(S)$ for all $S\in\mathcal{S}$. Let $\Delta = \P_{X_S} - \sum_{j\in S} \P_{X_j}$, and by Lemma \ref{lem:cluster-angle} and \ref{lem:proj-decomp}, we have $\|\Delta\|_F \leq \mathcal{O}  (s^\star \sqrt{\log p / n}  )$.

\vspace{0.1in}
\underline{Step 1: find a lower bound for $\mathrm{pe}(S)$ for all $S\in\mathcal{S}$}. For any $k\in S^\star$, we have $\|\P_{X_S^\perp} X_k\|_2^2 \geq 0$.

For any $k\neq l \in S^\star$, we have $\left| X_k^\top \P_{X_S^\perp} X_l \right|_2 \leq \left| X_k^\top X_l \right| + \sum_{j\in S} \left|X_k^\top \P_{X_j} X_l \right| + \left|X_k^\top \Delta X_l \right| $. Note that for any $j\in S$, either $j \notin \mathcal{C}_k$ or $j \notin \mathcal{C}_l$; without loss of generality, we assume that $j \notin \mathcal{C}_l$. By Lemma \ref{lem:cluster-angle}, we have $\left|X_k^\top \P_{X_j} X_l \right| \leq \left\| \P_{X_j} X_l \right\|_2 = \tr(\P_{X_j} \P_{X_l})^{1/2} \leq \mathcal{O} (\sqrt{\log p / n}  ) $. Furthermore, $\left| X_k^\top \Delta X_l \right| \leq \|\Delta\|_F \leq s^\star \mathcal{O} (\sqrt{\log p / n} )$. Therefore, $\left| X_k^\top \P_{X_S^\perp} X_l \right|_2 \leq s^\star \mathcal{O} (\sqrt{\log p / n} )$.

For any $k \in S^\star$, we have $\left| \epsilon^\top \P_{X_S^\perp} X_k \right| \leq \left|\epsilon^\top X_k \right| + \sum_{j\in S} \left|\epsilon^\top \P_{X_j} X_k \right| + \left|\epsilon^\top \Delta X_k \right| $. By Lemma \ref{lem:cluster-angle}, for any $j\in S$, we have $\left|\epsilon^\top \P_{X_j} X_k \right| \leq \|\epsilon\|_2 \mathcal{O} (\sqrt{\log p / n} )$. Furthermore, $\left|\epsilon^\top \Delta X_k \right| \leq \|\epsilon\|_2 \|\Delta\|_F \leq \|\epsilon\|_2 \mathcal{O} ( s^\star \sqrt{\log p / n}  )$. Hence, $\left|  \epsilon^\top \P_{X_S^\perp} X_k\right| \leq \|\epsilon\|_2 \mathcal{O} ( s^\star \sqrt{\log p / n} )$.

Finally, $\left\| \P_{X_S^\perp} \epsilon \right\|_2^2 = \epsilon^\top \P_{X_S^\perp} \epsilon = \|\epsilon\|_2^2 - \sum_{j\in S} \|\P_{X_j} \epsilon\|_2^2 - \epsilon^\top \Delta \epsilon $. By Lemma \ref{lem:cluster-angle}, for any $j\in S$, we have $\left\| \P_{X_j} \epsilon \right\|^2_2 = \|\epsilon\|_2^2 \tr(\P_{X_j} \P_\epsilon) \leq \|\epsilon\|^2_2 \mathcal{O} (\log p / n )$, and $\left|\epsilon^\top \Delta \epsilon \right|_2 \leq \|\epsilon\|_2^2 \|\Delta\|_F \leq \|\epsilon\|_2^2 \mathcal{O} ( s^\star \sqrt{\log p / n} )$. Therefore, $\left\| \P_{X_S^\perp} \epsilon \right\|_2^2  \geq \|\epsilon\|_2^2 \left[1 - \mathcal{O} ( s^\star \sqrt{\log p / n}  ) \right] $. 

In summary, when $S \in \mathcal{S}$, 
$$
\mathrm{pe}(S) \geq L_3:=  \|\epsilon\|_2^2 - \mathcal{O}\left( s^\star \sqrt{\log p / n} \right) \left( \|\epsilon\|_2 + \|\beta^\star\|_1 \right)^2.
$$

\vspace{0.1in}
\underline{Step 2: find an upper bound for $\mathrm{pe}(S)$ for all $S\in\mathcal{S}$}. 

    For any $k\in S^\star$, there must be one and only one $j_0\in S$ such that $\mathrm{RF}(j_0) =k$. Hence, $\left\| \P_{{X_S}^\perp} X_k \right\|_2^2 = X_k^\top \P_{{X_S}^\perp} X_k = \left\|\P_{X_{j_0}^\perp} X_k \right\|_2^2 - \sum_{j\in S\setminus\{j_0\}} \left\|\P_{X_j} X_k \right\|_2^2 - X_k^\top \Delta X_k $. By Lemma \ref{lem:cluster-angle}, we know $\left\|\P_{X_{j_0}^\perp} X_k \right\|_2^2 \leq \frac{\eta_1^2}{1 + \eta_1^2} + \mathcal{O} ( \sqrt{ \log p / n } )$. Furthermore, $\left\| \P_{X_j} X_k \right\|_2^2 = \tr(\P_{X_j} \P_{X_k}) \leq \mathcal{O} ( \log p /n  ) $ for any $j\in S\setminus \{j_0\}$. Furthermore, $\left|X_k^\top \Delta X_k \right| \leq \|\Delta\|_F \leq  s^\star \mathcal{O} ( \sqrt{ \log p / n } ).
    $
    As a result, $\left\| \P_{{X_S}^\perp} X_k \right\|_2^2 \leq \frac{\eta_1^2}{1 + \eta_1^2} + s^\star \mathcal{O} (\sqrt{ \log p / n }  ) $.

    Similarly as the step 1, we have $\left| X_k^\top \P_{X_S^\perp} X_l \right|_2 \leq s^\star \mathcal{O} (\sqrt{\log p / n}  )$, and $\left|  \epsilon^\top \P_{X_S^\perp} X_k\right| \leq \|\epsilon\|_2 \mathcal{O} ( s^\star \sqrt{\log p / n} )$ for any $k\neq l\in S^\star$. Finally, $\|\P_{X_S^\perp} \epsilon\|_2^2 \leq \|\epsilon\|_2^2$.

    In summary, when $S \in \mathcal{S}$,
    $$
    \mathrm{pe}(S) \leq U_3:= \frac{\eta_1^2}{1 + \eta_1^2} \|\beta^\star\|_2^2 +  \|\epsilon\|_2^2 + \mathcal{O}\left( s^\star \sqrt{\log p / n} \right) \left( \|\beta^\star\|_1 +  \|\epsilon\|_2\right)^2.
    $$

    \underline{Step 3: Draw the conclusion}. The prediction error is upper bounded by $U_3 - L_3$: 
    $$
    \max_{S_1 \neq S_2 \in \mathfrak{S} } \left| \mathrm{pe}(S_1) - \mathrm{pe}(S_2) \right| \leq U_3 - L_3.
    $$

\end{proof}

\vspace{0.2in}
\subsection{The complex dependency setup}
\subsubsection{Preliminaries}
\begin{lemma}
    Suppose that 
    $$
    \max\left\{
    \max_{k\neq k' \in \mathcal{K}\cup \mathcal{I} }\left| X_k^\top     X_{k'} \right|, \  
    \max_{k\in\mathcal{K}\cup\mathcal{I} } \frac{ \left| X_k^\top \delta \right| }{ \|X_k\|_2 \cdot \|\delta\|_2 }
    \right\} \leq \eta_0,
    $$ 
    and $\eta_0 = \mathcal{O}(\sqrt{\log p / n})$. We further assume that $K^2 \eta_0 \rightarrow 0$ as $p,n \rightarrow \infty$.
    Then we have
    \begin{enumerate}[label = \normalfont(\arabic*)]
        \item For any $j\in \mathcal{K}$, we have
        $
        \tr(\P_\delta \P_{X_j+\delta}) \in \left[0, \frac{(\eta_0 + \eta_1)^2}{1  - 2\eta_1 \eta_0} \right].
        $ \\
        Additionally, $\tr(\P_\delta \P_{X_j+\delta}) \leq \mathcal{O}(\sqrt{\log p / n})$.

        \item For any $j\in \mathcal{K}$ and $k\in \mathcal{I}$, we have
        $
        \tr(\P_{X_k} \P_{X_j+\delta}) \in \left[0, \frac{\eta_0^2 (1 + \eta_1)^2}{1 - 2\eta_1 \eta_0} \right].
        $ \\
        Additionally, $\tr(\P_{X_k} \P_{X_j+\delta}) \leq \mathcal{O}(\log p / n )$.
    
        \item For any $S \subseteq \mathcal{K}$, $k\in \mathcal{K}$ but $k\notin S$, define $Z = \sum_{j\in S} X_j + \delta$.  \\
        Then $\tr(\P_{X_k} \P_{Z } ) \in \left[0, \ \frac{\eta_0^2 (|S| + \eta_1)^2}{ |S| + \eta_1^2 - |S|(|S|-1)\eta_0 - 2 |S|\eta_0\eta_1  } \right]$. 
        Additionally, $\tr(\P_{X_k} \P_Z) \leq \mathcal{O}(|S| \log p / n)$.

        \item For any $S \subseteq \mathcal{K}$, define $Z = \sum_{j\in S} X_j + \delta$. Then $\tr(\P_{\epsilon} \P_{Z } ) \in \left[0, \ \frac{\eta_0^2 (|S| + \eta_1)^2  }{ |S| + \eta_1^2 - |S|(|S|-1)\eta_0 - 2 |S|\eta_0\eta_1  } \right]$. \\
        Additionally, $\tr(\P_{\epsilon} \P_Z) \leq \mathcal{O}(|S| \log p / n) $.


        \item For any $S \subseteq \mathcal{K}$ and $\widetilde{S} \subseteq S$, define $Z = \sum_{j\in S} X_j + \delta$. Then 
        $$
        \left( \sum_{j\in \widetilde{S}} \beta_j^\star X_j^\top  \right) \P_{Z} \left( \sum_{j\in \widetilde{S}} \beta_j^\star X_j \right) \leq \frac{ (\sum_{j\in \widetilde{S}} \beta_j^\star)^2  + \|\beta^\star_{\widetilde{S}} \|_1^2 (t^2 \eta_0^2 + 2 t\eta_0) } {|S| + \eta_1^2 - |S|(|S| - 1) \eta_0 - 2 |S| \eta_1 \eta_0 } . 
        $$
        where $t = |S| + |\widetilde{S}| + \eta_1$.  \\
        Additionally, 
        $
        \left( \sum_{j\in \widetilde{S}} \beta_j^\star X_j^\top  \right) \P_{Z} \left( \sum_{j\in \widetilde{S}} \beta_j^\star X_j \right) \leq  \frac{ (\sum_{j\in \widetilde{S}} \beta_j^\star)^2  }{|S| + \eta_1^2} + \mathcal{O}( \sqrt{\log p / n} ) \cdot \|\beta^\star_{\widetilde{S}}\|_1^2 . 
        $

        
    \end{enumerate}
\end{lemma}

\begin{proof}
    \noindent\textbf{Part (1).} Let $\xi = \delta / \|\delta\|_2$, then
    $$
    \tr(\P_{\delta} \P_{X_j + \delta}) 
    = \frac{|\delta^\top (X_j+\delta)|^2}{\|\delta\|^2_2 \cdot \|X_j + \delta\|^2_2} 
    \leq \frac{( |\xi^\top X_{j} | \cdot \|\delta\|_2 + \eta_1^2)^2}{(1 - 2|X^\top_j \xi| \cdot \|\delta\|_2 ) \cdot \eta_1^2}
    \leq \frac{(\eta_0 + \eta_1)^2}{1  - 2\eta_1 \eta_0}.
    $$

    \noindent\textbf{Part (2).} Let $\xi = \delta / \|\delta\|_2$, then
    $$
    \tr(\P_{X_k} \P_{X_j + \delta}) 
    = \frac{|X_k^\top (X_j+\delta)|^2}{\|X_k\|^2_2 \cdot \|X_j + \delta\|^2_2} 
    = \frac{( |X_k^\top X_{j} |  + |X_k^\top \xi| \cdot \|\delta\|_2  )^2}{1 - 2|X^\top_j \xi| \cdot \|\delta\|_2 }
    \leq \frac{(1+ \eta_1)^2 \eta_0^2 }{1  - 2\eta_1 \eta_0}.
    $$

    \noindent\textbf{Part (3).} Let $\xi = \delta / \|\delta\|_2$, then
    \begin{align*}
    \begin{split}
        \tr(\P_{X_k} \P_{Z}) 
    &= \frac{|X_k^\top (\sum_{j\in S} X_j+\delta)|^2}{\|X_k\|^2_2 \cdot \|\sum_{j\in S} X_j+\delta\|^2_2} 
    \leq \frac{( \sum_{j\in S} |X_k^\top X_{j} |  + |X_k^\top \xi| \cdot \|\delta\|_2  )^2}{|S| + \eta_1^2 - 2 \sum_{j<k \in S} |X^\top_j X_k| - 2 \sum_{j\in S} |X_j^\top \xi| \cdot \|\delta\|_2 }\\
    &\leq \frac{\eta_0^2(|S|+ \eta_1)^2}{|S| + \eta_1^2 - |S|(|S| - 1) \eta_0 - 2 |S| \eta_0\eta_1}.
    \end{split}
    \end{align*}

    \noindent\textbf{Part (4).} Let $\xi = \delta / \|\delta\|_2$, and let $\varepsilon = \epsilon / \|\epsilon\|_2$. Then
    \begin{align*}
    \begin{split}
        \tr(\P_{\epsilon} \P_{Z}) 
    &= \frac{|\epsilon^\top (\sum_{j\in S} X_j+\delta)|^2}{\|\epsilon\|^2_2 \cdot \|\sum_{j\in S} X_j+\delta\|^2_2} 
    \leq \frac{( \sum_{j\in S} |\varepsilon^\top X_{j} | \cdot \|\epsilon\|_2  + |\varepsilon^\top \xi| \cdot \|\epsilon\|_2 \cdot \|\delta\|_2  )^2}{\|\epsilon\|_2^2 \cdot (|S| + \eta_1^2 - 2 \sum_{j<k \in S} |X^\top_j X_k| - 2 \sum_{j\in S} |X_j^\top \xi| \cdot \|\delta\|_2) }\\
    &\leq \frac{\eta_0^2(|S|+ \eta_1)^2}{|S| + \eta_1^2 - |S|(|S| - 1) \eta_0 - 2 |S| \eta_0\eta_1}.
    \end{split}
    \end{align*}

    \noindent\textbf{Part (5).} Let $\xi = \delta / \|\delta\|_2$. Then
    $
    \left( \sum_{j\in \widetilde{S}} \beta_j^\star X_j^\top  \right) \P_{X_{Z}} \left( \sum_{j\in \widetilde{S}} \beta_j^\star X_j \right) 
    = \frac{ \left[ (\sum_{j\in S} X_j+\delta)^\top (\sum_{j\in \widetilde{S}} \beta^\star_j X_j ) \right]^2 }{ \|\sum_{j\in S} X_j+\delta\|^2_2 }
    $.
    The numerator is upper bounded by $\left[ (\sum_{j\in S} X_j+\delta)^\top (\sum_{j\in \widetilde{S}} \beta^\star_j X_j ) \right]^2 \\
    = \left[ \sum_{j\in \widetilde{S}} \beta^\star_j  + \sum_{j\in \widetilde{S}}  \beta^\star_j X_j^\top(2 \sum_{k\in \widetilde{S}, k\neq j} X_k +  \sum_{k\in S\setminus\widetilde{S}} X_k + \delta)   \right]^2 \\
    \leq \left[ \sum_{j\in \widetilde{S}} \beta^\star_j \right]^2 +  \left[ \sum_{j\in \widetilde{S}} |\beta_j^\star|  \left( 2\sum_{k\in \widetilde{S}, k\neq j} |X_j^\top X_k| +  \sum_{k\in S\setminus \widetilde{S}} |X_j^\top X_k| + | X_j^\top \xi| \cdot \|\delta\|_2 \right) \right]^2 \\
    + 2\left[ \sum_{j\in \widetilde{S}} |\beta^\star_j| \right] \left[ \sum_{j\in \widetilde{S}} |\beta_j^\star|  \left(2 \sum_{k\in \widetilde{S}, k\neq j} |X_j^\top X_k| + \sum_{k\in S\setminus \widetilde{S}} |X_j^\top X_k| + | X_j^\top \xi| \cdot \|\delta\|_2 \right) \right] \\
    \leq \left[ \sum_{j\in \widetilde{S}} \beta^\star_j \right]^2 
    + \|\beta^\star_{\widetilde{S}} \|_1^2 (t^2 \eta_0^2 + 2t\eta_0),\\$
    where $t = |S| + |\widetilde{S}| + \eta_1$. The denominator is given by
    $ \|\sum_{j\in S} X_j + \delta\|_2^2 \geq 1 + \eta_1^2 - |S| (|S| - 1) \eta_0 - 2 |S|\eta_0\eta_1 $.

\end{proof}

\vspace{0.2in}
\begin{lemma}\label{lem:U-set-property2}
    Given any $S \in \mathcal{S}$ and the set of ``undesired'' directions $\mathcal{U}$ in Theorem \ref{prop:block-consist-subsample}:
    \begin{enumerate}[label = \normalfont(\arabic*)]
        \item Suppose that $N^\star_\mathcal{B} = \varnothing$. For any features $X_j$ where $j\in \mathcal{B}\setminus S \cup N^\star_{\mathcal{I}}$, denote $r_j:= X_j - \P_{X_S} X_j$. Then there must exist $u\in \mathcal{U}$ such that $r_j = u - \P_{X_S} u$ or $-r_j = u - \P_{X_S} u$.
        
        \item Suppose that $N^\star_\mathcal{B} \neq \varnothing$. For any features $X_j$ where $j \in N^\star_{\mathcal{B}} \cup \{c\} \cup N^\star_\mathcal{I} $, denote $r_j :=X_j - \P_{X_S} X_j$. Then there must exist $u \in \mathcal{U}$ such that $r_j = u - \P_{X_S} u$ or $-r_j = u - \P_{X_S} u$.

        \item The $\mathcal{U}$ set is far away from $\spa(X_S)$: $\max_{S\in\mathcal{S}, \ u\in \mathcal{U}} \tr(\P_u \P_{X_S}) \leq  \max\left\{ \frac{\eta_1^2}{1 + \eta_1^2},\ \frac{|S^\star_\mathcal{B}|}{K + \eta_1^2} \right\} + \mathcal{O}(s^\star (K+s^\star) \sqrt{\log p / n})  $.
    \end{enumerate}
\end{lemma}
\begin{proof}
    \textbf{Part (1).}
    Consider the following cases: (1) If $j\in \mathcal{B}$ and $c\notin S$, then we must have $S = \mathcal{K}$ and $j = c$. Thus $r_j = X_c - \P_{X_S} X_c = (\sum_{k\in \mathcal{K}} X_k + \delta) - \P_{X_S} (\sum_{k\in \mathcal{K}} X_k + \delta) = \delta - \P_{X_S} \delta$, meaning that $u = \delta \in \mathcal{U}$. (2) If $j \in \mathcal{B}$ and $c\in S$, then we must have $j\in \mathcal{K}$ and $X_j + \delta \in \spa(X_S)$. Thus $r_j = X_j - \P_{X_S} X_j = (X_j + \delta - \delta) - \P_{X_S} (X_j + \delta - \delta) = -\delta + \P_{X_S }\delta$, meaning that $u = \delta \in \mathcal{U}$. (3) If $j \in N^\star_{\mathcal{I}}$, then it is trivial that $u = X_j \in \mathcal{U}$.

    \vspace{0.1in}
    \noindent\textbf{Part (2).}
    It is trivial that $u = X_j \in \mathcal{U}$ for all $j\in N^\star_\mathcal{B} \cup \{C\} \cup N^\star_\mathcal{I}$.

    \vspace{0.1in}
    \noindent\textbf{Part (3).} Consider the following cases: (1) The scenario when $N^\star_\mathcal{B} = \varnothing$. If $u = \delta$, we have $\tr(\P_\delta \P_{X_S}) \leq \max_{j_0 \in \mathcal{K}} \tr(\P_\delta (\sum_{k\in S^\star \setminus \{j_0\} } \P_{X_k} + \P_{X_{j_0} +\delta} + \Delta)) \leq 
    \frac{\eta_1^2}{1 + \eta_1^2} + s^\star \mathcal{O}(\sqrt{\log p / n}) $, where $\Delta = \P_{ \{X_j\}_{j\in S^\star \setminus\{j_0\} } \cup \{ X_{j_0 + \delta} \} } - \sum_{k\in S^\star \setminus \{j_0\} } \P_{X_k} - \P_{X_{j_0} + \delta } $, and $\|\Delta\|_F \leq s^\star \mathcal{O}(\sqrt{\log p / n})$. On the other hand, if $u = X_j$ with $j \in N^\star_\mathcal{I}$, then we have $\tr(\P_{X_j} \P_{X_S}) \leq \max_{j_0 \in \mathcal{K}} \tr(\P_{X_j} (\sum_{k\in S^\star \setminus \{j_0\} } \P_{X_k} + \P_{X_{j_0} +\delta} + \Delta)) 
    \leq s^\star \mathcal{O}(\sqrt{\log p / n}) $.

    (2) The scenario when $N^\star_\mathcal{B} \neq \varnothing$. If $u = X_j$ where $j\in N^\star_\mathcal{B} \cup N^\star_\mathcal{I}$, we then have $\tr(\P_{X_j} \P_{X_S})  = \tr(\P_{X_j} (\sum_{k\in S^\star} X_k + \Delta) ) \leq \mathcal{O}(s^\star \sqrt{\log p / n}) $, where $\Delta = \P_{X_S} - \sum_{k\in S^\star} X_k$ with $\|\Delta\|_F \leq \mathcal{O}(s^\star \sqrt{\log p / n} )$. On the other hand, if $u = X_c$, then we have 
    $
    \tr(\P_{X_c} \P_{X_S}) = \frac{(\sum_{k\in\mathcal{K}} X_k + \delta)^\top \P_{X_S} (\sum_{k\in\mathcal{K}} X_k + \delta)}{ \| \sum_{k\in\mathcal{K}} X_k + \delta \|_2^2 } = \frac{|S^\star_\mathcal{B} |}{K + \eta_1^2} + \mathcal{O}(K^2 \sqrt{\log p / n})
    $.

    Therefore, 
    $\max_{S\in\mathcal{S}, \ u\in \mathcal{U}} \tr(\P_u \P_{X_S}) 
    \leq 
    \max\left\{ \frac{\eta_1^2}{1 + \eta_1^2},\ \frac{|S^\star_\mathcal{B}|}{K + \eta_1^2} \right\} + \mathcal{O}(K (K+s^\star ) \sqrt{\log p / n})$.
    
\end{proof}

\vspace{0.2in}
\subsubsection{Proof of Theorem \ref{prop:block-consist-s0}}
\label{sec:sup-proof-blockbase}
\begin{proof}
    The prediction error for $S$ is given by
    \begin{align*}
    \begin{split}
        \mathrm{pe}(S) 
    &= \left\| \P_{{X_S}^\perp} y \right\|_2^2 
    = \left\| \sum_{k\in\mathcal{K} \cup \mathcal{I}} \beta_k^\star  \P_{{X_S}^\perp} X_k +  \P_{{X_S}^\perp} \epsilon \right\|_2^2 \\
    &= \sum_{k\in S^\star\cap S^c} \left| \beta^\star_k \right|^2 \left\| \P_{X_S^\perp}X_k \right\|_2^2 + \left\|\P_{X_S^\perp} \epsilon \right\|_2^2 + 2 \sum_{k\neq l \in S^\star\cap S^c} \beta_k^\star \beta_l^\star X_k^\top \P_{X_S^\perp} X_l + 2 \sum_{k\in S^\star\cap S^c} \beta^\star_k \epsilon^\top \P_{X_S^\perp} X_k.
    \end{split}
    \end{align*}

    \noindent\textbf{Part (1a): the scenario when $c\not\in S$.} Let $\Delta = \P_{X_S} - \sum_{j\in S} \P_{X_j}$. By Lemma \ref{lem:proj-decomp}, we have $\|\Delta\|_F \leq \mathcal{O}(s_0 \sqrt{\log p / n} )$.

    \vspace{0.1in}
    \underline{Step 1: consider the scenario when $S$ misses $r$ signal features}. In other words, $|S^\star \cap S^c| = r$. For any $k \in S^\star \cap S^c$, $\|\P_{X_S^\perp} X_k\|_2^2 = X_k^\top \P_{X_S^\perp} X_k = 1 - \sum_{j\in S} \tr( \P_{X_j} \P_{X_k}) - X_k^\top \Delta X_k $, where $\tr(\P_{X_j} \P_{X_k}) \leq \mathcal{O}(\log p / n)$, and $|X_k^\top \Delta X_k| \leq \|\Delta\|_F \leq \mathcal{O}(s_0 \sqrt{\log p / n})$. As a result, $\|\P_{X_S^\perp} X_k\|_2^2 \geq 1 - \mathcal{O}(s_0 \sqrt{\log p / n})$.

    For any $k\neq l \in S^\star\cap S^c$, we have $\left| X_k^\top \P_{X_S^\perp} X_l \right| \leq \left| X_k^\top X_l \right| + \sum_{j\in S} \left|X_k^\top \P_{X_j} X_l \right| + \left|X_k^\top \Delta X_l \right| $, where $|X_k^\top X_l| \leq \mathcal{O}(\sqrt{\log p / n})$, $|X_k^\top \P_{X_j} X_l| \leq \mathcal{O}(\sqrt{\log p / n})$, and $|X_k^\top \Delta X_l| \leq \|\Delta\|_F \leq \mathcal{O}(s_0\sqrt{\log p / n})$. As a result, $\left| X_k^\top \P_{X_S^\perp} X_l \right| \leq \mathcal{O}(s_0 \sqrt{\log p / n})$. 

    For any $k \in S^\star$, we have $\left| \epsilon^\top \P_{X_S^\perp} X_k \right| \leq \left|\epsilon^\top X_k \right| + \sum_{j\in S} \left|\epsilon^\top \P_{X_j} X_k \right| + \left|\epsilon^\top \Delta X_k \right| $, where $|\epsilon^\top X_k| \leq \|\epsilon\|_2 \mathcal{O}(\sqrt{\log p / n})$, $|\epsilon^\top \P_{X_j} X_k| \leq \|\epsilon\|_2 \mathcal{O}(\sqrt{\log p / n})$, and $|\epsilon^\top \Delta X_k| \leq \|\epsilon\|_2 \|\Delta\|_F \leq \|\epsilon\|_2 \mathcal{O}(s_0\sqrt{\log p / n})$. As a result, $|\epsilon^\top \P_{X_S^\perp} X_k| \leq \|\epsilon\|_2 \mathcal{O}(s_0 \sqrt{\log p / n})$.

    Finally, $\| \P_{X_S^\perp} \epsilon \|_2^2 = \epsilon^\top \P_{X_S^\perp} \epsilon = \|\epsilon\|_2^2 - \sum_{j\in S} \|\P_{X_j} \epsilon\|_2^2 - \epsilon^\top \Delta \epsilon $, where $\|\P_{X_j} \epsilon\|_2^2 \leq \|\epsilon\|_2^2 \mathcal{O}(\log p / n)$, and $\epsilon^\top \Delta \epsilon \leq \|\epsilon\|_2^2 \|\Delta\|_F \leq \|\epsilon\|_2^2 \mathcal{O}(s_0\sqrt{\log p / n})$. As a result, $\| \P_{X_S^\perp} \epsilon \|_2^2 \geq \|\epsilon\|_2^2 \left[1 - \mathcal{O}(s_0 \sqrt{\log p / n}) \right] $.
    
    In summary, when $|S^\star \cap S^c| = r$, 
    $$
    \mathrm{pe}(S) \geq L_1(r):= r\min_{j\in S^\star} \left|\beta_j^\star \right|^2  + \|\epsilon\|_2^2 - \mathcal{O}\left( s_0 \sqrt{\log p / n} \right) \left( \|\beta^\star\|_1 + \|\epsilon\|_2\right)^2.
    $$

    \vspace{0.1in}
    \underline{Step 2: consider the scenario when $S$ misses no signals, and draw the conclusion}. Note that when $c\notin S$, missing no signals implies $S = S^\star$.  Note that $\mathrm{pe}(S^\star) = \|\P_{X_S^\perp} \epsilon\|_2^2 \leq \|\epsilon\|_2^2$. Therefore, under the beta-min condition \ref{eq:complex-beta-min1}, we have $ \mathrm{pe}(S^\star) \leq \|\epsilon\|_2^2 < \min_{r \geq 1} L_1(r) \leq \mathrm{pe}(S)$ for any $S$ such that $|S^\star \cap S^c| = r \geq 1$. Hence $S^\star$ is more preferred by $\widehat{\beta}$.

    \vspace{0.1in}
    \noindent\textbf{Part (1b): the scenario when $c\in S$ and $N_\mathcal{B}^\star = \varnothing$.}
    
    \underline{Step 1: consider the scenario when $S$ misses signal features in $\mathcal{B}$}. Note that when $|S_\mathcal{B}| = |S^\star_\mathcal{B}|$, all block signals are regarded as being ``selected'', since their subspaces are similar up to $\delta$. Suppose that $r_1 := |S^\star_\mathcal{B} \cap S^c| \geq 2$ signals in $\mathcal{B}$ and $r_2 := |S^\star_\mathcal{I} \cap S^c| \geq 0$ signals in $\mathcal{I}$ are missed. Then the prediction error is given by:
    \begin{align}\label{eq:1b-pred-erro}
    \begin{split}
        \mathrm{pe}(S) 
    &= \left\| \sum_{k\in S^\star \cap S^c } \beta^\star_k \P_{X_S^\perp} X_k + \P_{X_S^\perp} \epsilon  \right\|_2^2 
    = \left\| \sum_{k\in S^\star_\mathcal{B} \cap S^c} \P_{X_S^\perp} \beta_k^\star X_k \right\|_2^2
    + \left\| \sum_{k\in S^\star_\mathcal{I} \cap S^c} \P_{X_S^\perp} \beta_k^\star X_k   \right\|_2^2 + \left\| \P_{X_S^\perp} \epsilon \right\|_2^2  \\
    & \quad + 2\left( \sum_{k\in S^\star_\mathcal{B} \cap S^c} \beta^\star_k X_k^\top \right) \P_{X_S^\perp}  \left( \sum_{j\in S^\star_\mathcal{I} \cap S^c} \beta^\star_j X_j  \right) 
    + 2\epsilon^\top \P_{X_S^\perp} \sum_{k\in S^\star \cap S^c} \beta^\star_k X_k.
    \end{split}
    \end{align}

    Let $\Delta = \P_{X_S} - \sum_{j\in S \setminus\{c\} } \P_{X_j} + \P_{\sum_{j\in S^\star_\mathcal{B} \cap S^c} X_j + \delta}$, where $\|\Delta\|_F \leq (s_0+K) \mathcal{O}(\sqrt{\log p / n}) $.
    
    For the first term in (\ref{eq:1b-pred-erro}), $ \left\| \sum_{k\in S^\star_\mathcal{B} \cap S^c} \P_{X_S^\perp} \beta_k^\star X_k \right\|_2^2 = \left\| \sum_{k\in S^\star_\mathcal{B} \cap S^c} \beta_k^\star X_k \right\|_2^2 - \left\| \sum_{k\in S^\star_\mathcal{B} \cap S^c} \beta_k^\star \P_{X_S}  X_k \right\|_2^2$. Specifically, 
    $
    \left\| \sum_{k\in S^\star_\mathcal{B} \cap S^c} \beta_k^\star X_k \right\|_2^2 
    \geq \sum_{k\in S^\star_\mathcal{B} \cap S^c} |\beta^\star_k|^2 - 2 \sum_{j\neq k \in S^\star_\mathcal{B} \cap S^c } |\beta_j^\star\beta_k^\star| |X_j X_k|  \\
    \geq \sum_{k\in S^\star_\mathcal{B} \cap S^c}  |\beta_k^\star|^2 - \sum_{j\neq k \in S^\star_\mathcal{B} \cap S^c } |\beta_j^\star\beta_k^\star| \mathcal{O}(\sqrt{\log p / n}).
    $ 
    
    \noindent Moreover, 
    $
        \left\| \sum_{k\in S^\star_\mathcal{B} \cap S^c} \beta_k^\star \P_{X_S}  X_k \right\|_2^2 
    =  \left( \sum_{k\in S^\star_\mathcal{B} \cap S^c} \beta_k^\star  X_k \right)^\top \P_{X_S} \left( \sum_{k\in S^\star_\mathcal{B} \cap S^c} \beta_k^\star  X_k \right)  \\
    = \left( \sum_{k\in S^\star_\mathcal{B} \cap S^c} \beta_k^\star  X_k \right)^\top \left( \sum_{j\in S\setminus\{c\} } \P_{X_j} + \P_{\sum_{j\in S^\star_\mathcal{B} \cap S^c} X_j + \delta} + \Delta \right) \left( \sum_{k\in S^\star_\mathcal{B} \cap S^c} \beta_k^\star  X_k \right)  \\
    \leq \|\beta^\star_{S^\star_\mathcal{B} \cap S^c} \|_1^2 \cdot \mathcal{O} ( (s_0+K) \sqrt{\log p / n} ) + \frac{ \left( \sum_{k\in S^\star_\mathcal{B} \cap S^c } \beta_j^\star \right)^2 }{r_1 + \eta_1^2} .
    $
    
    \noindent As a result, $\left\| \sum_{k\in S^\star_\mathcal{B} \cap S^c} \P_{X_S^\perp} \beta_k^\star X_k \right\|_2^2  \\
    \geq \sum_{k\in S^\star_\mathcal{B} \cap S^c}  |\beta_k^\star|^2 - \sum_{j\neq k \in S^\star_\mathcal{B} \cap S^c } |\beta_j^\star\beta_k^\star| \mathcal{O}(\sqrt{\log p / n}) - \frac{ \left( \sum_{k\in S^\star_\mathcal{B} \cap S^c } \beta_j^\star \right)^2 }{r_1 + \eta_1^2} - \|\beta^\star_{S^\star_\mathcal{B} \cap S^c} \|_1^2 \cdot  \mathcal{O}((s_0 +K) \sqrt{\log p / n}) \\
    \geq \frac{ 1 }{ r_1 + \eta_1^2} \left(\sum_{j\neq k \in S^\star_\mathcal{B} \cap S^c} |\beta^\star_j - \beta^\star_k|^2
    + \eta_1^2 \|\beta^\star_{S^\star_\mathcal{B} \cap S^c}\|_2^2  \right)
    - \left( \sum_{j\neq k \in S^\star_\mathcal{B} \cap S^c } |\beta_j^\star\beta_k^\star| + \|\beta^\star_{S^\star_\mathcal{B} \cap S^c} \|_1^2  \right) \mathcal{O}((s_0 + K) \sqrt{\log p / n}). 
    $
    
    \vspace{0.1in}
    For the second term in (\ref{eq:1b-pred-erro}), $\left\| \sum_{k\in S^\star_\mathcal{I} \cap S^c} \P_{X_S^\perp} \beta_k^\star X_k  \right\|_2^2 = \left\| \sum_{k\in S^\star_\mathcal{I} \cap S^c} \beta^\star_k X_k \right\|_2^2 - \left\| \sum_{k\in S^\star_\mathcal{I} \cap S^c} \P_{X_S} \beta_k^\star X_k  \right\|_2^2 $.
    Specifically, $\left\| \sum_{k\in S^\star_\mathcal{I} \cap S^c} \beta^\star_k X_k \right\|_2^2 \geq \sum_{k\in S^\star_\mathcal{I} \cap S^c}  |\beta_k^\star|^2 - \sum_{j\neq k \in S^\star_\mathcal{I} \cap S^c } |\beta_j^\star\beta_k^\star| \mathcal{O}(\sqrt{\log p / n})$.

    \noindent Moreover, $\left\| \sum_{k\in S^\star_\mathcal{I} \cap S^c} \P_{X_S} \beta_k^\star X_k  \right\|_2^2 \\
    = \left( \sum_{k\in S^\star_\mathcal{I} \cap S^c} \beta_k^\star  X_k \right)^\top \left( \sum_{j\in S\setminus\{c\} } \P_{X_j} + \P_{\sum_{j\in S^\star_\mathcal{B} \cap S^c} X_j + \delta} + \Delta \right) \left( \sum_{k\in S^\star_\mathcal{I} \cap S^c} \beta_k^\star  X_k \right)  \\
    \leq \|\beta^\star_{S^\star_\mathcal{I} \cap S^c} \|_1^2 \cdot \mathcal{O} ( (s_0+K) \sqrt{\log p / n}  ) $.

    \noindent As a result, $\left\| \sum_{k\in S^\star_\mathcal{I} \cap S^c} \P_{X_S^\perp} \beta_k^\star X_k  \right\|_2^2 \\
    \geq \sum_{k\in S^\star_\mathcal{I} \cap S^c}  |\beta_k^\star|^2 - \left( \sum_{j\neq k \in S^\star_\mathcal{I} \cap S^c } |\beta_j^\star\beta_k^\star| + \| \beta^\star_{S^\star_\mathcal{I} \cap S^c} \|_1^2 \right) \mathcal{O}( (s_0 + K) \sqrt{\log p / n})$.

    \vspace{0.1in}
    For the third term in (\ref{eq:1b-pred-erro}), $\left\| \P_{X_S^\perp} \epsilon \right\|_2^2
    = \|\epsilon\|_2^2 - \left\| \P_{X_S} \epsilon \right\|_2^2
    $.

    \noindent Specifically,  $\left\| \P_{X_S} \epsilon \right\|_2^2 = \epsilon^\top \left( \sum_{j\in S\setminus\{c\} } \P_{X_j} + \P_{\sum_{j\in S^\star_\mathcal{B} \cap S^c} X_j + \delta} + \Delta \right) \epsilon \leq \|\epsilon\|_2^2 \cdot \mathcal{O} ( (s_0+K) \sqrt{\log p / n} )$.

    \noindent As a result,  $\left\| \P_{X_S^\perp} \epsilon \right\|_2^2 \geq \|\epsilon\|_2^2 - \|\epsilon\|_2^2 \cdot \mathcal{O} ( (s_0+K) \sqrt{\log p / n} ) $.

    \vspace{0.1in}
    For the fourth term in (\ref{eq:1b-pred-erro}), $\left( \sum_{k\in S^\star_\mathcal{B} \cap S^c} \beta^\star_k X_k^\top \right) \P_{X_S^\perp}  \left( \sum_{j\in S^\star_\mathcal{I} \cap S^c} \beta^\star_j X_j  \right) \\
    = \left( \sum_{k\in S^\star_\mathcal{B} \cap S^c} \beta^\star_k X_k^\top \right)  \left( \sum_{j\in S^\star_\mathcal{I} \cap S^c} \beta^\star_j X_j  \right)  - \left( \sum_{k\in S^\star_\mathcal{B} \cap S^c} \beta^\star_k X_k^\top \right) \P_{X_S}  \left( \sum_{j\in S^\star_\mathcal{I} \cap S^c} \beta^\star_j X_j  \right) 
    $. 

    \noindent Specifically, $\left( \sum_{k\in S^\star_\mathcal{B} \cap S^c} \beta^\star_k X_k^\top \right)  \left( \sum_{j\in S^\star_\mathcal{I} \cap S^c} \beta^\star_j X_j  \right) 
    \leq \sum_{k\in S^\star_\mathcal{B} \cap S^c} \sum_{j\in S^\star_\mathcal{I} \cap S^c} |\beta^\star_k \beta^\star_j| \mathcal{O}(\sqrt{\log p / n}).
    $

    \noindent Moreover, $\left( \sum_{k\in S^\star_\mathcal{B} \cap S^c} \beta^\star_k X_k^\top \right) \P_{X_S}  \left( \sum_{j\in S^\star_\mathcal{I} \cap S^c} \beta^\star_j X_j  \right)  \\
    = \left( \sum_{k\in S^\star_\mathcal{B} \cap S^c} \beta^\star_k X_k^\top \right) \left( \sum_{j\in S\setminus\{c\} } \P_{X_j} + \P_{\sum_{j\in S^\star_\mathcal{B} \cap S^c} X_j + \delta} + \Delta \right)  \left( \sum_{j\in S^\star_\mathcal{I} \cap S^c} \beta^\star_j X_j  \right)  \\
    \leq \sum_{k\in S^\star_\mathcal{B} \cap S^c} \sum_{j\in S^\star_\mathcal{I} \cap S^c} |\beta^\star_k \beta^\star_j| \mathcal{O}  ( (s_0 + K) \sqrt{\log p / n} )
    $.

    \noindent As a result, $\left( \sum_{k\in S^\star_\mathcal{B} \cap S^c} \beta^\star_k X_k^\top \right) \P_{X_S^\perp}  \left( \sum_{j\in S^\star_\mathcal{I} \cap S^c} \beta^\star_j X_j  \right) \leq \sum_{k\in S^\star_\mathcal{B} \cap S^c} \sum_{j\in S^\star_\mathcal{I} \cap S^c} |\beta^\star_k \beta^\star_j| \mathcal{O}  ( (s_0 + K) \sqrt{\log p / n} ) $.

    \vspace{0.1in}
    For the fifth term in (\ref{eq:1b-pred-erro}), $ \epsilon^\top \P_{X_S^\perp} \sum_{k\in S^\star \cap S^c} \beta^\star_k X_k  \\
    = \sum_{k\in S^\star \cap S^c} \beta_k^\star \epsilon^\top X_k - 
    \epsilon^\top \left( \sum_{j\in S\setminus\{c\} } \P_{X_j} + \P_{\sum_{j\in S^\star_\mathcal{B} \cap S^c} X_j + \delta} + \Delta \right) \left( \sum_{k\in S^\star \cap S^c} \beta^\star_k X_k \right)  \\
    \leq \|\epsilon\|_2 \cdot \|\beta^\star_{S^\star \cap S^c} \|_1 \mathcal{O} ( (s_0 + K) \sqrt{\log p / n} ). 
    $
    
    \vspace{0.1in}
    Therefore, 
    \begin{align*}
    \begin{split}
        \mathrm{pe}(S) 
    &\geq L_2(r_1, r_2)
    :=  \frac{ \binom{r_1}{2} \min_{j\neq k \in S^\star_\mathcal{B} } |\beta^\star_j - \beta^\star_k|^2 + r_1 \eta_1^2 \min_{j\in S^\star_\mathcal{B} } |\beta^\star_j|^2 }{r_1 + \eta_1^2}  
    + r_2 \min_{j\in S^\star_\mathcal{I}  } |\beta^\star_j |^2 + \|\epsilon\|_2^2 \\
    &\quad + \mathcal{O}\left( (s_0 +K) \sqrt{\log p / n} \right) \left( \|\beta^\star\|_1 + \|\epsilon\|_2 \right)^2.
    \end{split}
    \end{align*}

    \vspace{0.1in}
    \underline{Step 2: consider the scenario when $S$ misses signals in $\mathcal{I}$ but not $\mathcal{B}$}, which means that $|S_\mathcal{B}| = |S^\star_\mathcal{B}|$, but $r_2 = |S^\star_\mathcal{I} \cap S^c| \geq 1$. In this case, there exists at most one element in $S^\star_\mathcal{B} \cap S^c$ and let $\{j_0\} = S^\star_\mathcal{B} \cap S^c$. Then the prediction error is given by:
    \begin{align*}
    \begin{split}
        \mathrm{pe}(S) 
    &=  \left| \beta^\star_{j_0} \right|^2 \left\| \P_{X_S^\perp}X_{j_0} \right\|_2^2 
    + \left\| \sum_{k\in S^\star_\mathcal{I} \cap S^c} \P_{X_S^\perp} \beta_k^\star X_k   \right\|_2^2 + \left\| \P_{X_S^\perp} \epsilon \right\|_2^2  
    + \beta^\star_{j_0} X_{j_0}^\top \P_{X_S^\perp} \left( \sum_{j\in S^\star_\mathcal{I} \cap S^c} \beta^\star_j X_j \right) \\
    &\quad + 2 \sum_{k\in \{j_0\} \cup (S^\star_\mathcal{I} \cap S^c) } \beta^\star_{k} \epsilon^\top \P_{X_S^\perp} X_{k}.
    \end{split}
    \end{align*}

    Let $\Delta = \P_{X_S} - \sum_{j\in S} \P_{X_k}$ where $\|\Delta\|_F \leq \mathcal{O}( s_0 \sqrt{\log p / n})$.
    Note that $\left\| \P_{X_S^\perp}X_{j_0} \right\|_2^2  \geq 0$. Similarly as in Step 1, we have 
    $
    \left\| \sum_{k\in S^\star_\mathcal{I} \cap S^c} \P_{X_S^\perp} \beta_k^\star X_k  \right\|_2^2 
    \geq \sum_{k\in S^\star_\mathcal{I} \cap S^c}  |\beta_k^\star|^2 - \| \beta^\star_{S^\star_\mathcal{I} \cap S^c} \|_1^2 \mathcal{O}(s_0 \sqrt{\log p / n})$, and
    $
    \left\| \P_{X_S^\perp} \epsilon \right\|_2^2 \geq \|\epsilon\|_2^2 - \|\epsilon\|_2^2 \cdot \mathcal{O} ( s_0 \sqrt{\log p / n} ) $. 
    Moreover, we have
    $
    \beta^\star_{j_0} X_{j_0}^\top  \P_{X_S^\perp}  \left( \sum_{j\in S^\star_\mathcal{I} \cap S^c} \beta^\star_j X_j  \right) \leq \sum_{j\in S^\star_\mathcal{I} \cap S^c} |\beta^\star_{j_0} | |\beta^\star_j| \mathcal{O} ( s_0 \sqrt{\log p / n} ) $, and
    $
    \sum_{k\in \{j_0\} \cup (S^\star_\mathcal{I} \cap S^c) } \beta^\star_{k} \epsilon^\top \P_{X_S^\perp} X_{k}  
    \leq \|\epsilon\|_2 \cdot \|\beta^\star_{S^\star \cap S^c} \|_1 \mathcal{O} ( s_0 \sqrt{\log p / n} ). 
    $
    Therefore,
    $$
    \mathrm{pe}(S) \geq \widetilde{L}_2(r_2):= r_2 \min_{k\in S^\star_\mathcal{I} }  |\beta_k^\star|^2 + \|\epsilon\|_2^2 - \mathcal{O}\left( s_0 \sqrt{\log p / n} \right) (\|\beta^\star\|_1 + \|\epsilon\|_2)^2.
    $$

    \vspace{0.1in}
    \underline{Step 3: consider the scenario when $S$ misses no signals}, which means that $|S_\mathcal{B}| = |S^\star_\mathcal{B}|$ and $S_\mathcal{I} = S^\star_\mathcal{I}$. In this case, there exists at most one element in $S^\star_\mathcal{B} \cap S^c$ and let $\{j_0\} = S^\star_\mathcal{B} \cap S^c$. Then the prediction error is given by:
    \begin{align*}
    \begin{split}
        \mathrm{pe}(S) 
    &=  \left| \beta^\star_{j_0} \right|^2 \left\| \P_{X_S^\perp}X_{j_0} \right\|_2^2 + \left\|\P_{X_S^\perp} \epsilon \right\|_2^2 + 2 \beta^\star_{j_0} \epsilon^\top \P_{X_S^\perp} X_{j_0}.
    \end{split}
    \end{align*}

    Note that $\left\|\P_{X_S^\perp} X_{j_0}\right\|_2^2 = 1 - \left\| \P_{X_S} X_{j_0} \right\|_2^2 \leq 1 - \left\| \P_{X_{j_0} + \delta} X_{j_0}\right\|_2^2 \leq \frac{\eta_1^2}{1 + \eta_1^2} + \mathcal{O}(\sqrt{\log p / n})$. In addition, $\|\P_{X_S^\perp} \epsilon\|_2^2 \leq \|\epsilon\|_2^2$. Moreover, $\left| \epsilon^\top \P_{X_S^\perp }X_{j_0} \right| \leq \left|\epsilon^\top X_{j_0} \right| + \left| \epsilon^\top (\sum_{j\in S} \P_{X_j} + \Delta) X_{j_0} \right| \leq \|\epsilon\|_2 \mathcal{O}(s_0 \sqrt{\log p / n}) $, where $\Delta = \P_{X_S} - \sum_{j\in S} \P_{X_k}$ with $\|\Delta\|_F \leq \mathcal{O}( s_0 \sqrt{\log p / n})$.

    Therefore,
    $$
    \mathrm{pe}(S) \leq U_2 := \|\beta^\star \|_2^2 \frac{\eta_1^2}{1 + \eta_1^2} + \|\epsilon\|_2^2 + \mathcal{O}\left(s_0 \sqrt{\log p / n}  \right) \left( \|\beta^\star\|^2_2 + 
    \|\beta^\star\|_1 \|\epsilon\|_2  \right).
    $$

    \vspace{0.1in}
    \underline{Step 4: draw the conclusion}. When 
    \begin{align*}
    \begin{split}
     &   \min\left\{ \frac{\min_{j\neq k \in S^\star_\mathcal{B}} |\beta^\star_j - \beta^\star_k|^2 + 2  \eta_1^2 \min_{j\in S^\star_\mathcal{B} } |\beta^\star_j|^2}{2 + \eta_1^2},\ 
    \min_{j\in S^\star_\mathcal{I}} |\beta_j^\star|^2 \right\} \\
    > {}& \frac{\eta_1^2}{1 + \eta_1^2} \|\beta^\star\|_2^2 + \mathcal{O}\left( (s_0 + K) \sqrt{\log p / n} \right) (\|\beta^\star\|_1 + \|\epsilon\|_2)^2, 
    \end{split}
    \end{align*}
    we have $U_2 < L_2 (r_1, r_2)$ and $U_2 < \widetilde{L}_2(\widetilde{r}_2)$ for any $r_1 \geq 2$, $r_2 \geq 0$ and $ \widetilde{r}_2\geq 1$, meaning that no signal features should be missed.

    \vspace{0.1in}
    \noindent\textbf{Part (1c): the scenario when $c\in S$ and $N_\mathcal{B}^\star \neq \varnothing$.} 
    
    \underline{Step 1: consider the scenario when $S$ misses signal features in $\mathcal{B}$}. Note that when $|S_\mathcal{B}| = K$, all block signals are regarded as being ``selected''. Suppose that $r_1 := |S^\star_\mathcal{B} \cap S^c| \geq 1$ signals in $\mathcal{B}$ and $r_2 := |S^\star_\mathcal{I} \cap S^c| \geq 0$ signals in $\mathcal{I}$ are missed. Moreover, denote $r_3:= |N^\star_\mathcal{B} \cap S^c| \geq 0$. Additionally note that when $r_3 = 0$, we must have $r_1 \geq 2$. Then the prediction error is the same as (\ref{eq:1b-pred-erro}):

    \begin{align*}
    \begin{split}
        \mathrm{pe}(S) 
    &= \left\| \sum_{k\in S^\star \cap S^c } \beta^\star_k \P_{X_S^\perp} X_k + \P_{X_S^\perp} \epsilon  \right\|_2^2 
    = \left\| \sum_{k\in S^\star_\mathcal{B} \cap S^c} \P_{X_S^\perp} \beta_k^\star X_k \right\|_2^2
    + \left\| \sum_{k\in S^\star_\mathcal{I} \cap S^c} \P_{X_S^\perp} \beta_k^\star X_k   \right\|_2^2 + \left\| \P_{X_S^\perp} \epsilon \right\|_2^2  \\
    & \quad + \left( \sum_{k\in S^\star_\mathcal{B} \cap S^c} \beta^\star_k X_k^\top \right) \P_{X_S^\perp}  \left( \sum_{j\in S^\star_\mathcal{I} \cap S^c} \beta^\star_j X_j  \right) 
    + \epsilon^\top \P_{X_S^\perp} \sum_{k\in S^\star \cap S^c} \beta^\star_k X_k.
    \end{split}
    \end{align*}

    Let $\Delta = \P_{X_S} - \sum_{j\in S \setminus\{c\} } \P_{X_j} + \P_{\sum_{j\in \mathcal{B} \cap S^c} X_j + \delta}$, where $\|\Delta\|_F \leq (s_0+K) \mathcal{O}(\sqrt{\log p / n}) $.
    
    Similarly as the Step 1 in Part (1b), for the first term in (\ref{eq:1b-pred-erro}), 
    $\\
    \left\| \sum_{k\in S^\star_\mathcal{B} \cap S^c} \P_{X_S^\perp} \beta_k^\star X_k \right\|_2^2  
    \geq \frac{ 1 }{ r_1 + r_3 + \eta_1^2} \left(\sum_{j\neq k \in S^\star_\mathcal{B} \cap S^c} |\beta^\star_j - \beta^\star_k|^2
    + (\eta_1^2 + r_3) \|\beta^\star_{S^\star_\mathcal{B} \cap S^c}\|_2^2  \right) \\
    - \left( \sum_{j\neq k \in S^\star_\mathcal{B} \cap S^c } |\beta_j^\star\beta_k^\star| + \|\beta^\star_{S^\star_\mathcal{B} \cap S^c} \|_1^2  \right) \mathcal{O}((s_0 + K) \sqrt{\log p / n}). 
    $
    
    \vspace{0.1in}
    For the second term in (\ref{eq:1b-pred-erro}), 
    $\\
    \left\| \sum_{k\in S^\star_\mathcal{I} \cap S^c} \P_{X_S^\perp} \beta_k^\star X_k  \right\|_2^2 
    \geq \sum_{k\in S^\star_\mathcal{I} \cap S^c}  |\beta_k^\star|^2 - \left( \sum_{j\neq k \in S^\star_\mathcal{I} \cap S^c } |\beta_j^\star\beta_k^\star| + \| \beta^\star_{S^\star_\mathcal{I} \cap S^c} \|_1^2 \right) \mathcal{O}( (s_0 + K) \sqrt{\log p / n})$.

    \vspace{0.1in}
    For the third term in (\ref{eq:1b-pred-erro}), $\left\| \P_{X_S^\perp} \epsilon \right\|_2^2 \geq \|\epsilon\|_2^2 - \|\epsilon\|_2^2 \cdot \mathcal{O} ( (s_0+K) \sqrt{\log p / n} ) $.

    \vspace{0.1in}
    For the fourth term in (\ref{eq:1b-pred-erro}), 
    $\\
    \left( \sum_{k\in S^\star_\mathcal{B} \cap S^c} \beta^\star_k X_k^\top \right) \P_{X_S^\perp}  \left( \sum_{j\in S^\star_\mathcal{I} \cap S^c} \beta^\star_j X_j  \right) \leq \sum_{k\in S^\star_\mathcal{B} \cap S^c} \sum_{j\in S^\star_\mathcal{I} \cap S^c} |\beta^\star_k \beta^\star_j| \cdot \mathcal{O} ( (s_0 + K) \sqrt{\log p / n} ) $.

    \vspace{0.1in}
    For the fifth term in (\ref{eq:1b-pred-erro}), 
    $
    \epsilon^\top \P_{X_S^\perp} \sum_{k\in S^\star \cap S^c} \beta^\star_k X_k  
    \leq \|\epsilon\|_2 \cdot \|\beta^\star_{S^\star \cap S^c} \|_1 \cdot \mathcal{O}( (s_0 + B) \sqrt{\log p / n} ). 
    $
    
    \vspace{0.1in}
    Therefore, 
    \begin{align*}
    \begin{split}
        \mathrm{pe}(S) 
    &\geq L_3(r_1, r_2)
    :=  \frac{ \binom{r_1}{2} \min_{j\neq k \in S^\star_\mathcal{B} } |\beta^\star_j - \beta^\star_k|^2 + r_1 (\eta_1^2 + r_3) \min_{j\in S^\star_\mathcal{B} } |\beta^\star_j|^2 }{r_1 + r_3 + \eta_1^2}  
    + r_2 \min_{j\in S^\star_\mathcal{I}  } |\beta^\star_j |^2 + \|\epsilon\|_2^2 \\
    &\quad + \mathcal{O}\left( (s_0 +K) \sqrt{\log p / n} \right) \left( \|\beta^\star\|_1 + \|\epsilon\|_2 \right)^2.
    \end{split}
    \end{align*}

    \vspace{0.1in}
    \underline{Step 2: consider the scenario when $S$ misses signals in $\mathcal{I}$ but $\mathcal{B}$}, which means that either $S_\mathcal{B} = S^\star_\mathcal{B}$ or $|S_\mathcal{B}| = K$, but $r_2 = |S^\star_\mathcal{I} \cap S^c| \geq 1$. Similarly as in Step 2 of Part (1b), the prediction error is given by:
    $$
    \mathrm{pe}(S) \geq \widetilde{L}_3(r_2):= r_2 \min_{k\in S^\star_\mathcal{I} }  |\beta_k^\star|^2 + \|\epsilon\|_2^2 - \mathcal{O}\left( s_0 \sqrt{\log p / n} \right) (\|\beta^\star\|_1 + \|\epsilon\|_2)^2.
    $$

    \vspace{0.1in}
    \underline{Step 3: consider the scenario when $S$ misses no signals}, which means that either $S = S^\star$, or $|S_\mathcal{B} |= K$ and $S_\mathcal{I} = S^\star_\mathcal{I}$. Similarly as in Step 3 of Part (1b), the prediction error is given by:
    $$
    \mathrm{pe}(S) \leq U_3 := \|\beta^\star \|_2^2 \frac{\eta_1^2}{1 + \eta_1^2} + \|\epsilon\|_2^2 + \mathcal{O}\left(s_0 \sqrt{\log p / n}  \right) \left( \|\beta^\star\|^2_2 + 
    \|\beta^\star\|_1 \|\epsilon\|_2  \right).
    $$

    \vspace{0.1in}
    \underline{Step 4: draw the conclusion}. When 
    \begin{align*}
    \begin{split}
       & \min\left\{
    \frac{\min_{j\neq k \in S^\star_{\mathcal{B}}} |\beta^\star_j - \beta^\star_k|^2 + 2 \eta_1^2 \min_{j\in S^\star_\mathcal{B} } |\beta^\star_j|^2 }{K + \eta_1^2},\ 
    \frac{  (\eta_1^2 + 1) \min_{j\in S^\star_\mathcal{B} } |\beta^\star_j|^2}{K + \eta_1^2},\ 
    \min_{j\in S^\star_\mathcal{I}} |\beta_j^\star|^2 \right\} \\
    > {} & \frac{\eta_1^2}{1 + \eta_1^2} \|\beta^\star\|_2^2 + \mathcal{O}\left( (s_0 + K) \sqrt{\log p / n} \right) (\|\beta^\star\|_1 + \|\epsilon\|_2)^2, 
    \end{split}
    \end{align*}
    we have $U_3 < L_3 (r_1, r_2)$ and $U_3 < \widetilde{L}_2(\widetilde{r}_2)$ for any $r_1 \geq 1$, $r_2 \geq 0$ and $ \widetilde{r}_2\geq 1$, meaning that no signal features should be missed.

   \vspace{0.1in}
    \noindent\textbf{Part (1) Summary.} 
    Combining all results in Part (1a)--(1c) above, the condition under which that no signals are missed is given by:
    \begin{align*}
    \begin{split}
     &   \min\left\{ \frac{\min_{j\neq k \in S^\star_\mathcal{B}} |\beta^\star_j - \beta^\star_k|^2   + 2 \eta_1^2 \min_{j\in S^\star_\mathcal{B} } |\beta^\star_j|^2}
    { 1 + K + (1-K) \mathds{1}_{\{N^\star_\mathcal{B} = \varnothing \}}  + \eta_1^2},\ 
    \frac{  (\eta_1^2 + 1) \min_{j\in S^\star_\mathcal{B} } |\beta^\star_j|^2 \mathds{1}_{\{N^\star_\mathcal{B} \neq \varnothing\}} }{K + \eta_1^2},\ 
    \min_{j\in S^\star_\mathcal{I}} |\beta_j^\star|^2 \right\} \\
    >{}& \frac{\eta_1^2}{1 + \eta_1^2} \|\beta^\star\|_2^2 + \mathcal{O}\left( (s_0 + K)\sqrt{\log p / n} \right) (\|\beta^\star\|_1 + \|\epsilon\|_2)^2,
    \end{split}
    \end{align*}
    and additionally
    $$
    \min_{j\in S^\star_\mathcal{B}} |\beta^\star_j|^2 > \mathcal{O}\left( s_0 \sqrt{\log p / n} \right)\left( \|\beta^\star\|_1 + \|\epsilon\|_2 \right)^2 . $$

    \vspace{0.1in}
    \noindent\textbf{Part (2).} In the case of $s_0 = s^\star$, the upper bound for $\mathrm{pe}(S)$, when $S$ misses no signals and $N^\star_\mathcal{B} \neq \varnothing$, can be shaper than that in Part (1). In this case, we must have $S = S^\star$, and hence
    $
    \mathrm{pe}(S) = \|\P_{X_S^\perp } \epsilon \|_2^2 \leq \left\| \epsilon \right\|_2^2.
    $
    As a result, 
    the condition under which no signals are missed is given by:
    \begin{align*}
    \begin{split}
     &   \min\left\{ \frac{\min_{j\neq k \in S^\star_\mathcal{B}} |\beta^\star_j - \beta^\star_k|^2 \mathds{1}_{\{N^\star_\mathcal{B} = \varnothing \}}  + (1 + \mathds{1}_{\{N^\star_\mathcal{B} = \varnothing \}})  \eta_1^2 \min_{j\in S^\star_\mathcal{B} } |\beta^\star_j|^2}
    {1 + K + (1-K) \mathds{1}_{\{N^\star_\mathcal{B} = \varnothing \}}  + \eta_1^2},\ 
    \min_{j\in S^\star_\mathcal{I}} |\beta_j^\star|^2 \right\} \\
    >{}& \frac{\eta_1^2}{1 + \eta_1^2} \|\beta^\star\|_2^2 \cdot \mathds{1}_{\{ N^\star_\mathcal{B} = \varnothing \}} + \mathcal{O}\left( (s_0 + K)\sqrt{\log p / n} \right) (\|\beta^\star\|_1 + \|\epsilon\|_2)^2,
    \end{split}
    \end{align*}
    and additionally
    $$
    \min_{j\in S^\star_\mathcal{B}} |\beta^\star_j|^2 > \mathcal{O}\left( s_0 \sqrt{\log p / n} \right)\left( \|\beta^\star\|_1 + \|\epsilon\|_2 \right)^2 . $$

    Finally, we conclude that, when $s_0 = s^\star$, the only solutions that miss no signals are those $\supp(\widehat\beta) \in \mathcal{S}$.

\end{proof}

\vspace{0.2in}
\subsubsection{Proof of Theorem \ref{prop:block-consist-subsample}}
\label{sec:sup-proof-blockfsss}

\begin{proof}
By condition (i), we know that for all $\ell\in\{1, \dots, B\}$, we have $k\in \widehat{S}^{(\ell)}$ for all $k\in S^\star$, except possibly for one index $k_0 \in S^\star$, where $ X_{k_0} \notin \spa(X_{\widehat{S}^{(\ell)}} ) $ while $ X_{k_0} + \delta \in \spa(X_{\widehat{S}^{(\ell)} })$. 
Note that $\|\P_{X_k} - \P_{X_k + \delta} \|_2 \leq \sqrt{1 - \tr(\P_{X_k}   \P_{X_k + \delta} )} \leq \sqrt{\frac{\eta_1^2}{1 + \eta_1^2}} + \mathcal{O}(\sqrt{\log p / n})$.

Now consider a set of new basis vectors $\{Z_j^{(\ell)}\}_{j\in \widehat{S}^{(\ell)}}$ for $\spa(X_{\widehat{S}^{(\ell)} } )$ and construct the ``quasi-child'' $W^{(\ell)}$ as follows: 
(1) If $c \notin \widehat{S}^{(\ell)}$, then take $Z_j^{(\ell)} = X_j$ for all $j \in \widehat{S}^{(\ell)}$, and let $W^{(\ell)} = 0$; 
(2) If $c \in \widehat{S}^{(\ell)}$ and $S^\star \subseteq \widehat{S}^{(\ell)}$, then take $Z_j^{(\ell)} = X_j$ for all $j\in \widehat{S}^{(\ell)} \setminus\{c\}$, take $Z_c^{(\ell)} = X_c - \sum_{j\in \widehat{S}^{(\ell)} \cap \mathcal{B}} X_j + \delta $, and let $W^{(\ell)} = Z_c^{(\ell)}$; 
(3) If $c \in \widehat{S}^{(\ell)}$ and there exists one $k_0 \in S^\star$ such that $S^\star\setminus\{k_0\} \subseteq  \widehat{S}^{(\ell)}$ while $X_{k_0} + \delta \in \spa(X_{\widehat{S}^{(\ell)} })$, then take $Z_j^{(\ell)} = X_j$ for all $j\in \widehat{S}^{(\ell)} \setminus\{c\}$, $Z_c^{(\ell)} = X_{k_0} + \delta$, and let $W^{(\ell)} = X_{k_0}$.


\vspace{0.1in}
\underline{Step 1: find a lower bound for $\sigma_{|S|} (\P_{X_S} \P_{\rm avg}\P_{X_S} ) $ for any $S \in \mathcal{S}$}. 

Let $\Delta_S^{(\ell)} = \P_{X_{\widehat{S}^{(\ell)}}} - \sum_{k\in \widehat{S}^{(\ell)} \setminus\{c\}} \P_{Z_k} - \P_{W^{(\ell)}}$. We then have $\|\Delta_S^{(\ell)} \|_2 \leq \left\| \P_{X_{\widehat{S}^{(\ell)}}} - \sum_{j\in \widehat{S}^{(\ell)}} \P_{Z_j^{(\ell)}} \right\|_2 + \left\| \P_{Z^{(\ell)}_c} - \P_{W^{(\ell)}} \right\|_2 
\leq \sqrt{\frac{\eta_1^2}{1 + \eta_1^2}} + \mathcal{O} ((s_0 + B) \sqrt{\log p / n} ) $.

Moreover, let $\Delta = \P_{X_S} - \sum_{k\in S^\star } \P_{X_k} $, and $\|\Delta\|_2 \leq \| \P_{X_S} - \sum_{j\in S}\P_{X_j} \|_2 + \|\sum_{j\in S} \P_{X_j} - \sum_{k\in S^\star } \P_{X_k} \|_2 \leq  \sqrt{\frac{ \eta_1^2 }{1 + \eta_1^2}} + \mathcal{O}(s^\star \sqrt{\log p / n}) $. Therefore,
\begin{align*}
\begin{split}
   &  \P_{X_S} \P_{\rm avg} \P_{X_S}  \\
={}& \left( \sum_{k\in S^\star } \P_{X_k} + \Delta \right) 
\frac{1}{B}\sum_{\ell=1}^B \left( \sum_{k\in S^\star} \P_{X_k} + \sum_{k\in \widehat{S}^{(\ell)} \setminus\{c\} \setminus S^\star} \P_{Z_k} + \P_{W^{(\ell)}} \mathds{1}_{ W^{(\ell)} = Z^{(\ell)}_c } + \Delta_S^{(\ell)} \right)
\left( \sum_{k\in S^\star } \P_{X_k} + \Delta \right)  \\
={}& \sum_{k\in S^\star } \P_{X_k} + \widetilde{\Delta},
\end{split}
\end{align*}
where $\| \widetilde{\Delta} \|_2 \leq \sum_{i=1}^{12} \mathfrak{T}_i(s_0, n, p) $. Specifically,

\noindent $
\mathfrak{T}_1 = \left\| \sum_{j,k,l\in S^\star \setminus \{j = k= l\} } \P_{X_i} \P_{X_j} \P_{X_l} \right\|_2 \leq (s^\star)^3 \mathcal{O}(\sqrt{\log p / n})
$; 

\noindent $\mathfrak{T}_2 = \left\| \sum_{j,k \in S^\star} \P_{X_j} \frac{1}{B} \sum_{\ell=1}^B \left(  \sum_{k\in \widehat{S}^{(\ell)} \setminus\{c\} \setminus S^\star} \P_{Z_k} + \P_{W^{(\ell)}} \mathds{1}_{ W^{(\ell)} = Z^{(\ell)}_c } \right) \P_{X_k} \right\|_2 \leq  \mathcal{O}((s^\star)^2 (s_0+B) \sqrt{  \log p / n }) $; 

\noindent $\mathfrak{T}_3 = \left\| \sum_{j,k \in \mathrm{RF}(S)} \P_{X_j} \P_{X_k} \Delta \right\|_2 \leq (s^\star)^2 \|\Delta\|_2 \leq (s^\star)^2 \sqrt{\frac{\eta_1^2}{1 + \eta_1^2}} + \mathcal{O}((s^\star)^3 \sqrt{\log p / n} ) $;

\noindent $\mathfrak{T}_4 = \left\| \sum_{j \in S^\star} \P_{X_j} \frac{1}{B} \sum_{\ell=1}^B \left(  \sum_{k\in \widehat{S}^{(\ell)} \setminus\{c\} \setminus S^\star} \P_{Z_k} + \P_{W^{(\ell)}} \mathds{1}_{ W^{(\ell)} = Z^{(\ell)}_c } \right) \Delta \right\|_2 
\leq s^\star s_0 \sqrt{\frac{\eta_1^2}{1 + \eta_1^2}} + \mathcal{O}((s^\star)^2 s_0 \sqrt{\log p / n} )$;

\noindent $\mathfrak{T}_5 = \left\| \sum_{j,k\in S^\star} \P_{X_j} \left(\frac{1}{B} \sum_{\ell=1}^B \Delta_S^{(\ell)} \right) \P_{X_k} \right\|_2 \leq (s^\star)^2 \frac{1}{B}\sum_{\ell=1}^B \|\Delta_S^{(\ell)}\|_2 \leq (s^\star)^2 \sqrt{\frac{\eta_1^2}{1 + \eta_1^2}} + \mathcal{O}( (s^\star)^2 (s_0+B) \sqrt{\log p / n} )$;

\noindent $\mathfrak{T}_6 = \left\| \sum_{j\in S^\star} \P_{X_j} \left(\frac{1}{B} \sum_{\ell=1}^B \Delta_S^{(\ell)} \right) \Delta \right\|_2 \leq s^\star \left(\frac{1}{B} \sum_{\ell=1}^B \| \Delta_S^{(\ell)} \|_2 \right) \|\Delta\|_2 \leq s^\star \frac{\eta_1^2}{1 + \eta_1^2} + \mathcal{O}( (s^\star)^2 (s_0+B) \sqrt{\log p / n} )$;

\noindent $\mathfrak{T}_7 = \left\| \sum_{j,k\in S^\star} \Delta \P_{X_j} \P_{X_k} \right\|_2 \leq (s^\star)^2 \|\Delta\|_2 \leq (s^\star)^2 \sqrt{\frac{\eta_1^2}{1 + \eta_1^2}} +  \mathcal{O}((s^\star)^3 \sqrt{\log p / n} ) $;

\noindent $\mathfrak{T}_8 = \left\|\sum_{j\in S^\star}  \Delta \P_{X_j} \Delta \right\|_2 \leq s^\star \|\Delta\|_2^2 \leq s^\star \frac{\eta_1^2}{1 + \eta_1^2} + \mathcal{O}( (s^\star)^3 \sqrt{\log p / n} )$;

\noindent $\mathfrak{T}_9 = \left\| \sum_{j\in S^\star}  \Delta \left(  \sum_{k\in \widehat{S}^{(\ell)} \setminus\{c\} \setminus S^\star} \P_{Z_k} + \P_{W^{(\ell)}} \mathds{1}_{ W^{(\ell)} = Z^{(\ell)}_c } \right)  \P_{X_j} \right\|_2  \leq s^\star s_0  \|\Delta\|_2 
\leq s_0 s^\star \sqrt{\frac{\eta_1^2}{1 + \eta_1^2}} + \mathcal{O}( (s^\star)^2 s_0 \sqrt{\log p / n} )$;

\noindent $\mathfrak{T}_{10} = \left\|  \Delta \left(  \sum_{k\in \widehat{S}^{(\ell)} \setminus\{c\} \setminus S^\star} \P_{Z_k} + \P_{W^{(\ell)}} \mathds{1}_{ W^{(\ell)} = Z^{(\ell)}_c } \right)  \Delta \right\|_2 \leq s_0 \|\Delta\|_2^2 
\leq s_0  \frac{\eta_1^2}{1 + \eta_1 ^2} + \mathcal{O} ((s^\star)^2 s_0 \sqrt{\log p / n } )
$;

\noindent $\mathfrak{T}_{11} = \left\| \sum_{j\in S^\star} \Delta \left(\frac{1}{B}\sum_{\ell=1}^B \Delta_S^{(\ell)} \right)  \P_{X_k} \right\|_2 \leq s^\star \|\Delta\|_2 \left(\frac{1}{B} \sum_{\ell=1}^B \|\Delta_S^{(\ell)}\|_2 \right)
\leq s^\star \frac{\eta_1^2}{1 + \eta_1^2} + \mathcal{O}( (s^\star)^2 (s_0+B) \sqrt{\log p / n} )
$;

\noindent $\mathfrak{T}_{12} = \left\| \Delta \left(\frac{1}{B} \sum_{\ell=1}^B \Delta_S^{(\ell)} \right) \Delta \right\|_2 \leq \|\Delta\|_2^2 \left(\frac{1}{B} \sum_{\ell=1}^B \|\Delta_S^{(\ell)}\|_2 \right) 
\leq  (\frac{\eta_1^2}{1 + \eta_1^2})^{3/2} + \mathcal{O}((s^\star)^2 (s_0+B) \sqrt{\log p / n} ).
$

\noindent  Combining them together, we have $\|\widetilde{\Delta}\|_2 
\leq (3(s^\star)^2 + 2s^\star s_0)  \sqrt{\frac{\eta_1^2}{1 + \eta_1^2}} + (3 s^\star + s_0 ) \frac{\eta_1^2}{1 + \eta_1^2} +  ( \frac{\eta_1^2}{1 + \eta_1^2} )^{3/2} + \mathcal{O}((s^\star)^2 (s_0+B) \sqrt{\log p / n}) $. 

As a result, by Weyl's inequality, 
$
\sigma_{|S|} (\P_{X_S} \P_{\rm avg} \P_{X_S}) \geq \sigma_{|S|}( \sum_{k\in S^\star } \P_{X_k} ) - \|\widetilde{\Delta}\|_2 \geq 1 - (3(s^\star)^2 + 2s^\star s_0)  \sqrt{\frac{\eta_1^2}{1 + \eta_1^2}} - (3 s^\star + s_0 ) \frac{\eta_1^2}{1 + \eta_1^2} -  ( \frac{\eta_1^2}{1 + \eta_1^2} )^{3/2} + \mathcal{O}((s^\star)^2 (s_0+B) \sqrt{\log p / n})  =: L_4,
$
where the last inequality follows from Lemma \ref{lem:sigmamin-quasi-proj} and condition (ii).

\vspace{0.1in}
\underline{Step 2: find an upper bound for the stability of newly added direction} when adding an ``undesired feature'' to $S \in \mathcal{S}$. The aim is to ensure that (1) when $N^\star_\mathcal{B} = \varnothing$, no extra features in $\mathcal{B} \setminus S$ and $N^\star_\mathcal{I}$ can be added to $S$ with high probability; (2) when $N^\star_\mathcal{B} \neq \varnothing$, no extra features in $N^\star_\mathcal{B} \cup \{c\}$ and $N^\star_\mathcal{I}$ can be added to $S$ with high probability. 

Define $r_j :=X_j - \P_{X_S} X_j$. By Lemma \ref{lem:U-set-property2}, we know that any ``undesired'' features $X_j$, there must exist $u \in \mathcal{U}$ such that $r_j = u - \P_{X_S} u$ or $-r_j = u - \P_{X_S} u$. Now for any ``undesired'' feature $X_j$ and any subsample $\ell\in\{1, \dots, B\}$, consider
\begin{align}\label{eq:consist-stop-early}
    \tr(\P_{X_{\widehat{S}^{(\ell)}}} \P_{r_j}) 
= \frac{ r_j^\top \P_{X_{\widehat{S}^{(\ell)}}} r_j }{r_j^\top r_j}
= \frac{ (u - \P_{X_S} u)^\top \P_{X_{\widehat{S}^{(\ell)}}} (u - \P_{X_S} u) }
{ (u - \P_{X_S} u)^\top (u - \P_{X_S} u) }.
\end{align}
Define the space $\mathcal{F}^{(\ell)}_S$ as follows: (1) if $S^\star \subseteq \widehat{S}^{(\ell)}$, then let $\mathcal{F}^{(\ell)}_S = \spa(\{X_j^{(\ell)}\}_{j\in S^\star})$; (2) If there exists one $k_0\in S^\star$ such that $S^\star \setminus \{k_0\} \subseteq \widehat{S}^{(\ell)}$ while $X_{k_0} + \delta \in \spa(X_{\widehat{S}^{(\ell)}})$, then let $\mathcal{F}^{(\ell)}_S = \spa(\{ X_j^{(\ell)} \}_{j\in S^\star \setminus\{k_0\} } \cup \{ X_{k_0} + \delta \}) $. 

Let $\nabla = \P_{{\mathcal{F}_S^{(\ell)}} } - \P_{X_S}$, then 
$
\|\nabla\|_F 
\leq \|\P_{X_S} - \sum_{j\in S\setminus \{c\} } \P_{X_j} - \P_{X_{j_0+\delta} } \mathds{1}_{\{ X_{j_0} + \delta \in \spa(X_S) \}} \|_F \\
+  \|\P_{\mathcal{F}_S^{(\ell)}} - \sum_{j\in S^\star\setminus \{k_0\} } \P_{X_j} 
-  \P_{X_{k_0 + \delta}} \mathds{1}_{\{ X_{k_0} + \delta \in \spa(X_S) \}} \|_F \\
+ \| \sum_{j\in S\setminus \{c\} } \P_{X_j} + \P_{X_{j_0+\delta} } \mathds{1}_{\{ X_{j_0} + \delta \in \spa(X_S) \}} - \sum_{j\in S^\star\setminus \{k_0\} } \P_{X_j} 
-  \P_{X_{k_0 + \delta}} \mathds{1}_{\{ X_{k_0} + \delta \in \spa(X_S) \}}  \|_F \\
\leq \|\P_{X_{j_0 + \delta}} - \P_{X_{j_0}}\|_F 
+ \|\P_{X_{k_0 + \delta}} - \P_{X_{k_0}}\|_F + \mathcal{O}(s^\star \sqrt{\log p / n} ) \\
\leq 2\sqrt{\frac{\eta_1^2}{1 + \eta_1^2}} + \mathcal{O}(s^\star \sqrt{\log p / n} ) .
$

For the numerator of (\ref{eq:consist-stop-early}), 
\begin{align*}
\begin{split}
    & (u - \P_{X_S} u)^\top \P_{X_{\widehat{S}^{(\ell)}}} (u - \P_{X_S} u) 
    = u^\top \P_{X_{\widehat{S}^{(\ell)}}} u - 2 u^\top \P_{X_S} \P_{X_{\widehat{S}^{(\ell)}}} u + u^\top \P_{X_S} \P_{X_{\widehat{S}^{(\ell)}}} \P_{X_S} u  \\
={}&   u^\top \P_{X_{\widehat{S}^{(\ell)}}} u - 2 u^\top (\P_{ \mathcal{F}_S^{(\ell)}} - \nabla ) \P_{X_{\widehat{S}^{(\ell)}}} u + u^\top (\P_{ \mathcal{F}_S^{(\ell)} } - \nabla ) \P_{X_{\widehat{S}^{(\ell)}}} (\P_{ \mathcal{F}_S^{(\ell)} } - \nabla ) u  \\
={}& u^\top \P_{X_{\widehat{S}^{(\ell)}}} u - u^\top \P_{\mathcal{F}_S^{(\ell)}} u + 2 u^\top \nabla \P_{X_{\widehat{S}^{(\ell)}} } u - 2 u^\top \nabla \P_{ \mathcal{F}_S^{(\ell)} } u + u^\top \nabla \P_{X_{\widehat{S}^{(\ell)}} } \nabla u  \\
={}& u^\top \P_{X_{\widehat{S}^{(\ell)}}} u - u^\top \P_{X_S} u - u^\top\nabla u + 2 u^\top \nabla \P_{X_{\widehat{S}^{(\ell)}} } u - 2 u^\top \nabla \P_{ \mathcal{F}_S^{(\ell)} } u + u^\top \nabla \P_{X_{\widehat{S}^{(\ell)}} } \nabla u \\
\leq{}& u^\top \P_{X_{\widehat{S}^{(\ell)}} } u - u^\top \P_{X_S} u  + 5\|\nabla\|_F \|u\|_2^2 + \|\nabla\|^2_F \|u\|_2^2.
\end{split}
\end{align*}
Therefore,
\begin{align*}
\begin{split}
    \tr(\P_{X_{\widehat{S}^{(\ell)}}} \P_{r_j}) 
&\leq \frac{ \tr(\P_u \P_{X_{\widehat{S}^{(\ell)}}} ) - \tr( \P_u \P_{X_S} )  }{ 1 - \tr( \P_u \P_{X_S} ) } + \frac{ 5\|\nabla\|_F + \|\nabla\|_F^2 }{ 1 - \tr( \P_u \P_{X_S} )  } \\
&\leq \tr(\P_u \P_{X_{\widehat{S}^{(\ell)}}} )  + 10 \|\nabla\|_F + 2 \|\nabla\|_F^2 \\
&\leq \tr(\P_u \P_{X_{\widehat{S}^{(\ell)}}} ) + 20 \sqrt{\frac{\eta_1^2}{1 + \eta_1^2}} +   \frac{ 8\eta_1^2  }{ 1 + \eta_1^2 } + \mathcal{O} \left( s^\star \sqrt{\log p / n} \right),
\end{split}
\end{align*}
where the second inequality follows from the fact that $\tr(\P_u \P_{X_{\widehat{S}^{(\ell)}}} ) \leq 1$ and $\max_{u\in\mathcal{U}}\tr( \P_u \P_{X_S} ) \leq \frac{1}{2}$ by Lemma \ref{lem:U-set-property2} and condition (ii). 

Finally, for the given $\alpha_0\in (0, \frac{1}{2})$, we have
\begin{align*}
\begin{split}
    &\mathbb{P}\left[\frac{1}{B} \sum_{\ell=1}^B \tr(\P_u \P_{X_{\widehat{S}^{(\ell)}}} )  \geq 1-\alpha_0 \right] 
    \leq \mathbb{P}\left[\frac{1}{B/2} \sum_{\ell=1}^{B/2} \prod_{i\in\{0,1\} } \tr(\P_u \P_{X_{\widehat{S}^{(2\ell-i)}}} )  \geq 1-2\alpha_0 \right] \\
    \leq{}& \frac{1}{1 - 2\alpha_0} \frac{1}{B/2} \sum_{\ell=1}^{B/2} \mathbb{E}\left[ \tr(\P_u \P_{X_{\widehat{S}^{(2\ell))}}}) \right] \mathbb{E}\left[ \tr(\P_u \P_{X_{\widehat{S}^{(2\ell-1)}}}) \right] 
    \leq \frac{\gamma^2}{1 - 2\alpha_0},
\end{split}
\end{align*}
and hence $\mathbb{P}\left[ \max_{u \in \mathcal{U}} \frac{1}{B} \sum_{\ell=1}^B \tr(\P_u \P_{X_{\widehat{S}^{(\ell)}}} )  \geq 1-\alpha_0 \right] \leq |\mathcal{U}| \gamma^2 / (1 - 2\alpha_0) = \left( p - s^\star \right) \gamma^2 / (1 - 2\alpha_0)$.
In summary, with probability at least $1 - \left( p - s^\star \right) \gamma^2 / (1 - 2\alpha_0)$, for any ``undesired'' feature $X_j$, we have
\begin{align*}
\begin{split}
    \tr(\P_{r_j} \P_{\rm avg}) 
    &\leq \frac{1}{B} \sum_{\ell=1}^B \max_{u\in \mathcal{U}} \tr(\P_{u} \P_{X_{\widehat{S}^{(\ell)}}} ) + 20 \sqrt{\frac{\eta_1^2}{1 + \eta_1^2}} +   \frac{ 8\eta_1^2  }{ 1 + \eta_1^2 } + \mathcal{O} \left( s^\star \sqrt{\log p / n} \right)  \\
    &\leq 1 - \alpha_0 + 20 \sqrt{\frac{\eta_1^2}{1 + \eta_1^2}} +   \frac{ 8\eta_1^2  }{ 1 + \eta_1^2 } + \mathcal{O} \left( s^\star \sqrt{\log p / n} \right) =:U_4.
\end{split}
\end{align*}

\vspace{0.1in}
\underline{Step 3: Draw the conclusion}. 
When $U_4 < \alpha < L_4$, the only selection sets can be returned by our algorithm are those in $\{S: S  \in \mathcal{S}\}$. Note that $ \sigma_{|S|}(\P_{X_S} \P_{\rm avg} \P_{X_S} ) \leq \inf_{Z\in \spa(X_S)}\tr(\P_{Z} \P_{\rm avg})$, implying that for any $S \in\mathcal{S}$, both the ``new direction condition'' and the ``$\sigma_{\min}$-condition" will be satisfied and hence only $S$ rather than its subsets will be returned.

\end{proof}

\vspace{0.2in}
\subsubsection{Proof of Corollary \ref{cor:pfer2}}
\label{sec:sup-proof-blockcor}

\begin{proof}
    Let $S$ be any FSSS selection set. When $S \in \mathcal{S}$, we have
    \begin{align*}
    \begin{split}
        \tr(\P_{X_S} \P_{ X_{ S^{\star}}^\perp })
    &= s^\star - \tr(\P_{X_S} \P_{ X_{S^\star} }) \\
    &= s^\star - \sum_{j\in S\setminus\{c\} } \tr(\P_{X_j} \P_{X_S^\star} ) - \tr(\P_{X_{j_0} + \delta} \P_{ X_{S^\star} }) \mathds{1}_{\{ X_{j_0} + \delta \in \spa(X_S) \}} - \tr(\Delta \P_{X_{S^\star} }) \\
    &\leq 1 - \tr(\P_{X_{j_0} + \delta} \P_{X_{S^\star}} ) + s^\star \|\Delta\|_2 
    \leq  \frac{\eta_1^2}{1 + \eta_1^2} + \mathcal{O}\left((s^\star)^2 \sqrt{\log p / n} \right).
    \end{split}
    \end{align*}
    Otherwise, $\tr(\P_{X_S} \P_{ X_{ S^{\star\perp} } }) \leq s_0$. The upper bound for $\mathbb{E} \left[ \tr(\P_{X_S} \P_{ X_{ S^{\star}}^\perp }) \right]$ hence follows.
\end{proof}

\vspace{0.2in}
\subsubsection{Proof of Remark \ref{rem:theory-block}}
\label{sec:sup-proof-blockrem}

\begin{proof}
    \underline{Scenario 1: $N^\star_\mathcal{B} \neq \varnothing$}.
    In the special case, the base procedure $\widehat{S}^{(\ell)}$ would include $S^\star$ and uniformly select $(s_0 - s^\star)$ other directions from $ \{X_c\} \cup \bigcup_{j \in N^\star_\mathcal{B}} \{X_j\}  \cup \bigcup_{j \in N^\star_\mathcal{I}} \{X_j\}$.

    For $u =X_c$, if $c \in \widehat{S}^{(\ell)}$ (w.p. $\frac{s_0 - s^\star}{ p - s^\star }$), we have $\tr(\P_u \P_{X_{\widehat{S}^{(\ell)}}}) = 1$. If $c \notin \widehat{S}^{(\ell)}$, then $\widehat{S}^{(\ell)}$ must select $m$ features in $N_\mathcal{B}^\star$ and $s_0 - s^\star$ features in $N_\mathcal{I}^\star$ for $m\in\{0, \dots, |N^\star_\mathcal{B}|\}$; this happens w.p. $\frac{ \binom{|N^\star_\mathcal{B}|}{m} \binom{1}{0} \binom{p - s^\star - N^\star_\mathcal{B} - 1}{s_0 - s^\star - m} }{ \binom{p - s^\star}{ s_0 - s^\star } }$, and in this case $ \tr(\P_u \P_{X_{\widehat{S}^{(\ell)}}}) = \frac{|S^\star_\mathcal{B}| + m}{K + \eta_1^2}$.

    For any $u = X_j \in \bigcup_{j \in  N^\star_\mathcal{B}} \{X_j\} $, if $j \notin \widehat{S}^{(\ell)}$ and $c \notin \widehat{S}^{(\ell)}$, we have $\tr(\P_u \P_{X_{\widehat{S}^{(\ell)}}}) = 0 $. If $j \notin \widehat{S}^{(\ell)}$ but $c \in \widehat{S}^{(\ell)}$, then $\widehat{S}^{(\ell)}$ must select $m$ features in $N^\star_{\mathcal{B}} \setminus \{j\}$ and $s_0 - s^\star - m$ features in $N^\star_\mathcal{I}$ for $m\in\{ 0,\dots, |N^\star_\mathcal{B}| - 1 \}$; this happens w.p. $\frac{ \binom{|N^\star_\mathcal{B}| - 1}{m} \binom{1}{0} \binom{1}{1} \binom{p-s^\star - N^\star_\mathcal{B} - 1}{ s_0 - s^\star - m }  }{ \binom{p - s^\star}{s_0 - s^\star} }$, and in this case $\tr(\P_u \P_{X_{\widehat{S}^{(\ell)}}}) = \frac{1}{ |N^\star_\mathcal{B}| - m + \eta_1^2 } $.
    
    For any $u = X_j \in \bigcup_{j\in N^\star_\mathcal{I}} \{X_j\} $, if $j\in \widehat{S}^{(\ell)}$ (w.p. $\frac{s_0 - s^\star}{p - s^\star}$), we have $\tr(\P_u \P_{X_{\widehat{S}^{(\ell)}}}) = 1 $; and $\tr(\P_u \P_{X_{\widehat{S}^{(\ell)}}}) = 0$ otherwise. 

    \vspace{0.1in}
    \underline{Scenario 2: $N^\star_\mathcal{B} = \varnothing$}.
    In the special case, the base procedure would include $S^\star$ and uniformly select $(s_0 - s^\star)$ other directions from $\{\delta\} \cup \bigcup_{j\in N^\star_\mathcal{I}} \{X_j\}$. Therefore, for any $u\in\mathcal{U}$, we have $\tr(\P_u \P_{X_{\widehat{S}^{(\ell)}}}) = 1$ w.p. $\frac{s_0 - s^\star}{p - s^\star}$, and $\tr(\P_u \P_{X_{\widehat{S}^{(\ell)}}}) = 0$ otherwise.
    
\end{proof}

\end{document}